\documentclass[11pt,a4paper]{article}
\pdfoutput=1

\usepackage{epsfig}
\usepackage{graphics}
\usepackage{subfigure}

\usepackage{epstopdf}

\usepackage[english]{babel}
\usepackage{mdwlist}
\usepackage{graphicx}

\usepackage{appendix}
\usepackage{dsfont}
\usepackage[usenames, dvipsnames]{xcolor}
\usepackage[bookmarksopen=true,colorlinks=true,linkcolor=Emerald,citecolor=NavyBlue]{hyperref}

\usepackage{tikz}
\usetikzlibrary[positioning]
\usetikzlibrary{patterns}
\usetikzlibrary{decorations.pathreplacing,angles,quotes}
\usepackage{amsmath}
\usepackage{amsbsy}
\usepackage{amsfonts}
\usepackage{amssymb}
\usepackage{amscd}
\usepackage{amsthm}
\usepackage{empheq}
\usepackage{placeins}

\usepackage[section]{algorithm}
\usepackage{algorithmic}

\usepackage{bbm}
\usepackage{bm}
\usepackage{xspace}

\input{src/preamble_shortcut}

\newtheorem{thm}{Theorem}[section]
\newtheorem{defnt}[thm]{Definition}
\newtheorem{prop}[thm]{Proposition}
\newtheorem{corollary}[thm]{Corollary}
\newtheorem{lemme}[thm]{Lemma}
\newtheorem{remark}[thm]{Remark}
\newtheorem{assumption}[thm]{Assumption}

\title{On oracle-type local recovery guarantees in compressed sensing}

\author{Ben Adcock\footnote{Simon Fraser University, Burnaby, BC, Canada. e-mail: ben\_adcock@sfu.ca},  Claire Boyer\footnote{Sorbonne Universit\'{e} \& \replace{}{Ecole normale sup\'{e}rieure, Paris, PSL University,} France. e-mail: claire.boyer@sorbonne-universite.fr}, and Simone Brugiapaglia\footnote{\replace{}{Concordia University, Montreal, QC, Canada, and} Simon Fraser University, Burnaby, BC, Canada. e-mail: \replace{simone\_brugiapaglia@sfu.ca}{simone.brugiapaglia@concordia.ca}}}

\let\Re\relax\DeclareMathOperator{\Re}{Re}

\let\vect\relax\DeclareMathOperator{\vect}{vec}

\usepackage[normalem]{ulem}
\usepackage{cancel}
\newcommand{\replace}[2]{{#2}}
\usepackage{cancel}
\newcommand{\replacemath}[2]{{#2}}

\usepackage[colorinlistoftodos,textwidth=2.3cm]{todonotes}

\begin{document}
\maketitle

\begin{abstract}
We present improved sampling complexity bounds for stable and robust sparse recovery in compressed sensing. Our unified analysis based on  $\ell^1$ minimization encompasses the case where (i) the measurements are block-structured samples in order to reflect the structured acquisition that is often encountered in applications; (ii) \replace{where}{} the signal has an arbitrary structured sparsity, by results depending on its support $S$.  Within this framework and under a random sign assumption, the number of measurements needed by $\ell^1$ minimization can be shown to be of the same order than the one required by an oracle least-squares estimator. Moreover, these bounds can be minimized by adapting the variable density sampling to a given prior on the signal support and to the coherence of the measurements. We illustrate both numerically and analytically that our results can be successfully applied to recover Haar wavelet coefficients that are sparse in levels from random Fourier measurements in dimension one and two, which can be of particular interest in {imaging problems}. Finally, a preliminary numerical investigation shows the potential of this theory for devising adaptive sampling strategies in sparse polynomial approximation. 
\end{abstract}


\section{Introduction}

\subsection{Motivations}

{Standard Compressed Sensing (CS) concerns the recovery of a sparse vector $x$ from linear measurements.  The theory of CS is well-established, with one of its signature results being the existence of suitable decoders (for instance, based on convex optimization) that achieve recovery from near-optimal numbers of suitably-chosen measurements (e.g.\ Gaussian random measurements\replace{}{)} scaling linearly with the sparsity and logarithmically with the ambient dimension.}

{However, many applications of CS exhibit more structure than sparsity alone.  Hence there is a need to understand its performance for more structured signal models.  With this in mind, the purpose of this paper is \replace{to derive oracle-type lower bounds}{to provide sufficient conditions} on the number of measurements required \replace{in CS}{} in order to ensure recovery of a structured sparse signal $x\in \Cbb^n$ with random signs, via quadratically-constrained Basis Pursuit:
$$
\min_{z \in \Cbb^n} \|z\|_1 \qquad \text{such that} \qquad \| y-Az \|_2 \leq \eta.
$$
Our analysis has two key features.  First, we consider a general type of measurement matrix $A$, based on a block structure.  Although the sampling is random, the block structure is a way to get a theoretical setting closer to realistic sampling than standard CS measurement matrix constructions: often in applications, only certain sampling patterns are allowed and those ``joint" measurements are modeled here by the block structure of $A$. 
}
For instance, in Magnetic Resonance Imaging (MRI), a common practice is to acquire samples along radial lines or straight lines, as illustrated in Figure \ref{fig:illus_MRI_lines}. Moreover, in many other applications, such as ultrasound imaging or interferometry, the acquisition is often constrained to specific sampling patterns \cite{liu2014optimum,quinsac2010compressed}.
\begin{figure}
\begin{center}
\begin{tabular}{cc}
\includegraphics[height=6cm]{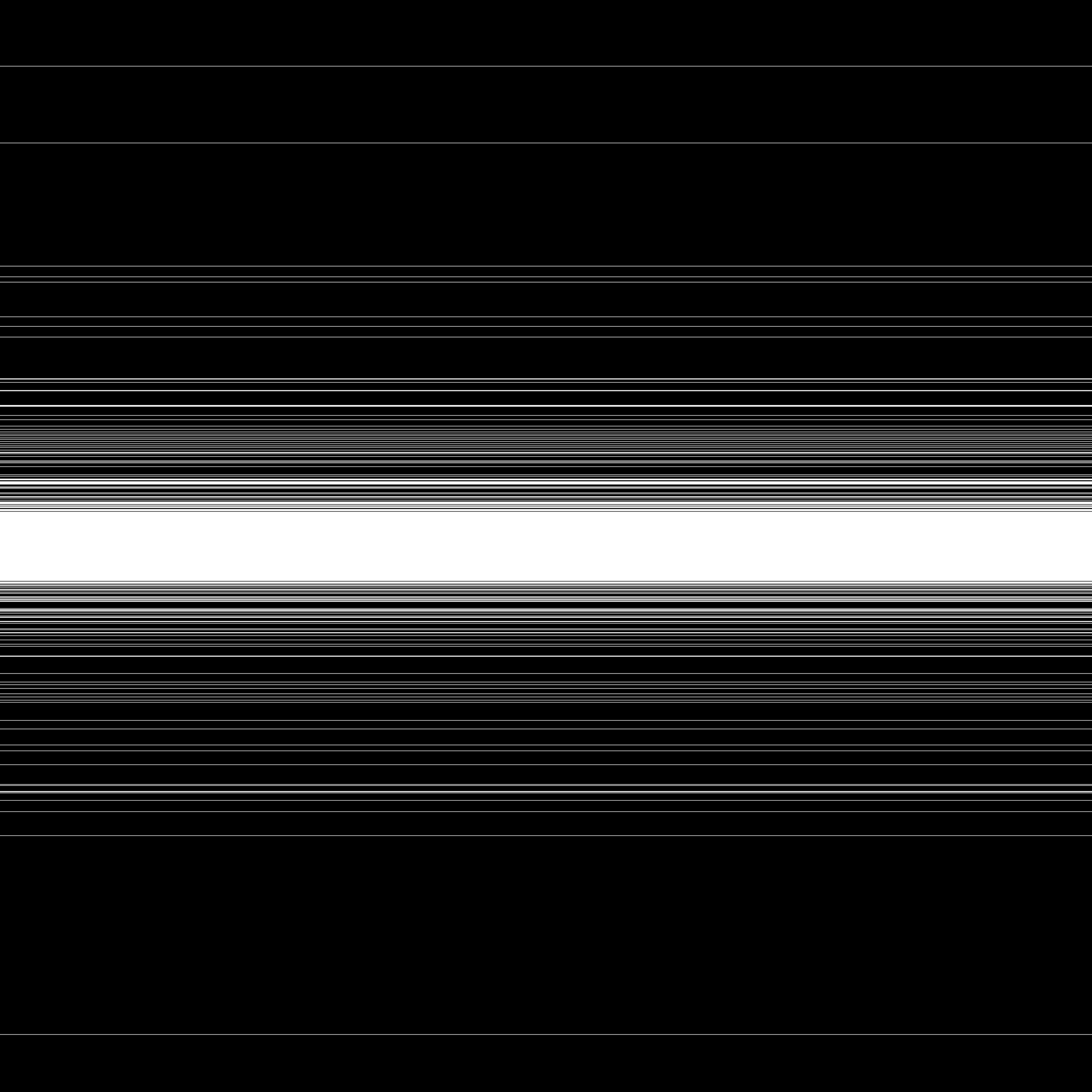} &
\includegraphics[height=6cm]{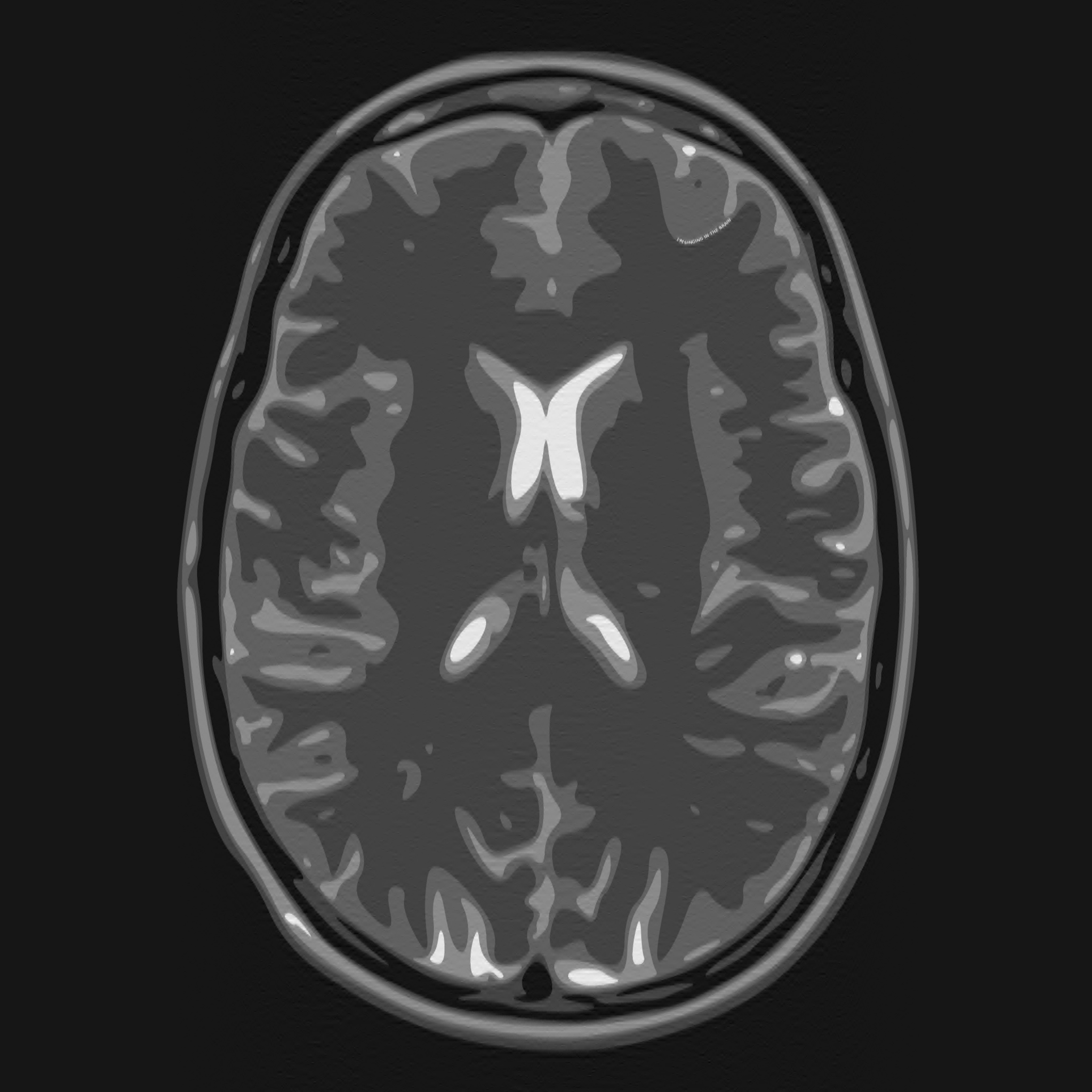} \\
{(a)} & {(b)} 
\end{tabular}
\end{center}
\caption{\label{fig:illus_MRI_lines} Example of structured sampling in 2D MRI. In (a), samples are taken along horizontal lines in the acquisition domain. In (b), the reconstruction of a synthetic image has been generated using \eqref{pb:qBP} and samples located as in (a). }
\end{figure}
\replace{Moreover}{In addition}, the block structure of $A$ considered in this paper can also model the case of multiple-sensor data acquisition, considered in applications such as parallel MRI \cite{wang2000description}. 
It should be noted that the block-structured acquisition encompasses standard CS strategies in which isolated measurements are sampled from a given isometry (see for instance \cite[Chapter 12]{foucart2013mathematical}).

Second, our analysis provides recovery guarantees that are local \replace{}{with respect} to the support of the sparse vector $x$; in particular, they impose no signal model (e.g.\ the sparse model).  As a result, they allow one to consider certain structured sparsity models for $x$.  In
 recent years, structured sparsity has been proved to be a more appropriate prior than standard sparsity when dealing with real-world {problems, such as imaging} \cite{adcock2017breaking}. Structured sparsity is a key feature to consider in order to {devise optimal sampling strategies for image reconstruction} \cite{adcock2014quest}. {Moreover}, it is able to leverage and theoretically legitimize block-structured sensing \cite{boyer2017compressed,chun2017compressed}.

A number of recent works have considered local recovery guarantees in CS \cite{boyer2017compressed,chun2017compressed}.  However, unlike in the classical setting, it was unknown whether the corresponding measurement conditions were optimal.  Note that in the case of Gaussian measurements, in \cite{amelunxen2014living}, a theoretical phase transition has been identified around the statistical dimension, denoted here by $\delta$, of the descent cone associated to the $\ell^1$ norm: if the number of measurements $m$ is such that $m\geq \delta +\sqrt{n}$ then basis pursuit succeeds in recovering an $s$-sparse vector $x$, if $m\leq \delta - \sqrt{n}$ then basis pursuit fails to recover $x$, both cases with high probability. However, as soon as the sampling is more structured, meaning that the sampling is based either on isolated measurements from a structured isometry (e.g.\ such as the Fourier transform), or on blocks of structured measurements, there is no optimality guarantee on the required number of measurements to ensure recovery.  {To address this, in this paper our measurement conditions for the (quadratically-constrained) Basis Pursuit decoder are compared with those of the oracle least-squares estimator.  The latter relies on \textit{a priori} knowledge of the support of $x$ (hence the term `oracle'), something which is of course not available to the former.}

\subsection{Comparison with existing results}


In the seminal paper \cite{candes2007sparsity}, under a random sign assumption on the signal to reconstruct, the authors proposed to  draw uniformly at random \replace{the}{} rows from an isometry $A_0=(\replacemath{a}{d}_k^*)_{1\leq k \leq n}$, leading to stable reconstruction with probability at least $1-\varepsilon$ with the following required number of measurements:
\begin{equation}
\label{eq:coherence}
m \gtrsim s \cdot n \, \max_{k} \|\replacemath{a}{d}_k \|_\infty^2 \cdot \ln \left( n / \varepsilon \right).
\end{equation}
This result can be of interest when considering totally incoherent transforms such as the Fourier matrix for which $n \, \max_{k} \|\replacemath{a}{d}_k \|_\infty^2 = O(1)$ .
However, this is not relevant anymore in the case of coherent transforms, such as the Fourier-Haar transform used to model MRI acquisition, where $n \, \max_{k} \|\replacemath{a}{d}_k \|_\infty^2 = O(n)$.

In this paper, our results include the previous ones, but are also extended to: 
(i) the case of variable density sampling; 
(ii) stability robustness results when measurements are corrupted with bounded noise;
(iii) structured measurements using blocks of measurements;
(iv) optimization of the sampling density with respect to prior information on the signal support, such as structured sparsity.

In \cite{boyer2017compressed,chun2017compressed}, acquisition of i.i.d.\ blocks of measurements was introduced to model structured acquisition closer to applications constraints. In this paper, we allow non-identically distributed blocks of measurements, requiring only independence between the blocks. Note that in \cite{chun2017compressed}, this extended setting was also considered to handle parallel acquisition. However, the main results in \cite{boyer2017compressed,chun2017compressed} were involving the maximum between two quantities in the required number of measurements. The latter prevents theoretically and numerically from any minimization of the obtained bound with respect to the way of drawing measurements. In this paper, with an additional assumption on the signal sign randomness, we derive a bound on the number of measurements depending only on one quantity making its minimization easier and analytically explicit. New results on optimal sampling strategies are presented showing that they should not only depend on the coherence of the sensing matrix (as in standard CS, see for instance \cite{chauffert2013variable} and \cite[Chapter 12]{foucart2013mathematical}) but also on the prior on the signal structure.

\subsection{Contributions}

In this paper, we extend the setting of block-structured sensing introduced in \cite{bigot2014analysis} and also considered in \cite{boyer2017compressed,chun2017compressed}  to the case where the blocks are not identically distributed and in which some of them can be deterministically chosen.
In this setting, we derive stable and robust recovery guarantees for CS, while considering a random sign assumption on the signal of interest. The recovery guarantees are \emph{local} (or equivalently, \emph{nonuniform} or \emph{signal-based}), in the sense that they ensure the recovery with high probability for a fixed signal.

Let $S$ be the set of $s$ largest absolute entries of a given signal and let $F$ be the probability model used to draw random (possibly block-structured) measurements from a finite-dimensional isometry $A_0$. Our recovery guarantees are based on two notions of local coherence, denoted as $\Lambda(S,F)$ and $\Theta(S,F)$ such that $\Lambda(S,F) \leq \Theta(S,F)$ and on a global coherence measure $\Gamma(F)$ (these three quantities are formally introduced in Definition~\ref{def:quantities}).
Using this notation, an \emph{oracle-type inequality} is a \replace{lower bound}{condition} on the number of measurements of the form
\begin{equation}
\label{eq:oracle}
m \gtrsim \Lambda(S,F) \cdot \ln(|S|/\varepsilon),
\end{equation}
which guarantees robust recovery from noisy measurements via the oracle least-squares estimator. First, we prove stable and robust recovery for CS with probability at least $1-\varepsilon$ under the  condition
\begin{equation}
\label{eq:condTheta}
m \gtrsim \Theta(S,F) \cdot \ln^2(n/\varepsilon).
\end{equation}
Moreover, we show that the oracle-type inequality
\begin{equation}
\label{eq:condLambda}
m \gtrsim \Lambda(S,F) \cdot \ln(n/\varepsilon),
\end{equation}
is sufficient to guarantee stable and robust recovery under the extra assumption $\Lambda(S,F) \gtrsim \Gamma(F) \cdot \ln(n/\varepsilon)$, which is verified in cases of practical interest. 
In the case of standard sparsity and uniform random sampling as in \cite{candes2007sparsity}, conditions \eqref{eq:condTheta} and \eqref{eq:condLambda} are implied by \eqref{eq:coherence} and, more in general, they {refine the measurement conditions given} in \cite{krahmer2014stable,chauffert2014variable,puy2011variable} for the case of variable density sampling and standard sparsity.

A main consequence of Conditions \eqref{eq:condTheta} and \eqref{eq:condLambda} is to give an optimal sampling strategy in order to minimize the required number of measurements while taking into account prior information on the support $S$, such as structured sparsity. We derive a closed form expression \replace{of}{for} the drawing probability, i.e.\ how to choose the measurements, in Section~\ref{sec:optimal_drawing}. This is a substantial contribution since in previous CS approaches, only variable density sampling based on the sensing transform coherence was performed. Here, the optimal strategy is shown to be not only dependent on the sampling coherence but also on the signal structure. 

For illustrative purposes, let us briefly describe the implications of our contribution to the case of random isolated Fourier measurements for the recovery a one-dimensional signal that is assumed to be sparse in levels with respect to the Haar transform (this case study is discussed in detail in Section~\ref{sec:applications_mri_iso}). Let us assume the signal to have sparsities in levels $(s_j)_{0\leq j \leq J}$, i.e., $s_j = |\Omega_j \cap S|$, where $(\Omega_j)_{0 \leq j \leq J}$ are Haar wavelet subbands. Moreover, let us divide the space of Fourier frequencies $\{-n/2+1, \ldots, n/2\}$ into subbands $(W_j)_{0 \leq j \leq J}$ and denote as $j(k)$ the frequency band associated with the $k$-th frequency. Then, minimizing the quantity $\Lambda(S,F)$ in \eqref{eq:condLambda} leads to drawing the $k$-th frequency with probability
\begin{equation}
\label{eq:proba_sparsity_level}
\pi_k = \frac{2^{-j(k)} \sum_{j=1}^J 2^{-|j-j(k)|}s_j}{\sum_{\ell=1}^n 2^{-j(\ell)}\sum_{j'=1}^J 2^{-|j(\ell)-j'|} s_{j'}}.
\end{equation}
The resulting \replace{lower bound}{sufficient condition} on the number of measurements is
\begin{equation}
\label{eq:m_condition_FH}
m \gtrsim \left( \sum_{j=1}^J s_j + \sum_{j'=1\atop j'\neq j}^J 2^{-|j-j'|} s_{j'} \right) \cdot \ln(s) \ln(n/\varepsilon).
\end{equation}
This improves the previous conditions from \cite{adcock2017breaking,boyer2017compressed} by decreasing the interference between the sparsities in levels by a square-root factor. Namely, we have $2^{-|j-j'|}$ instead of $2^{-|j-j'|/2}$ in \eqref{eq:m_condition_FH}. Analogous considerations hold for the case of the two-dimensional Fourier-Haar transform with Fourier measurements structured along vertical or horizontal lines, discussed in Section~\ref{sec:application_MRI}.

Finally, the explicit dependence of the optimal sampling measure on the signal support allows for adaptive sampling strategies, which can be particularly relevant when the signal structure or the sampling coherence are not known \emph{a priori}. Preliminary numerical experiments in Section~\ref{sec:adaptive} illustrate how our analysis can be applied to derive adaptive sampling strategies for sparse polynomial approximation.

\subsection{Organization of the paper}

In Section \ref{sec:setting}, we introduce the setting, the considered recovery algorithm, and the sampling strategy adopted to be compatible with physically-constrained acquisition. The definitions of crucial quantities involved in the theoretical analysis are also given therein.
In Section \ref{sec:main_results}, the main results are presented and compared with the oracle case. 
A main consequence of this work is discussed in Section \ref{sec:optimal_drawing}, in which an optimal sampling strategy is proposed.
Section \ref{sec:appli} gathers illustrations of the results, in particular in the case of Fourier-wavelets transforms encountered in MRI applications and in function interpolation, where an adaptive sampling strategy is proposed.
The proofs of the main results are organized in Appendices \ref{sec:proof_main}-\ref{proof:appli}.


\section{Setting }
\label{sec:setting}

In this section, we describe the formal setting of the paper. After introducing some standard notation in Section~\ref{sec:notation}, we discuss the recovery strategies considered in the case of noiseless and noisy measurements in Section~\ref{sec:recovery}. In Section~\ref{subsec:sampling}, we describe the sampling strategies analyzed in this paper; in particular, one can consider block-structured sampling (Section~\ref{subsec:sampling}~\ref{block_finite_setting}) and isolated measurements (Section~\ref{subsec:sampling}~\ref{iso_finite_setting}) from a finite-dimensional isometry. Finally, in Section~\ref{sec:assumptions} the main technical ingredients of the proposed theoretical analysis are introduced. We also discuss the random sign assumption on the signal to recover and define three key quantities (denoted by $\Theta$, $\Lambda$, and $\Gamma$) that will play a major role in the analysis carried out in Section~\ref{sec:main_results}. 

\subsection{Notation}
\label{sec:notation}

In this paper, $n$ denotes the dimension of the signal to reconstruct. 
The notation $S \subseteq  \{1, \hdots, n\}$ refers to the support of the signal to reconstruct and define $S^c := \{1, \hdots, n\} \setminus S$.
The vectors \replace{$\left( e_i \right)_{1\leq i \leq p}$}{$\left( e_i \right)_{1\leq i \leq d}$} denote the vectors of the canonical basis of $\Rbb^d$, where $d$ will be equal to $n$ or $\sqrt{n}$, depending on the context. 
For every $x \in \Cbb^d$, we define $x_S$ to be the restriction of $x$ to the components in $S$.  Notice that $x_S$ may be a $|S|$-dimensional or a $d$-dimensional vector, depending on the context; in the second case, the entries of $x_S$ in $S^c$ are set to be zero. Moreover, we set $P_S$ to be {the matrix defined by the linear projection $P_S x = x_S$ for every $x \in \Cbb^d$}. Again, $P_S$ can be a $d \times d$ or a $|S|\times d$ matrix, depending on whether $x_S$ is considered as a $|S|$-dimensional or as a $d$-dimensional vector. Observe that when $x$ is supported on $S$, then also $x = P_S^* x_S$ holds. We will use the shorthand notation $M_S$ to denote the matrix $M P_S^{*}$.
Similarly, if $M_k$ denotes a matrix indexed by $k$, then $M_{k,S}=M_k P_S^{*}$.
For any matrix $M$, for any $1\leq p,q \leq \infty$, the operator norm $\| M \|_{p\rightarrow q}$ is defined as
$$ \| M \|_{p\rightarrow q} = \sup_{\|v\|_p \leq 1} \| Mv \|_q,
$$
with $\|\cdot\|_p$ and $\|\cdot\|_q$ denoting the standard $\ell_p$ and $\ell_q$ norms. Note that for a matrix $M\in \Rbb^{n\times n}$,
$$ \| M \|_{\infty \rightarrow \infty} = \max_{1\leq i \leq n} \| e_i^* M \|_1.
$$
The function $\sgn : \Rbb^n \rightarrow \Rbb^n$ is defined by
\[ \left( \sgn ( x) \right)_i = \left\lbrace
\begin{array}{cc}
1 & \text{if}  \quad x_i >0 \\
-1 & \text{if} \quad  x_i < 0 \\
0 & \text{if} \quad x_i=0,
\end{array}
\right.
\]
and \replace{$\Id_n$}{$\Id$ (or $\Id_n$)} will denote the \replace{}{(}$n$-dimensional\replace{}{)} identity matrix. 
We denote by $\Rc (M)$ the range of the matrix $M$ and by $M^\dagger$ the left pseudo-inverse of  $M$, meaning that if $M$ has full column rank, $M^\dagger = (M^*M)^{-1} M^*$.

\subsection{Recovery techniques}
\label{sec:recovery}

Let  $x\in \Cbb^n$ be supported on $S \subset \{1, \hdots , n \}$. In the case of noiseless measurements, the collected data $y$ can be written as follows
\begin{align}
\label{eq:noiseless_meas}
y = Ax,
\end{align}
where $A$ is the sampling matrix.
In order to recover $x$, we consider $\ell^1$-minimization with equality constraint, also known as the Basis Pursuit (BP) optimization program:
\begin{align}
\tag{BP}
\label{pb:BP}
\min_{z \in \Cbb^n} \|z\|_1 & \quad \text{such that} \quad y=Az.
\end{align}

In the case where observations are corrupted with noise, we will assume the noise to be bounded. In particular, we will assume that there exists $ \eta >0$, supposed to be known, such that
\begin{align}
\label{eq:noisy_meas}
y = Ax + \epsilon, \qquad \| \epsilon \|_2 \leq \eta.
\end{align}
In order to estimate $x$, we then consider the $\ell^1$-minimization problem with inequality constraint, also called quadratically-constrained Basis Pursuit (qBP):
\begin{align}
\tag{qBP}
\label{pb:qBP}
\min_{z \in \Cbb^n} \|z\|_1 & \quad \text{such that} \quad \|y-Az\|_2 \leq \eta.
\end{align}

\subsection{Sampling strategy}
\label{subsec:sampling}

\paragraph{General setting} 
Given some distributions $(F_\ell)_{1 \leq \ell \leq m}$ respectively on sets of $p_\ell\times n$ matrices, with $p_\ell\geq 1$ for $\ell=1,\hdots , m$, the sampling strategy consists in drawing $m$ independent matrices $B_{1}, \hdots , B_m$ where $B_\ell \sim F_\ell$ for $\ell=1,\hdots,m$ and forming the sensing matrix as follows:
\begin{align}
\label{eq:sensing_matrix}
A= \frac{1}{\sqrt{m}}
  \begin{pmatrix}
B_{1} \\ \vdots \\ B_{m}  
   \end{pmatrix}, \qquad \text{with} \qquad B_\ell \sim F_\ell, \quad \text{for} \, \ell=1,\hdots, m.
\end{align}
 We assume the sampling to be isotropic, in the sense that 
$$
\Ebb (A^* A ) = \Ebb \left( \frac{1}{m} \sum_{\ell=1}^m B_\ell^*B_\ell \right) = \Id.
$$

This abstract setting can be specialized to the case where we are given an orthogonal matrix $A_{0} \in \Cbb^{n \times n}$ with rows $(\replacemath{a_i^*}{d_i^*})_{1\leq i\leq n} $ representing the set of possible linear measurements imposed by a specific sensor device. In particular, this framework encompasses the two following cases.

\begin{enumerate}
\item \textbf{Block-structured sampling from a finite-dimensional isometry. \label{block_finite_setting}}

Let $\left( \Ic_k \right)_{1\leq k \leq M}$ denote a partition of the set $\{1, \hdots , n \}$, i.e.\ a family of disjoint subsets 
$$
\Ic_k \subset \{1, \hdots , n \}
\quad \text{s.t.}\quad 
\bigsqcup_{k=1}^M \Ic_k = \{1,\hdots , n\}.
$$
The rows $(d_i^*)_{1\leq i \leq n} \in \Cbb^n$ of $A_0$ are then partitioned accordingly into a block dictionary $\left( D_k \right)_{1 \leq k \leq M}$, such that
$$ 
D_k = \left( d_i^* \right)_{i\in \Ic_k} \in \Cbb^{|\Ic_k| \times n}.
$$
Define the random blocks $B_1,\hdots , B_m$ to be i.i.d.\ copies of a random block $B$ such that
$$
\Pbb\left( B = D_k /\sqrt{\pi_k} \right) = \pi_k, \qquad \text{for} \quad k=1,\hdots , M,
$$
\replace{}{where $(\pi_k)_{1\leq k \leq M}$ is a discrete probability distribution on $\{1,\ldots,M\}$.} \replace{}{Note that in this case, all the distributions $(F_\ell)$'s are the same one, characterizing the law of the random block $B$ described right above.} The sensing matrix $A$ is then constructed by randomly drawing blocks as follows:
\begin{equation}
\label{eq:bos_sensing_matrix}
A =  \frac{1}{\sqrt{m}} \left( B_\ell \right)_{1\leq \ell \leq m}.
\end{equation}
Moreover, thanks to the renormalization, the random sensing matrix $ A $ satisfies
\begin{align}
\label{isotropyCondition}
\Ebb (A^* A) = \frac{1}{m} \sum_{\ell=1}^m \Ebb(B_\ell^* B_\ell) = \Ebb(B^* B) = \sum_{k=1}^M D_k^*D_k = A_0^* A_0 =\Id,
\end{align}
since $A_0$ is orthogonal and $\left( D_k\right)_{1\leq k \leq M}$ is a partition of the rows of $A_0$.

\item \textbf{Isolated measurements from a finite-dimensional isometry\label{iso_finite_setting}} (standard CS).

This is a particular case of the setting described in (i), which is standard in CS: each block corresponds to a row \replace{in}{of the} matrix $A_0 = (d_1 | d_2 | \hdots | d_n)^*$. Therefore, the sensing matrix is constructed by stacking random vectors drawn from the set of row vectors \replace{$\{ a_1^*, \hdots a_n^*\}$}{$\{ d_1^*, \hdots d_n^*\}$} and can be written as follows:
\begin{equation}
\label{eq:bos_iso_sensing_matrix}
\replacemath{{A = \frac{1}{\sqrt{m}} \left( \frac{1}{\sqrt{\pi_{J_\ell}}} a^*_{J_\ell} \right)_{1\leq \ell \leq m}, }}{A = \frac{1}{\sqrt{m}} \left(  a^*_{\ell} \right)_{1\leq \ell \leq m}, }
\end{equation}
\replace{
where $\left( J_\ell \right)_{1\leq \ell \leq m}$ are i.i.d.\ copies of a random variable $J$ such that}{where the random vectors $(a_{\ell})_{1\leq \ell \leq m}$ are 
i.i.d. copies of a random vector
$a$ such that}
\[
\replacemath{\Pbb(J = j) = \pi_{j},}{\Pbb (a = d_j/\sqrt{\pi_j} ) = \pi_{j},}
\]
for all $1 \leq j \leq n$.
\replace{}{Here again all the $(F_{\ell})$'s consists in the same distribution, designating the law of the random vector $a$.}
The isotropy condition, i.e. $\Ebb(A^*A) = \Ebb \left(\replacemath{\frac{a_J a_J^*}{\pi_J}}{\frac{a_\ell a_\ell^*}{\pi_\ell}} \right) = \Id$, is also satisfied.
\end{enumerate}

\begin{remark}\normalfont
The setting can be modified in order to encompass partial deterministic sampling. Consider some distributions \replace{$(F_\ell)_{m_0+1 \leq k \leq m}$}{$(F_\ell)_{m_0+1 \leq \ell \leq m}$} respectively on sets of $p_\ell\times n$ matrices, with $p_\ell\geq 1$ for $\ell=m_0+1,\hdots , m$, the sampling strategy consists in drawing $m - m_0 $ independent matrices $B_{m_0+1}, \hdots , B_m$ where $B_\ell \sim F_\ell$ for $\ell=m_0+1,\hdots,m$ and forming the sensing matrix as follows:
\begin{align}
\label{eq:sensing_matrix_ext}
A= \frac{1}{\sqrt{m- m_0}}
  \begin{pmatrix}
B_{1} \\ \vdots \\ B_{m}  
   \end{pmatrix}, \qquad \text{with} \qquad B_\ell \sim F_\ell, \quad \text{for} \, \ell=m_0+1,\hdots, m,
\end{align}
while $B_1,\hdots , B_{m_0}$ can be \emph{deterministically} chosen.
 We still need the sampling to be isotropic, in the sense that 
\begin{align}
\Ebb (A^* A ) &= \sum_{\ell=1}^{m_0} B_\ell^* B_\ell + \Ebb \left( \frac{1}{m-m_0} \sum_{\ell=m_0+1}^m B_\ell^*B_\ell \right) = \Id.
\end{align}
This possible extension is motivated as follows. In applications where the sensing matrix is randomly extracted from a  Fourier/wavelets transform, multi-level sampling strategies have been proved to be highly effective (see \cite{adcock2017breaking}). In particular, one may want to partition the Fourier space into levels and then \emph{saturate} (i.e., fully sample) some of them \cite{LiAdcockRIP}. Usually, the saturated levels are those corresponding to the lowest frequencies. 
However, saturating some levels using a fully random procedure as in (i) leads to a suboptimal sampling rate, due to the coupon collector effect. Allowing partial deterministic sampling of $m_0$ blocks is a simple way to circumvent this problem.
To avoid heavy notation, we will state the main results and the proofs with $m_0=0$.
\end{remark}

\subsection{Assumptions}
\label{sec:assumptions}

We assume that the signal we aim at reconstructing satisfies a random sign property, defined as follows.
\begin{assumption}
\label{ass:random_sign}
For any vector  $x\in \Rbb^n$ or $x\in \Cbb^n$ supported on $S$, we will say that $x$ satisfies the random sign assumption if $\sgn(x_S)$ is respectively a Rademacher or Steinhaus sequence.
\end{assumption}
The following quantities {are} crucial in the recovery guarantees. 
\begin{defnt}
\label{def:quantities}
Consider a block sampling strategy as previously described in \eqref{eq:sensing_matrix} where $(B_k)_{1 \leq k \leq m}$ are random blocks such that $B_k \sim F_k$. We denote the collection of probability distributions by $F = (F_k)_{\replacemath{k}{1 \leq k \leq m}}$.
Let $S\subset \{1, \hdots, n \}$.
Define the quantities $\Theta(S,F)$, $\Lambda(S,F)$, and $\Gamma(F)$ to be positive real numbers such that
\begin{align}
\label{ineq:Theta}
 \Theta(S,F) &\geq  \| B_\ell^*B_{\ell,S} \|_{\infty \rightarrow \infty}  \quad &\mbox{ a.s. } \quad B_\ell \sim F_\ell, \qquad \ell = (m_0+)1,\hdots , m,
 \\
 \Lambda(S , F) &\geq   \left\| B_{\ell,S}^* B_{\ell,S} \right\|_{2\rightarrow 2} \quad &\mbox{ a.s. } \quad B_\ell \sim F_\ell, \qquad \ell = (m_0+)1,\hdots , m, \\
\Gamma(F) & \geq  \| B_\ell \|_{1\rightarrow 2}^2 = \max_{1\leq i \leq n}\left\| B_\ell e_i \right\|_{2}^2 \quad &\mbox{ a.s. } \quad B_\ell \sim F_\ell, \qquad \ell = (m_0+)1,\hdots , m.
\end{align}
\end{defnt}

{Typically, but not always, $ \Theta(S,F)$, $\Lambda(S,F)$ and $\Gamma(F)$ will be taken \replace{to the}{as} the least-upper bounds.}
Note that since
$  \left\| B_{S}^* B_{S} \right\|_{2\rightarrow 2} \leq  \left\| B_S^* B_{S} \right\|_{\infty \rightarrow \infty} \leq \left\| B^* B_{S} \right\|_{\infty\rightarrow \infty}$
 (due to, e.g.\ \cite[Lemma A.8 and Remark A.10]{foucart2013mathematical}), {if $ \Theta(S,F)$ and $\Lambda(S,F)$ are taken as least-upper bounds, then}
\begin{align}
\label{eq:lambda_theta}
 \Lambda(S,F)\leq  \Theta(S,F).
\end{align}
For the sake of readability, sometimes we will simply use $\Theta$, $\Lambda$, and $\Gamma$ to refer to $\Theta(S,F)$, $\Lambda(S,F)$, and $\Gamma(S,F)$, respectively.

In the case of \replace{}{the} block-structured finite setting, one 
considers a block dictionary $\left( D_k\right)_{1\leq k \leq M}$ as in Section \ref{subsec:sampling}\ref{block_finite_setting}. Given the quantities in Definition \ref{def:quantities}, one can derive the following upper bounds: let $S\subset \{1, \hdots, n \}$ and $\pi$ be a probability distribution on $\{1,\hdots ,M\}$,
\begin{align}
\label{ineq:Theta_bos}
 \Theta(S,\pi) & \geq \max_{1\leq k \leq M}  \frac{1}{\pi_k} \| D_{k}^*D_{k,S}  \|_{\infty \rightarrow \infty} = \max_{1\leq k \leq M} \max_{1\leq i \leq n} \frac{ \| e_i^* D_k^* D_{k,S} \|_1}{\pi_k},
 \\
 \label{Lambda_bos}
 \Lambda(S , \pi) &\geq \max_{1\leq k \leq M}  \frac{1}{\pi_k}  \left\| D_{k,S}^* D_{k,S} \right\|_{2\rightarrow 2}, \\
\label{Gamma_bos}
\Gamma(\pi) & 
\geq  \max_{1\leq k \leq M} \frac{1}{\pi_k}\|D_k\|_{1 \to 2}^2 = \max_{1\leq k \leq M} \max_{1 \leq i \leq n} \frac{1}{\pi_k}  \left\| D_k e_i \right\|_{2}^2.
\end{align}

In the case of isolated measurements drawn from an isometry, one 
considers the rows $\left( \replacemath{a_k}{d_k^*}\right)_{1\leq k \leq n}$ of an orthogonal matrix $A_0$ as in Section \ref{subsec:sampling}\ref{iso_finite_setting}. Given the quantities in Definition \ref{def:quantities}, one can derive the following upper bounds: let $S\subset \{1, \hdots, n \}$ and $\pi$ be a probability distribution on $\{1,\hdots ,n\}$,
\begin{align}
\label{ineq:Theta_iso_bos}
 \Theta(S,\pi) & \geq \max_{1\leq k \leq n}  \frac{1}{\pi_k} \| \replacemath{a}{d}_{k} \|_\infty \| \replacemath{a}{d}_{k,S}  \|_{1},
 \\
 \label{Lambda_iso_bos}
 \Lambda(S , \pi) &\geq \max_{1\leq k \leq n}  \frac{1}{\pi_k}  \left\| \replacemath{a}{d}_{k,S}\right\|_{2}^2, \\
\label{Gamma_iso_bos}
\Gamma(\pi) & \geq  \max_{1\leq k \leq n} \frac{1}{\pi_k}  \left\| \replacemath{a}{d}_k \right\|_{\infty}^2.
\end{align}
In \eqref{Gamma_iso_bos}, one may recognize the standard definition of the global coherence in CS, see for instance \cite{candes2011probabilistic,foucart2013mathematical}.



\section{Main results}
\label{sec:main_results}

In this section, we derive recovery guarantees for  \eqref{pb:BP} and \eqref{pb:qBP} under a random sign assumption on the signal . They \replace{lead to lower bounds}{reveal sufficient conditions} on the required number of measurements, provided that one can evaluate $\Theta$, $\Lambda$, and $\Gamma$ defined in Section~\ref{sec:assumptions}.

\paragraph{Overview of the main results} 
Throughout the section, our benchmark will be an oracle-type inequality discussed in Section~\ref{sec:oracle}, namely 
\begin{equation}
\label{eq:oracle_simple}
m \gtrsim \Lambda(S,F) \cdot \ln(|S|/\varepsilon),
\end{equation}
which is proved to be sufficient for the robust recovery of an $S$-sparse via oracle-least squares with probability at least $1-\varepsilon$ in Proposition~\ref{prop:oracle}. We will refer to \eqref{eq:oracle_simple} as an \emph{oracle-type} inequality. Notice that  the oracle-least squares estimator requires an \emph{a priori} knowledge of $S$ to recover the signal, whereas the \eqref{pb:BP} and \eqref{pb:qBP} programs do not. 

In Sections~\ref{sec:noiseless} and \ref{sec:robustness} we make a first step towards oracle-type inequalities for CS. In particular, we show that 
\begin{equation}
\label{eq:Theta_bound_simple}
m \gtrsim \Theta(S,F) \cdot \ln^2(n/\varepsilon),
\end{equation}
is sufficient for the exact recovery from noiseless measurements (Theorem~\ref{thm:noiseless}) or robust recovery from noisy measurements (Theorem~\ref{thm:noisy}) of a signal supported on $S$ with probability at least $1-\varepsilon$. Note that, besides the additional logarithmic factor, \eqref{eq:Theta_bound_simple} is not necessarily an oracle-type inequality, in view of \eqref{eq:lambda_theta}. 

\replace{We provide oracle-type inequalities for CS in Section~\ref{sec:oracle-type}}{We progressively improve bounds on $m$ for CS recovery via \eqref{pb:BP} and \eqref{pb:qBP}}  in Section~\ref{sec:oracle-type}, presenting three results \replace{in this direction}{towards oracle-type inequalities}. Theorem~\ref{thm:oracle_theta} \replace{provides}{requires} \replace{an inequality for CS}{a bound} of the form 
$$
m \gtrsim \max_{j \in S^c}\Lambda(S \cup \{j\},F) \cdot \ln^2(n/\varepsilon),
$$ 
which is of oracle type up to enlarging the support by one element.
Theorems~\ref{thm:killthetheta} and \ref{thm:killthetalog} achieve oracle-type \replace{inequalities}{requirements on $m$} at the price of an extra assumption involving $\Lambda(S,F)$ and $\Gamma(F)$. They  only differ by a logarithmic factor. In particular, in Theorem~\ref{thm:killthetheta} robust recovery from noisy measurements (or exact recovery from noiseless measurement) is guaranteed with probability $1-\varepsilon$ if
$$
m \gtrsim \Lambda(S,F) \cdot \ln^2(n/\varepsilon),
$$
and provided that $\Lambda(S,F) \gtrsim \Gamma(F)$. Theorem~\ref{thm:killthetalog} achieves the same recovery guarantees if
$$
m \gtrsim \Lambda(S,F) \cdot \ln(n/\varepsilon), 
$$
and under the extra assumption $\Lambda(S,F) \gtrsim \Gamma(F) \ln(n/\varepsilon)$. These extra assumptions do not turn out to be restrictive in practice (see Section~\ref{sec:appli}). \replace{}{Finally, we note that the assumptions of Theorem~\ref{thm:oracle_theta}, \ref{thm:killthetheta}, and \ref{thm:killthetalog} are sufficient to guarantee stable and robust recovery from noisy measurements for both \eqref{pb:qBP} and the oracle least-squares estimator (see Remark~\ref{rem:stability}).}

\subsection{Preliminary: an oracle inequality}
\label{sec:oracle}

In this section, we derive a lower bound on the number of measurements sufficient to obtain robust recovery using an oracle least-squares estimator.

\begin{prop} 
\label{prop:oracle}
Let \replace{$x\in \Cbb^n$}{$x\in \Rbb^n$ or $\Cbb^n$} be a vector  supported on a set $S$ of size $s$ and suppose we are given noisy measurements $y = Ax + \epsilon$, with $\|\epsilon\|_2 \leq \eta$. Then, there exist universal constants $c_0,C_0 >0$ such that, for every $0 < \varepsilon < 1$ and provided
\begin{align}
\label{eq:bound_meas_oracle}
m \geq c_0 \cdot \Lambda(S,F) \cdot \ln \left( \frac{2s}{\varepsilon} \right),
\end{align} 
the following holds with probability at least $1-\varepsilon$:  the matrix $A_S$ has full column rank and the oracle-least squares estimator $x^\star \in\Cbb^n$ of the system $ y = A z$, defined by
\begin{equation}
\label{eq:oracle_ls_estimator}
x^\star_S = (A_S)^\dagger y, \quad x_{S^c}^\star = 0.
\end{equation} 
satisfies the error estimate
\begin{align}
\label{eq:oracle_robustness}
\| x - x^\star \|_2 \leq C_0 \eta.
\end{align} 
Possible values for the constants are $C_0 = \sqrt{2}$ and $c_0=32/3$.
\end{prop}
\begin{proof}
Recalling the definition \eqref{eq:oracle_ls_estimator} of $x^\star$ and the fact that both $x$ and $x^\star$ are supported on $S$, one has
\begin{align*}
\| x^\star - x \|_2   
&=   \| (A_S)^\dagger y - x_{S} \|_2 
= \| (A_S)^\dagger (Ax+\epsilon) - x_{S} \|_2  
= \| (A_S)^\dagger \epsilon + (A_S)^\dagger A_{S^c} x_{S^c} \|_2 \\
& \leq \| A_S^\dagger \|_{2\to 2} \| \epsilon\|_2 + \| (A_S^* A_S)^{-1} \|_{2\to 2} \| A_S^* A_{S^c} \|_{1\to 2} \| x_{S^c} \|_1 
\\
&\leq \frac{1}{\sigma_{\min}(A_S)}  \| \epsilon \|_2 +  \| (A_S^* A_S)^{-1} \|_{2\to 2} \| A_S^* A_{S^c} \|_{1\to 2} \| x_{S^c} \|_1 .
\end{align*}
If $\|A_S^* A_S - P_S \|_{2\to 2} \leq \delta$, then $\sigma_{\min}(A_S) \geq \sqrt{1-\delta}$ and $\| (A_S^* A_S)^{-1} \|_{2\to 2}\leq \frac{1}{1-\delta}$. Using Lemma \ref{lem:localIsometry_ext}, if 
$$
m \geq \frac{1+2 \delta/3}{\delta^2/2} \cdot \Lambda (S,F) \cdot \ln \left( \frac{2s}{\varepsilon} \right),
$$   
then $\| A_S^* A_S -P_S \|_{2\to 2} \leq \delta$. Considering that $\| x_{S^c} \|_1 =0$ and fixing $\delta=1/2$ leads to the desired result with the specified constants. 

\end{proof}

\begin{remark} \normalfont
In the following, we are going to derive ``oracle-type" estimate for signal recovery via \eqref{pb:BP} and \eqref{pb:qBP}. In this paper, ``oracle-type" estimates will refer to the bound on the number of measurements, i.e.\ bounds of the form \eqref{eq:bound_meas_oracle}. They will not concern the robustness bound obtained in \eqref{eq:oracle_robustness}, \replace{that}{which} we will \replace{comment}{discuss} later.
\end{remark}

\replace{}{Proposition~\ref{prop:oracle} implies robust recovery of sparse vectors when measurements are corrupted with bounded noise. In fact, it is also possible to prove the stability of the oracle least-squares estimator with respect to the standard sparsity model by considering a condition on $m$ slightly stronger than \eqref{eq:bound_meas_oracle}. For the sake of readability, in the following results we will bypass this additional technical difficulty by focusing only robust sparse recovery. For a more extended discussion on stability, we refer to Remark~\ref{rem:stability}.}

\subsection{Noiseless recovery}
\label{sec:noiseless}

Our main result for the success of \eqref{pb:BP} with an abstract block-structured framework presented in Section \ref{subsec:sampling} is the following. 
\begin{thm}
\label{thm:noiseless}
Let $x\in \Rbb^n$ or $\Cbb^n$ be a vector supported on $S$, such that $\sgn(x_S)$ forms a Rademacher or Steinhaus sequence. Let $A$ be the random sensing matrix defined in \eqref{eq:sensing_matrix} associated with parameter $\Theta(S,F)$. Suppose we are given the data $y=Ax$. Then, given $0<\varepsilon < 1$ and provided
$$
m \geq c_1 \cdot  \Theta (S,F) \cdot \ln^2\left( \frac{6n}{\varepsilon} \right),
$$ 
for $c_1$ a numerical constant (for instance $c_1=\replacemath{51}{82}$), the vector $x$ is the unique minimizer of the basis pursuit program \eqref{pb:BP} with probability at least $1-\varepsilon$.
\end{thm}

The proof is given in Appendix \ref{proof:noiseless}.

\begin{remark} \normalfont
In standard CS, the number of measurements usually depends on the degree of sparsity $s$ (such that $|S|\leq s$) on the one hand and on the sampling coherence on the other hand. Here, the quantity $\Theta(S,F)$ is encapsulating both information. This way, the use of coherent transforms is not prohibited anymore as soon as the support structure is adapted to it. Note that in the case of isolated measurements 	{described in Section~\ref{subsec:sampling}-(ii)}, the quantity $\Theta(S,F)$ can be bounded from above as follows, leading to the standard CS-type estimate:
$$
\Theta(S,F) \leq s\cdot \sup_{a\sim F} \| a \|_\infty^2.
$$
However, we observe that this upper bound is too crude in general, except under the following assumptions:
\begin{itemize}
\item  there is no structure in the signal sparsity, meaning that the only prior on $S$ is that $|S|=s$;
\item the sensing matrix $A$ is totally incoherent, meaning  that all the entries of $A$ have the same magnitude, $\forall (i,j),  |A_{ij}| = \| A \|_{1\to \infty} =1$ (this for instance the case of subsampled Fourier matrix). 
\end{itemize}
\end{remark}

\begin{remark} \normalfont
The bound on the number of measurements in Theorem \ref{thm:noiseless} depends only on one quantity, namely $\Theta(S,F)$, {which, in turn,} depends on the signal support and on the way of drawing blocks. On the contrary, in \cite{boyer2017compressed,chun2017compressed}, the authors derived similar results with a bound on $m$ depending on the maximum between two quantities (one of them corresponding to $\Theta(S,F)$ in this paper). However, the authors \cite{boyer2017compressed,chun2017compressed} have not considered the random sign assumption. The latter was useful to obtain a closed-form expression for the ``coherence" $\Theta(S,F)$. We will see that this is \replace{a}{} key to design optimal drawing strategies in Section \ref{sec:optimal_drawing}.
\end{remark}


\subsection{Robustness}
\label{sec:robustness}

Theoretical guarantees can \replace{be also}{also be} obtained when the observation vector $y$ is corrupted by noise.
\begin{thm}
\label{thm:noisy}
Let \replace{$x\in \Cbb^n$}{$x\in \Rbb^n$ or $\Cbb^n$} be a vector supported on $S$, such that $\sgn(x_S)$ forms a Rademacher or Steinhaus sequence and $|S|=s$. Let $A$ be the random sensing matrix defined in \eqref{eq:sensing_matrix} with parameter $\Theta(S,F)$. Suppose that the data $y$ is given such that $y=Ax+\epsilon$ with $\|\epsilon\|_2\leq \eta$ for some given $\eta>0$.  Then, given $0<\varepsilon < 1$ and provided
$$
m \geq c_1 \cdot  \Theta (S,F) \cdot \ln^2\left( \frac{6n}{\varepsilon} \right),
$$ 
for $c_1$ a numerical constant (for instance $c_1=\replacemath{780}{889}$), a minimizer $x^\sharp$ of \eqref{pb:qBP} satisfies
\begin{align}
\label{eq:noisy_rec}
\|x-x^\sharp\|_2 \leq  (C_1 + C_2\sqrt{s})\eta,
\end{align}
{with probability at least $1-\varepsilon$, where} $C_1,C_2$ are numerical constants.
\end{thm}

The proof of Theorem~\ref{thm:noisy} is given in Appendix \ref{proof:noisy}.

\begin{remark}
\label{rmrk:finer_error_bound}
\normalfont
The proof of Theorem~\ref{thm:noisy} reveals a more precise recovery guarantee, which \replace{takes}{} accounts for the stability with respect to the standard sparsity model. Namely, it is possible to generalize the theorem to the case \replace{$x \in \Cbb^n$}{$x \in \Rbb^n$ or $\Cbb^n$} where $S$ is \replace{the}{a} set of \replace{}{indices corresponding to the} $s$ large\replace{}{st} absolute entries of $x$ \replace{}{(or, more in general, any subset of $s$ indices)}. Under the same conditions as in Theorem~\ref{thm:noisy}, then, with high probability, one has
\begin{align}
\label{eq:noisy_stable_rec}
\|x-x^\sharp\|_2 \leq (C_1 + C_2\sqrt{s})\eta +  C_3 \|x-x_S\|_1.
\end{align}
This matches the bound derived in the standard setting of CS, see for instance \cite[Theorem 12.22]{foucart2013mathematical}, which may seem a bit disappointing. However, to our knowledge, this is the first result of stability and robustness derived in the case of the random sign assumption. 
\end{remark}

\noindent\textbf{How much are we off the oracle estimate of Proposition \ref{prop:oracle}?} Regarding the  number of measurements {$m$}, the \replace{bounds}{conditions} required by Theorems \ref{thm:noiseless} and \ref{thm:noisy} depend on $\Theta (S, F)$, instead of $\Lambda (S, F)$ as in Proposition \ref{prop:oracle}. Therefore, by \eqref{eq:lambda_theta}, we only obtained an upper bound to the oracle estimate.

Regarding the error bound \eqref{eq:noisy_rec}, one can notice that the noise level $\eta$ is amplified by an additional factor $\sqrt{s}$, when compared to \eqref{eq:oracle_robustness}. Nonuniform approaches {are known to} suffer from this extra $\sqrt{s}${-factor} in the robustness bound, see for instance  \cite[Theorem 12.22]{foucart2013mathematical}, although this may be an artifact of the proof strategies employed so far. To the best our knowledge, the only better bound of the form $c \cdot \eta$ for $c$ some constant has been obtained in \cite[Proposition 2.6]{tropp2015convex}. However, such a bound can be obtained so far only for Gaussian measurements since they are based on the control of the minimum conic singular value of the sensing matrix. \replace{To control}{Controlling} the statistical dimension of \replace{}{the} descent cone for structured measurements still remains an open question. That is why in the sequel we will leave aside the question of improving the robustness bound of a $\sqrt{s}$-factor and we will focus on deriving oracle-type bounds in terms of number of measurements.


\subsection{Getting oracle-type bounds for the required number of measurements}
\label{sec:oracle-type}


We start by presenting a first oracle-type estimate where the quantity $\Theta(S,F)$ is replaced with $\Lambda(S \cup \{j\},F)$, for some $j \in S^c$. As discussed in Proposition~\ref{prop:oracle}, if the support $S$ is known, the least amount of measurements needed  for robust recovery should be of the order of $\Lambda(S,F)$. Indeed,  $\Lambda(S,F)$ controls the condition number of the sensing matrix $A$ restricted to the space of vectors supported on $S$. Theorem \ref{thm:oracle_theta} almost invoke this condition number by slightly enlarging $S$ by one off-support component. The proof of this result is given in Appendix~\ref{proof:oracle_theta}. 

\begin{thm}
\label{thm:oracle_theta}
Let \replace{$x\in \Cbb^n$}{$x\in \Rbb^n$ or $\Cbb^n$} be a vector supported on $S$, such that $\sgn(x_S)$ forms a Rademacher or Steinhaus sequence. Let $A$ be the random sensing matrix defined in \eqref{eq:sensing_matrix} with parameter $\Lambda(S,F)$ {and let $y = Ax + \epsilon$, with $\|\epsilon\|_2 \leq \eta$}. {Then, there exist constants $c_1,C_1,C_2  > 0$ such that the following holds. For every $0<\varepsilon < 1$, if} 
$$
m \geq c_1  \max_{j \in S^c} \Lambda(S\cup \{j\},F) \ln \left(\frac{6 (s+1)(n-s)}{\varepsilon}\right)\ln\left(\frac{6n}{\varepsilon}\right),
$$ 
then, with probability at least $1-\varepsilon$, a minimizer $x^\sharp$ of \eqref{pb:qBP} satisfies
\begin{align*}
\|x-x^\sharp\|_2 \leq  (C_1 + C_2\sqrt{s})\eta.
\end{align*}
In particular, in the noiseless case (i.e., $\eta = 0$), $x$ is exactly recovered via \eqref{pb:BP} with probability at least $1-\varepsilon$ and with constant $c_1 = \replacemath{36}{19}$.
\end{thm}

In the following, we propose two oracle-type inequalities that improve Theorems \ref{thm:noiseless} and \ref{thm:noisy}, in which the bound of the required number of measurements actually depends on $\Lambda(S,F)$ instead of $\Theta(S,F)$, at the price of an extra assumption involving $\Lambda(S,F)$ and $\Gamma(F)$.

\begin{thm} 
\label{thm:killthetheta} 
Let $x\in \Rbb^n$ or $\Cbb^n$ be a vector supported on $S$ with $|S| = s \leq n/2$, such that $\sgn(x_S)$ forms a Rademacher or Steinhaus sequence. Let $A$ be the random sensing matrix defined in \eqref{eq:sensing_matrix}  associated with parameters $\Lambda(S,F)$ and $\Gamma(S,F)$. Suppose we are given the data $y=Ax + \epsilon$ such that $\|\epsilon\|_2 \leq \eta$. Then, there exist constants $c_1,c_2,C_1,C_2 > 0$ such that the following holds. For every $0 < \varepsilon < 1$, if
\begin{align}
\label{cond:killthetheta}
\Lambda(S,F) \geq c_1 \cdot \Gamma(F),
\end{align}
and if
$$
m \geq c_2 \cdot  \Lambda (S,F) \cdot \ln^2\left( \frac{3n}{\varepsilon} \right),
$$ 
then, with probability at least $1-\varepsilon$, 
a minimizer $x^\sharp$ of \eqref{pb:qBP} satisfies
\begin{align*}
\|x-x^\sharp\|_2 \leq  (C_1 + C_2\sqrt{s})\eta.
\end{align*}
\end{thm}

Making an assumption on $\Lambda(S,F)$ stronger than \eqref{cond:killthetheta} allows to \replace{}{``}kill\replace{}{''} an extra log factor in the required number of measurements. This is the purpose of the following theorem.

\begin{thm} 
\label{thm:killthetalog} 
Let $x\in \Rbb^n$ or $\Cbb^n$ be a vector supported on $S$ with $|S| = s \leq n/2$, such that $\sgn(x_S)$ forms a Rademacher or Steinhaus sequence. Let $A$ be the random sensing matrix defined in \eqref{eq:sensing_matrix}  associated with parameters $\Lambda(S,F)$ and $\Theta(S,F)$. Suppose we are given the data $y=Ax + \epsilon$, with $\|\epsilon\|_2 \leq \eta$. Then, there exist constants $c_1,c_2,C_1,C_2 > 0$ such that the following holds. For every $0 < \varepsilon < 1$, if
\begin{align}
\label{cond:killthethetalog}
\Lambda(S,F) \geq c_1 \cdot \Gamma(F) \cdot \ln(3 n/\varepsilon),
\end{align}
and if
$$
m \geq c_2 \cdot  \Lambda (S,F) \cdot \ln\left( \frac{3n}{\varepsilon} \right),
$$ 
then, with probability at least $1-\varepsilon$,
a minimizer $x^\sharp$ of \eqref{pb:qBP} satisfies
\begin{align*}
\|x-x^\sharp\|_2 \leq  (C_1 + C_2\sqrt{s})\eta.
\end{align*}
In particular, in the noiseless case (i.e., $\eta = 0$) the signal $x$ is exactly recovered via \eqref{pb:BP} with probability at least $1-\varepsilon$ and with constants $c_1 = 50$ and $c_2 = 100$.
\end{thm}

The proof of Theorem~\ref{thm:killthetalog} is given in Appendix \ref{proof:noiseless_log}. Moreover, the proof of Theorem~\ref{thm:killthetheta} easily follows from that of Theorem~\ref{thm:killthetalog}.  Although the bound on the number of measurements in Theorem \ref{thm:killthetalog} only depends on $\Lambda(S,F)$ that encapsulates support information and sampling coherence related to the support $S$, the global coherence is somewhat restricted by the extra condition \eqref{cond:killthethetalog}.  
\replace{In the same spirit of Remark~\ref{rmrk:finer_error_bound}, let us also point out that the proof reveals the more precise bound \eqref{eq:noisy_stable_rec} for the reconstruction error.}{}
Finally, we note in passing that the numerical values proposed for the constants $c_1$ and $c_2$ in the statement of Theorem~\ref{thm:killthetalog} could be further optimized.


We will see that Assumptions \eqref{cond:killthetheta} or \eqref{cond:killthethetalog} can be satisfied in practice in some applications in Section \ref{sec:appli}. However, in the case of isolated measurements, by recalling \eqref{Lambda_iso_bos} and \eqref{Gamma_iso_bos} condition \eqref{cond:killthethetalog} can be rewritten as follows:
$$
\sup_{a\sim F} \|a_S\|_2^2 \geq c_1 \cdot \sup_{a\sim F} \|a\|_\infty^2 \ln(3 n/\varepsilon),
$$
in which $\sup_{a\sim F} \|a\|_\infty^2$ is the global sampling coherence \cite{candes2011probabilistic}.  
By subsampling an isometry as in the setting \ref{iso_finite_setting} of Section \ref{subsec:sampling}, and considering totally incoherent sampling, as for instance sampling Fourier frequencies uniformly at random, the previous condition becomes
$$
s \geq c_1 \cdot \ln(3 n/\varepsilon),
$$
{which is not too restrictive in practice.}

\begin{remark}
\label{rem:stability}
\replace{}{The conditions involving $m$, $\Gamma$ and $\Lambda$ in Theorems~\ref{thm:oracle_theta}, \ref{thm:killthetheta}, or \ref{thm:killthetalog} are sufficient to guarantee stable and robust recovery estimates of the form \eqref{eq:noisy_stable_rec}. Moreover, it is possible to show that the same assumptions are also sufficient to guarantee stable and robust recovery for the oracle least-squares estimator. For further details, we refer the reader to the discussion in Appendix~\ref{app:stability}.} 
\end{remark}

\section{Deriving an optimal sampling strategy}
\label{sec:optimal_drawing}

The purpose of this section is to employ the results in Section~\ref{sec:main_results} to derive optimal sampling strategies within the framework of block-\replace{type}{structured} sampling and isolated measurements from a finite-dimensional isometry illustrated in Section~\ref{subsec:sampling}.
\replace{}{Indeed, the bounds on $m$ just presented in Section \ref{sec:main_results} have the remarkable feature to be the first ones in CS with generalized structured sampling, depending only on one quantity, namely either $\Theta$ or $\Lambda$. This is the key to derive theoretical optimal sampling strategies: minimizing such bounds on $m$ with respect to the sampling distribution, leads to the first theoretical closed form for the sampling probabilities when using structured blocks of measurements.}
 All the proofs are presented in Appendix \ref{app:proof_proba}. 

\label{subsec:opt_proba}

The case of block-structured sampling strategy was introduced in \cite{boyer2017compressed} and adapted to parallel acquisition in \cite{chun2017compressed}, where variable density sampling with structured acquisition was considered. However, based on the analysis proposed in these works, minimizing the bound on the number of measurements with respect to $\pi$ is not trivial and no closed form for $\pi$ could be analytically given in general. In the following, Propositions \ref{prop_opt_drawing_proba} and \ref{prop_opt_drawing_proba_lambda} are the first theoretical results  in CS on an optimal way to perform structured acquisition while taking into account partially coherent transforms and structured sparsity.
We start by deriving optimal sampling strategies from Theorems \ref{thm:noiseless} and \ref{thm:noisy} in the finite setting described in 	Section \ref{subsec:sampling}~\ref{block_finite_setting}.

\begin{prop}
\label{prop_opt_drawing_proba}
In a finite setting where we sample blocks in a partitioned isometry as in Section~\ref{subsec:sampling}~\ref{block_finite_setting}, Equation \eqref{eq:bos_sensing_matrix}, the drawing probability minimizing $\Theta(S,\pi)$ is the following:
\begin{align}
\label{eq:opt_drawing_proba}
\forall k \in \{1,\hdots , M \}, \quad \pi_k = \pi_k^\Theta =  \frac{\| D_k^* D_{k,S} \|_{\infty \to \infty}}{\sum_{\ell=1}^M \| D_\ell^* D_{\ell,S} \|_{\infty \to \infty}}.
\end{align}
With such a choice for $\pi$, the required number $m$ of blocks of measurements in Theorems \ref{thm:noiseless} and \ref{thm:noisy} can be then rewritten as follows:
\begin{align}
m \geq c \cdot \sum_{\ell=1}^M \| D_\ell^* D_{\ell,S} \|_{\infty \to \infty} \cdot  \ln^2(6n/\varepsilon),
\end{align}
for a suitable universal constant $c >0$. These conditions ensure noiseless recovery via \eqref{pb:BP}, or stable and robust recovery via \eqref{pb:qBP}, with probability at least $1-\varepsilon$.
\end{prop} 

At the price of the extra assumption, as \eqref{cond:killthetheta} or \eqref{cond:killthethetalog} on $\Lambda{(S,\pi)}$ and $\Gamma{(\pi)}$ in Theorems \ref{thm:killthetheta} or \ref{thm:killthetalog}, one could choose a drawing probability minimizing $\Lambda{(S,\pi)}$ instead of $\Theta{(S,\pi)}$. This is the purpose of the following results, corresponding to the case of block-structured and isolated measurements, respectively. For the sake of simplicity, we just illustrate the derivation from Theorem~\ref{thm:killthetalog}.

\begin{prop}
\label{prop_opt_drawing_proba_lambda}
In a finite setting where we sample blocks in a partitioned isometry as in Section \ref{subsec:sampling} \ref{block_finite_setting}, Equation \eqref{eq:bos_sensing_matrix}, 
assume that $|S|\leq n/2$ and that
\begin{equation}
\label{eq:extra_hyp_1}
1\geq c_1 \cdot \max_{1\leq k \leq M} \frac{\|D_k \|_{1\rightarrow 2}^2}{\| D_{k,S}^* D_{k,S} \|_{2\rightarrow 2}} \ln(3 n/\varepsilon),
\end{equation}
for $c_1$ a universal constant.
Then the optimal drawing probability minimizing $\Lambda(S,\pi)$ is 
\begin{align}
\label{eq:opt_drawing_proba_lambda}
\forall k \in \{1,\hdots , M \}, \quad \pi_k =  \pi_k^\Lambda = \frac{\| D_{k,S}^* D_{k,S} \|_{2 \to 2}}{\sum_{\ell=1}^M \| D_{\ell,S}^* D_{\ell,S} \|_{2 \to 2}}.
\end{align}
With such a choice for $\pi$, the required number $m$ of blocks of measurements in Theorem \ref{thm:killthetalog} can be then rewritten as follows:
\begin{align}
m \geq c_{2} \cdot \sum_{\ell=1}^M \| D_{\ell,S}^* D_{\ell,S} \|_{2 \to 2} \cdot  \ln(6n/\varepsilon),
\end{align}
for a suitable universal constant $c_2 >0$. These conditions ensure noiseless recovery via \eqref{pb:BP}, or stable and robust recovery via \eqref{pb:qBP}, with probability at least $1-\varepsilon$.
\end{prop}

In the following, we specialize the previous result to the case of isolated measurements subsampled from an isometry, described in Section \ref{subsec:sampling} \ref{iso_finite_setting}.

\begin{prop}
\label{prop_opt_proba_iso}
In a finite setting where we sample isolated measurements as in Section~\ref{subsec:sampling}~\ref{iso_finite_setting}, Equation \eqref{eq:bos_iso_sensing_matrix}, the drawing probability minimizing $\Theta(S,\pi)$ is the following:
\begin{align}
\label{eq:opt_drawing_proba_iso}
\forall k \in \{1,\hdots , n \}, \quad \pi_k =  \pi_k^\Theta = \frac{\| \replacemath{a}{d}_k \|_\infty \| \replacemath{a}{d}_{k,S} \|_1}{\sum_{\ell=1}^n \| \replacemath{a}{d}_\ell \|_\infty \| \replacemath{a}{d}_{\ell,S} \|_1}.
\end{align}
With such a choice for $\pi$, the required number $m$ of measurements in Theorems \ref{thm:noiseless} and \ref{thm:noisy} can be then rewritten as follows:
\begin{align}
\label{eq:m_bound_iso}
m \geq c \cdot \sum_{\ell=1}^n \| \replacemath{a}{d}_\ell \|_\infty \| \replacemath{a}{d}_{\ell,S} \|_1 \cdot  \ln^2(6n/\varepsilon),
\end{align}
for a suitable universal constant $c >0$. These conditions ensure noiseless recovery via \eqref{pb:BP}, or stable and robust recovery via \eqref{pb:qBP}, with probability at least $1-\varepsilon$.
\end{prop}

Let us put the previous result into context. Variable density strategies were introduced to deal with partially coherent transforms  (such as in the case of MRI) \cite{krahmer2014stable,chauffert2014variable,puy2011variable}. The idea is to sample the more coherent atoms $(\replacemath{a}{d}_k)_{k}$ {with higher probability} by setting for all $k\in \{1,\hdots , n\}$
\begin{equation}
\label{eq:variabledensity}
\pi_k = \frac{\|\replacemath{a}{d}_k \|_\infty^2}{\sum_{\ell=1}^n \| \replacemath{a}{d}_\ell\|_\infty^2},
\end{equation}
meaning that we tend to sample more where the sampling transform is coherent.
With such a choice, one can ensure robust recovery with probability $1-\varepsilon$ with a number of measurements of the order
\begin{equation*}
m \gtrsim s \cdot \sum_{\ell=1}^n \| \replacemath{a}{d}_\ell\|_\infty^2 \cdot \ln \left( n / \varepsilon \right).
\end{equation*}
For instance in the case of MRI, $\sum_{\ell=1}^n \| \replacemath{a}{d}_\ell\|_\infty^2 = O(\ln(n))$, making bounds on $m$ more realistic than those based on the global coherence $n \max_k \|\replacemath{a}{d}_k\|_\infty^2 = O(n \ln(n))$. The sampling strategy defined by \eqref{eq:opt_drawing_proba_iso} refines \eqref{eq:variabledensity} since it takes into account a more accurate notion of local coherence, involving both the coherence of the individual atoms $(\replacemath{a}{d}_k)_k$ and its interaction with prior information about the support of the signal. In this way, the sampling strategy can be optimized to recover signals with structured sparsity, as discussed in Sections~\ref{sec:applications_mri_iso} and \ref{sec:application_MRI} for the case of sparsity in levels. Moreover, the explicit dependence of the sampling probability on the support in \eqref{eq:opt_drawing_proba_iso} enables us to devise iterative adaptive sampling strategies (see Section~\ref{sec:adaptive}).

Still in the case of isolated measurements, one can improve the previous proposition by deriving sampling strategies derived from oracle-type bounds on the number of measurements. 
The following result applies Theorem~\ref{thm:killthetalog} to the setting illustrated in Section~\ref{subsec:sampling}~\ref{iso_finite_setting}.

\begin{prop}
\label{prop_opt_drawing_proba_iso_lambda}
In a finite setting where we sample isolated measurements as in Section \ref{subsec:sampling} \ref{iso_finite_setting}, Equation \eqref{eq:bos_iso_sensing_matrix}, 
assume that {$|S|\leq n/2$ and that}
\begin{align}
\label{cond_to_kill_log_iso}
1\geq c_1 \cdot \max_{1\leq k \leq n} \frac{\|\replacemath{a}{d}_k \|_{\infty}^2}{\| \replacemath{a}{d}_{k,S} \|_{2}^2} \ln(3 n/\varepsilon),
\end{align}
for $c_1$ a universal constant.
Then the optimal drawing probability minimizing $\Lambda(S,\pi)$ is 
\begin{align}
\label{eq:opt_drawing_proba_lambda_iso}
\forall k \in \{1,\hdots , n \}, \quad \pi_k  =  \pi_k^\Lambda  = \frac{\| \replacemath{a}{d}_{k,S} \|_{2}^2}{\sum_{\ell=1}^n \| \replacemath{a}{d}_{\ell,S} \|_{2}^2}.
\end{align}
The required number $m$ of measurements in Theorem \ref{thm:killthetalog} can be then rewritten as follows:
\begin{align}
m \geq c_{2} \cdot \sum_{\ell=1}^n \| \replacemath{a}{d}_{\ell,S} \|_{2}^2 \cdot  \ln(6n/\varepsilon),
\end{align}
{for a suitable universal constant $c_2 >0$.} These conditions ensure noiseless recovery via \eqref{pb:BP}, or stable and robust recovery via \eqref{pb:qBP}, with probability at least $1-\varepsilon$.
\end{prop} 

Finally, one could resort to Theorem~\ref{thm:oracle_theta} in order to remove condition \eqref{cond_to_kill_log_iso}, at the price of overestimating the support size by one element and of an extra log factor.

\begin{prop}
\label{prop:drawing_proba_lambda_ext}
Let $x\in \Cbb^n$ be a vector supported on $S$, such that $|S|=s$ and $\sgn(x_S)$ forms a Rademacher or Steinhaus sequence. 
Let $A$ be the random sensing matrix defined in \eqref{eq:sensing_matrix} with parameter $\Lambda(S,F)$. If {$0 < \varepsilon < 1$ and}
$$
m \geq c_1 \cdot \max_{j \in S^c} \Lambda (S \cup \{j\},F)  \cdot \ln^{2}\left( \frac{6n}{\varepsilon} \right),
$$ 
with $c_1>0$ a numerical constant, then, with probability at least $1-\varepsilon$, noiseless recovery via \eqref{pb:BP} or stable and robust recovery via \eqref{pb:qBP} is ensured. 
The probability $\pi$ minimizing the previous bound on $m$ reads:
\begin{align}
\label{eq:opt_drawing_proba_lambda_iso_tilde}
\forall k \in \{1,\hdots , n \}, \quad \pi_k = \pi_k^{\tilde{\Lambda}} = \frac{\displaystyle\max_{j \in S^c}\| \replacemath{a}{d}_{k,S \cup \{j\}} \|_{2}^2}{\displaystyle\sum_{\ell=1}^n \max_{j \in S^c}\| \replacemath{a}{d}_{\ell,S \cup \{j\}} \|_{2}^2}.
\end{align}
With such a choice, the required number $m$ of measurements can be rewritten as follows:
$$
m \geq c_1 \cdot \sum_{\ell=1}^n \max_{j \in S^c}\| \replacemath{a}{d}_{\ell,S \cup \{j\}} \|_2^2 \cdot  \ln^{2}\left( \frac{6n}{\varepsilon} \right).
$$ 
\end{prop}



\section{Applications and numerical experiments}
\label{sec:appli}

In this section, we discuss the application of the optimal sampling strategies derived in Section~\ref{sec:optimal_drawing} to three concrete case studies. Firstly, we consider the setting of isolated measurements from the one-dimensional Fourier-Haar transform and derive an explicit sampling probability that promotes the recovery of signals that are sparse in levels  (Section~\ref{sec:applications_mri_iso}). Secondly, the case of two-dimensional Fourier-Haar transform with block-structured measurements is discussed in Section~\ref{sec:application_MRI}, of considerable interest in MRI. In particular, we analyze the setting where the Fourier space is sampled using random {vertical (or horizontal)} lines and where the signal exhibits a particular type of sparsity in levels (see Definition~\ref{def:sparsity_in_levels_2D}). Thirdly, we illustrate how the proposed derivation of optimal sampling measures can be applied to devise adaptive sampling strategies for one-dimensional function approximation from pointwise data (Section~\ref{sec:adaptive}).

\subsection{Performing isolated measurements with Fourier-Haar transform}
\label{sec:applications_mri_iso}

Let us consider the finite-dimensional setting \ref{iso_finite_setting} described in Section \ref{subsec:sampling} {and let $n = 2^{J+1}$ for some $J \in \Nbb$}. 
Set $A_0$ to be the Fourier-Haar transform, i.e.\ $A_0 = {\mathcal{F}} H^* = \replacemath{(a_1 | a_2 | \hdots | a_n)^*}{(d_1 | d_2 | \hdots | d_n)^*}$ where $\mathcal{F} \in \Cbb^{n\times n}$ is the Fourier matrix {relative to frequencies $\{-n/2 +1, \ldots, n/2\}$} and $H^* \in \Rbb^{n\times n}$ is the inverse Haar transform. The matrix $A_0$ could be naturally chosen to model the acquisition in MRI, where the measurements are performed in the Fourier domain, and where MR images are considered sparse in the wavelet domain. In this framework, we aim at reconstructing the wavelet coefficients $x$ of the MR image.

First, we introduce the notion of sparsity in levels for the coefficient vector $x$. We define the levels $(\Omega_j)_{0\leq  j \leq J}$ as the partition of $\{1, \hdots , n \}$ corresponding to the subbands of the Haar transform, defined as $\Omega_0 := \{1,2\}$ and, for $j = 1,\ldots,J$, as
$$
\Omega_j := \{2^{j}+1,\ldots,2^{j+1}\}.
$$
The subband $\Omega_0$ corresponds to the scaling function and to the Haar mother wavelet. The subsequent subbands $\Omega_j$ with $j\geq 1$ correspond to wavelet functions dyadically refined at level $j$ (see \cite{adcock2016note} for further details).
The vector $x$ is assumed to be sparse in levels, i.e.\ denoting by $S$ the support of $x$, we suppose
\begin{align*}
|S \cap \Omega_j| = s_j,
\end{align*}
meaning that restricted to the $j$-th level $\Omega_j$, $x$ is $s_j$-sparse. 

In order to devise effective sampling strategies, we  partition the set of Fourier frequencies into frequency subbands $(W_j)_{0\leq j \leq J}$, defined as $W_0 := \{0,1\}$ and, for $j = 1,\ldots, J$, as
$$
W_j := \{-2^j+1,\ldots,-2^{j-1}\} \cup \{2^{j-1}+1\ldots,2^j\}.
$$
Notice that for every $j = 0,\ldots,J$ we have
$$
|\Omega_j| = |W_j|= 2^{\max(1,j)}.
$$
For the sake of simplicity, we will relabel the Fourier  frequencies $\{-n/2 +1 ,\ldots, n/2\}$ as $\{1,\ldots,n\}$. With this convention, we  denote as $j : \{1, \hdots , n \} \to \{0 , \hdots , J\}$ the function that maps a Fourier frequency to the corresponding frequency band, i.e.\ $j(k) = \ell$ if $k \in W_\ell$.

In this setting, local coherence upper bounds are explicitly computable. We recall the following estimate, corresponding to \cite[Lemma 1]{adcock2016note}:
\begin{equation}
\label{eq:localcoherenceFH}
\mu_{j,\ell} := \max_{k \in W_j} \max_{i\in \Omega_\ell} |({\mathcal{F}}H^*)_{ki}|^2 \leq C \cdot 2^{-j} 2^{-|j-\ell|},
\end{equation}
where $C>0$ is a universal constant.

\paragraph{Application of the theory} Suppose that random isolated measurements are performed in the Fourier domain according to some probability distribution $\pi$ (we note in passing that, although useful for illustrative purposes, sampling single points in the Fourier domain is not feasible in practice). Then, thanks to the local coherence upper bound \eqref{eq:localcoherenceFH}, one can estimate the quantity $\Theta(S,\pi)$ {as} (see \cite[Corollary 4.4]{boyer2017compressed}) 
\begin{align}
\Theta(S,\pi) = C \max_{1\leq k \leq n } \frac{1}{\pi_k} 2^{-j(k)} \sum_{p=0}^J 2^{-|j(k)-p|/2} s_p.
\end{align}
As a consequence, the probability \replace{}{distribution} minimizing $\Theta(S,\pi)$ in this setting can be written as follows (see Lemma~\ref{lem:max_over_simplex}):
\begin{align}
\forall k \in \{1,\hdots , n \}, \quad \pi^{\Theta}_k = \frac{2^{-j(k)} \sum_{p=0}^J 2^{-|j(k)-p|/2} s_p}{\sum_{\ell=1}^n 2^{-j(\ell)} \sum_{p=0}^J 2^{-|j(\ell)-p|/2} s_p}.
\end{align}
Note that the drawing probability \replace{}{distribution} is a function of the frequency level.
Assuming that $\sgn(x_S)$ is a Steinhaus or Rademacher sequence and employing Theorems \ref{thm:noiseless} and \ref{thm:noisy}, the required number of measurements needed to ensure exact or stable and robust recovery with probability $1-\varepsilon$ is 
\begin{align}
\label{eq:bound_m_iso_mri}
m \gtrsim \sum_{j=0}^J \left( s_j + \sum_{p=0 \atop p \neq j}^J 2^{-|j-p|/2} s_p \right) \cdot \ln^2(6n/\varepsilon),
\end{align} 
which matches results in \cite[Corollary 4.4]{boyer2017compressed} and \cite{adcock2017breaking}.

Replacing $\Theta(\pi,S)$ with $\Lambda(\pi,S)$ leads to an improvement of condition \eqref{eq:bound_m_iso_mri}, at the price of an extra assumption on the sparsity in levels of the signal.

\begin{corollary}
\label{corol:mri_iso_lambda}
Let $A_0 = FH^*$ be the one-dimensional Fourier-Haar transform and consider the splitting of the Haar and of the Fourier spaces into subbands $(\Omega_j)_{0 \leq j \leq J}$ and $(W_j)_{0 \leq j \leq J}$ as described above. Suppose that $x \in \Cbb^{n}$ is an $S$-sparse vector with random signs and sparsities in levels $(s_j)_{0\leq j \leq J}$.
Fix $\varepsilon \in (0,1)$. 
Then, condition \eqref{cond:killthethetalog} holds if
\begin{align}
\label{cond:killthetalog_mri_iso}
\min_{0\leq j \leq J} s_{j}  + \sum_{j'\neq j}  {s}_{j'}  2^{-|j-j'|} \gtrsim \ln({3}n/\varepsilon).
\end{align}
In this case, $\Lambda(S,\pi)$ can be chosen as follows 
\begin{align}
\Lambda({S},\pi) = \max_{1\leq k \leq n } \frac{1}{\pi_k} 2^{-j(k)} \sum_{p=0}^J 2^{-|j(k)-p|} {s}_p.
\end{align}
Then, the probability minimizing $\Lambda({S},\pi)$ is
\begin{align}
\forall k \in \{1,\hdots , n \}, \quad \pi^{\Lambda}_k = \frac{2^{-j(k)} \sum_{p=0}^J 2^{-|j(k)-p|} s_p}{\sum_{\ell=1}^n 2^{-j(\ell)} \sum_{p=0}^J 2^{-|j(\ell)-p|} s_p},
\end{align}
and the recovery guarantees of Theorem \ref{thm:killthetalog} hold with probability at least $1-\varepsilon$ if
\begin{align}
\label{eq:bound_m_iso_mri2}
m \gtrsim  \sum_{j=0}^J \left({s}_j + \sum_{p=0 \atop p \neq j}^J 2^{-|j-p|} {s}_p \right) \cdot \ln({3}n/\varepsilon).
\end{align}

\end{corollary}

The proof of this result is given in Appendix \ref{proof:mri_setup_iso}. Condition \eqref{cond:killthetalog_mri_iso} can be easily verified in practice. 
Note that by \eqref{eq:bound_m_iso_mri2}, we improved \eqref{eq:bound_m_iso_mri} (and, consequently, the state-of-the-art results in \cite{boyer2017compressed} and \cite{adcock2017breaking} when the signal is sparse with random signs) by (i) decreasing the interferences between different levels and (ii) removing a log factor.

\paragraph{Numerical experiments} In this part, we illustrate the improvement in recovery when the sampling strategy \replace{is taking}{takes} into account local coherences of the measurement vectors \emph{and} the structured sparsity of the signal to reconstruct. With this aim, we generate 100 random {signals with} wavelet coefficients in dimension $n=2048$ with $s/n \approx 6\%$ and having a  sparsity in levels structure as in Figure \ref{fig:boxplot_15meas}(a) with 4 subbands in total. 
We compare three different sampling schemes where isolated samples are randomly drawn according to:
\begin{itemize}
\item the probability \replace{}{distribution} minimizing the global coherence of each measurement vector: $\forall k \in \{ 1 , \hdots , n \}$, 
$$ \pi_k^\infty = \frac{\| \replacemath{a}{d}_{k} \|_\infty^2}{\sum_{p=1}^n \| \replacemath{a}{d}_p \|_\infty^2};
$$ 
\item the probability \replace{}{distribution} minimizing $\Theta(S,\pi)$: $\forall k \in \{ 1 , \hdots , n \}$, 
$$ \pi_k^\Theta = \frac{\| \replacemath{a}{d}_{k} \|_\infty \sum_{j=0}^J s_j \| \replacemath{a}{d}_{k,\Omega_j} \|_\infty}{\sum_{p=1}^n \| \replacemath{a}{d}_{p} \|_\infty \sum_{j=0}^J s_j \| \replacemath{a}{d}_{p,\Omega_j} \|_\infty};
$$ 
\item the probability \replace{}{distribution} minimizing $\Lambda(S,\pi)$: $\forall k \in \{ 1 , \hdots , n \}$, 
$$ \pi_k^\Lambda = \frac{\sum_{j=0}^J s_j \| \replacemath{a}{d}_{k,\Omega_j} \|_\infty^2}{\sum_{p=1}^n  \sum_{j=0}^J s_j \| \replacemath{a}{d}_{p,\Omega_j} \|_\infty^2}.
$$ 
\end{itemize}  
\replace{}{These three sampling distributions can be computed in the case of the Fourier-Haar transform.  This is illustrated in Figure~\ref{fig:proba_profile_iso_mri} for the 1D structured-sparsity pattern described in Figure \ref{fig:boxplot_15meas}(a). One can clearly see in Figure \ref{fig:proba_profile_iso_mri} that a certain prior on the sparsity structure in this case inclines the sensing in the low frequency domain.}
\replace{Numerical}{Reconstruction} results are shown in Figure \ref{fig:boxplot_15meas}(b). \replace{}{For the same amount of measurements,} the sampling strategy minimizing the bound on $\Theta (S,\pi)$ proposed in this paper is outperforming  the standard CS strategy taking only into account the global coherence. In this setting, sampling according to $\Lambda (S,\pi)$ improves the reconstruction quality compared to the $\Theta (S,\pi)$-based strategy.

\begin{figure}
\begin{tabular}{ccc}
\includegraphics[width=0.33\linewidth]{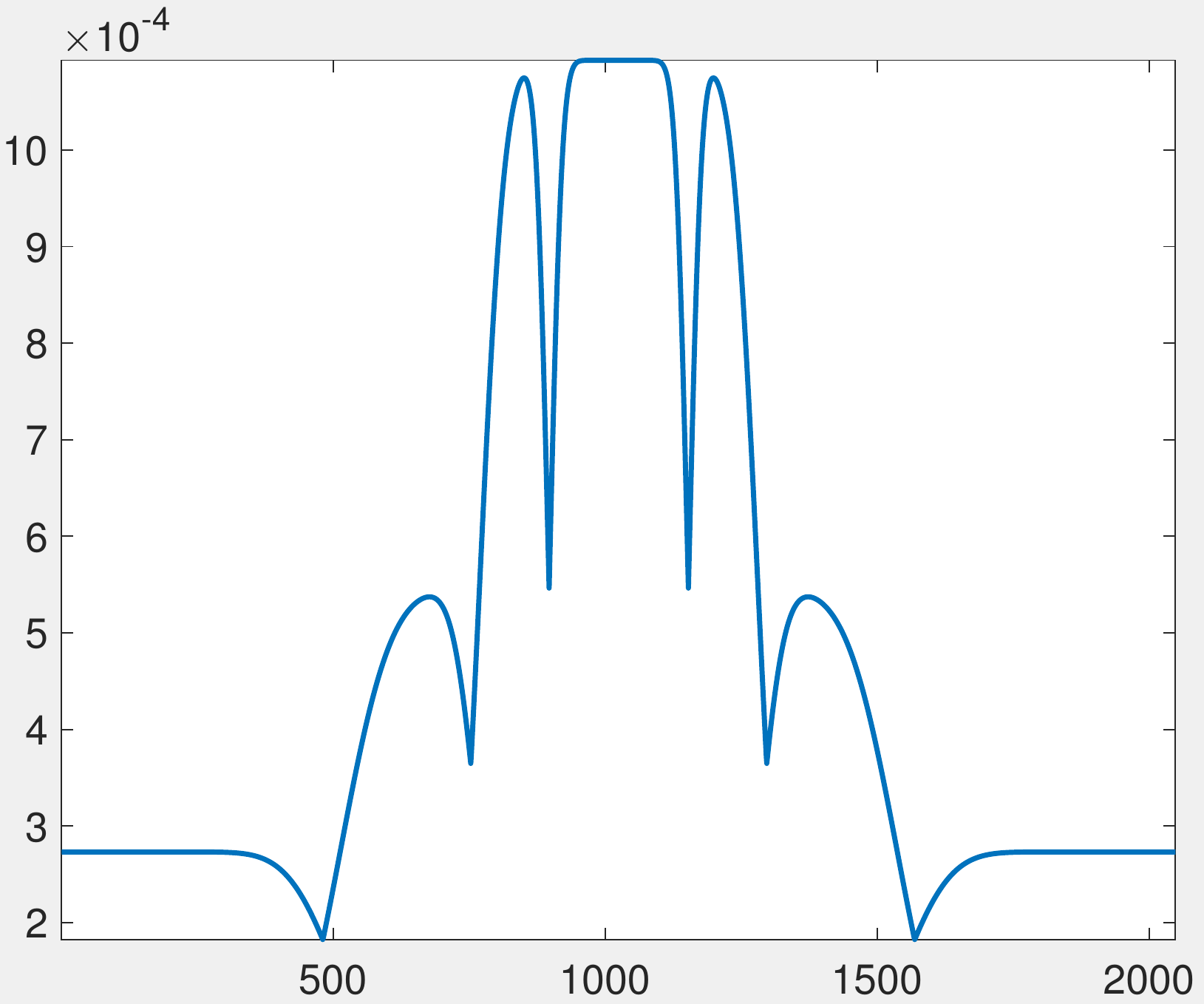} & 
\includegraphics[width=0.33\linewidth]{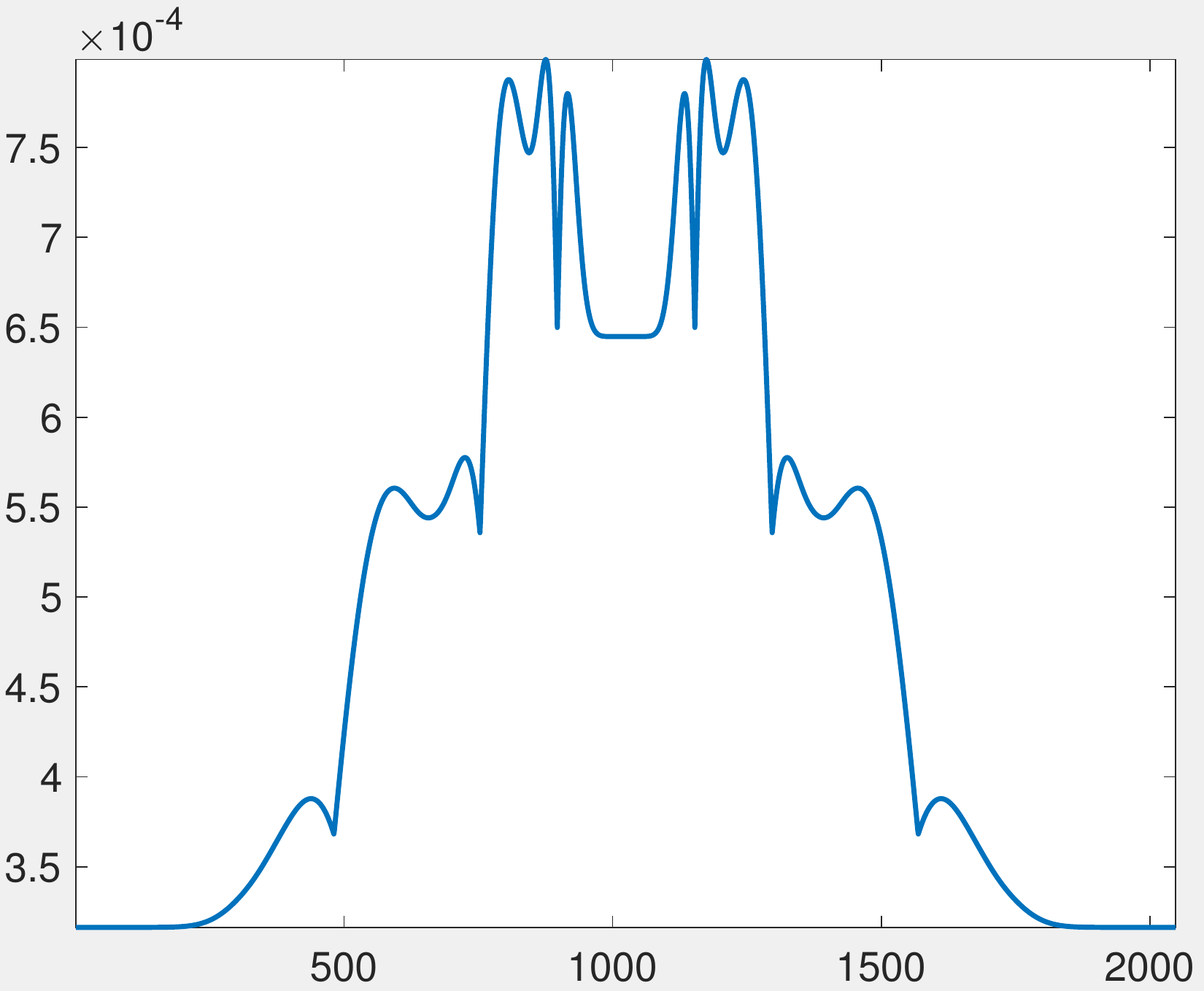}  &
\includegraphics[width=0.33\linewidth]{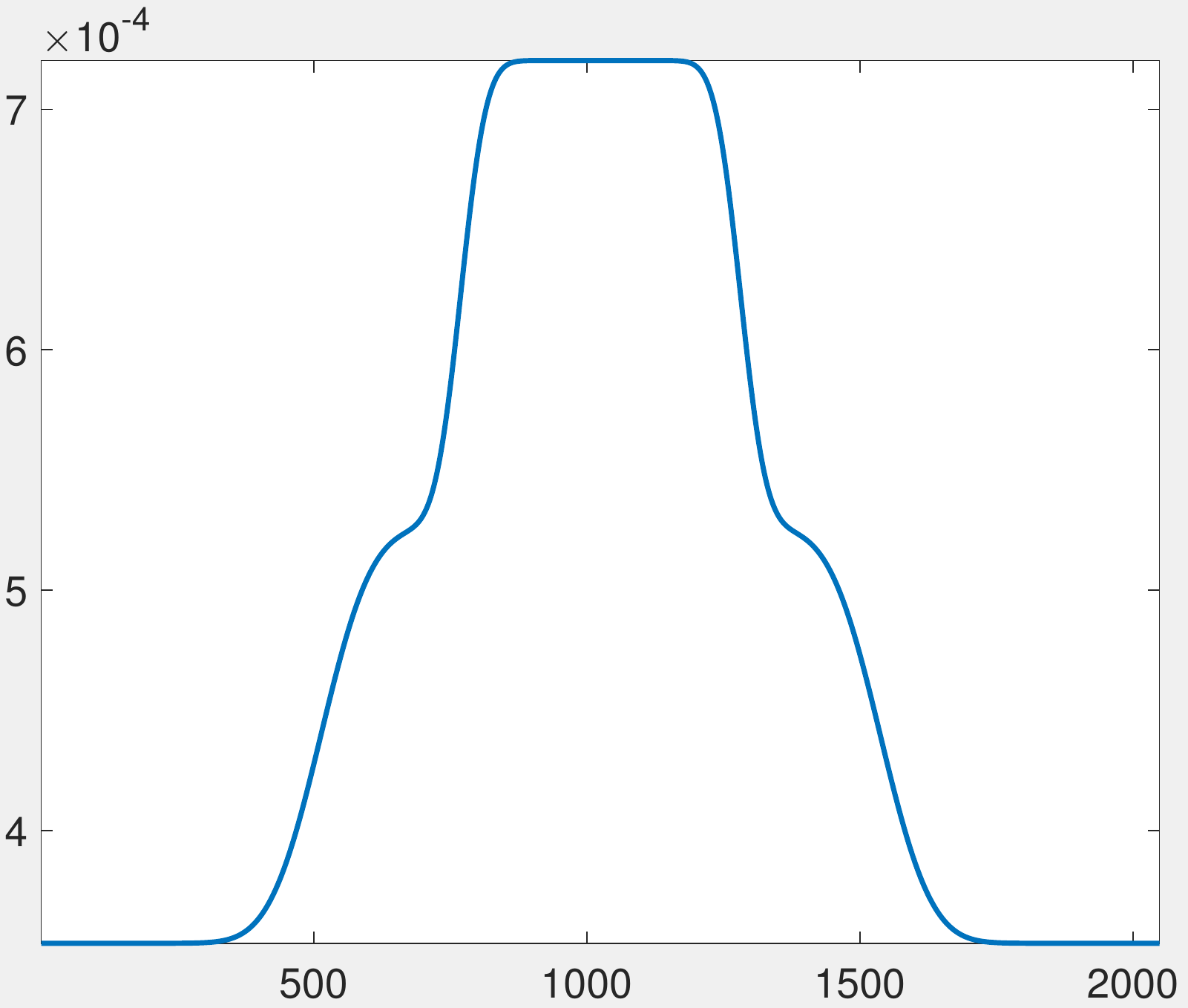}  \\
{\centering (a) $\pi^\infty$} & {\centering (b) $\pi^\Theta$} & {\centering (c) $\pi^\Lambda$} \\
\includegraphics[width=0.33\linewidth]{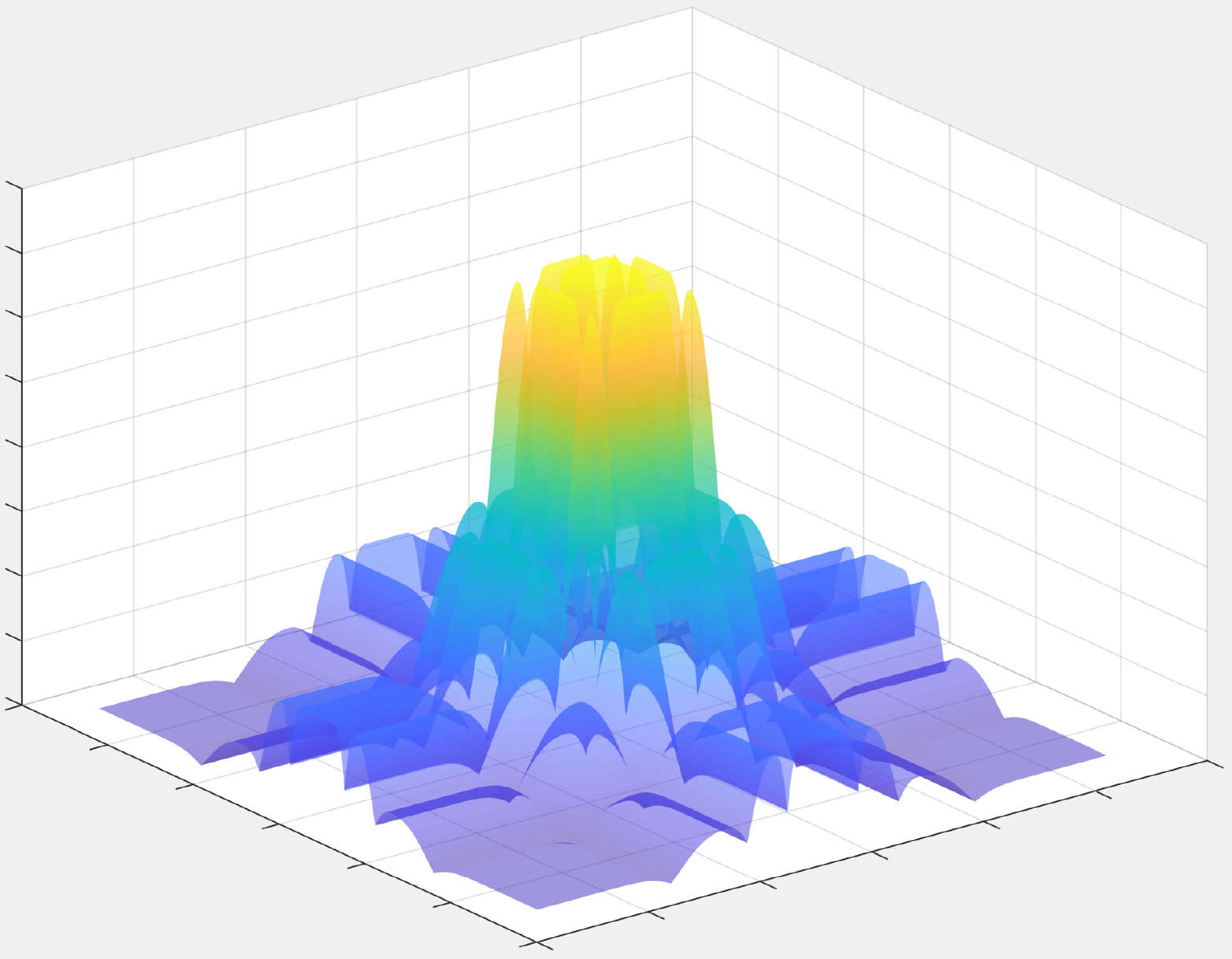} & 
\includegraphics[width=0.33\linewidth]{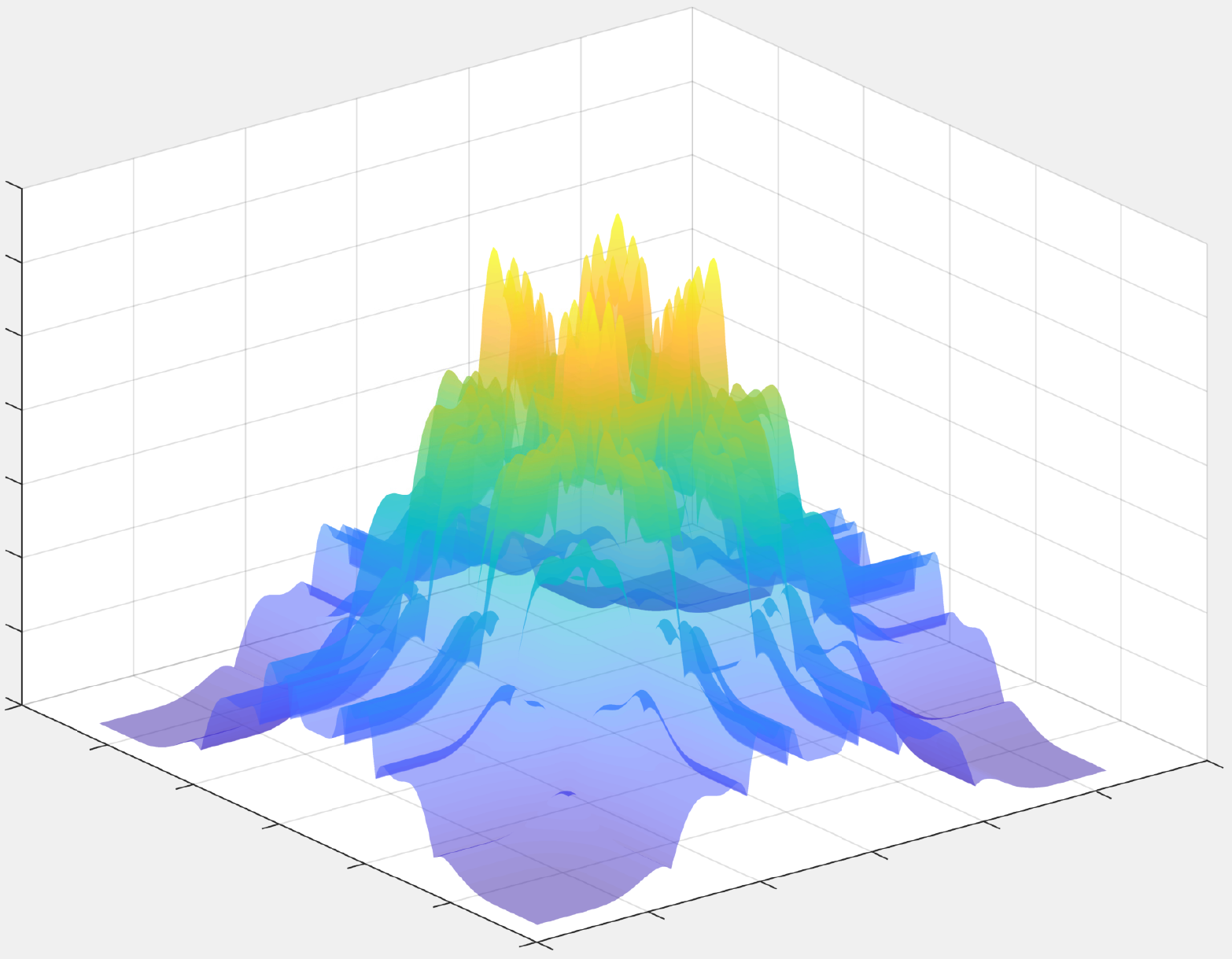}  &
\includegraphics[width=0.33\linewidth]{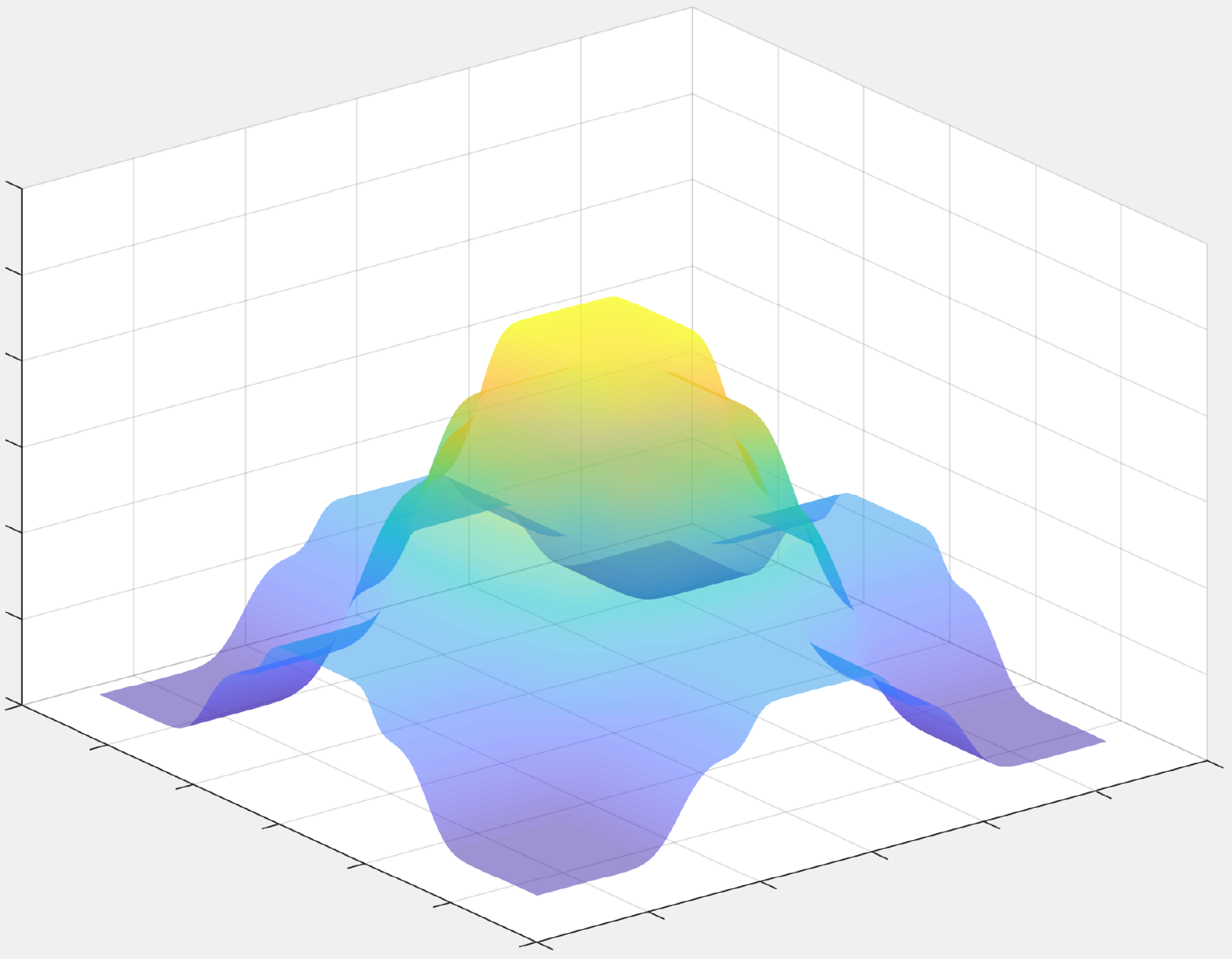}  \\
{\centering (d) $\pi^\infty$} & {\centering (e) $\pi^\Theta$} & {\centering (f) $\pi^\Lambda$}
\end{tabular}
\caption{\label{fig:proba_profile_iso_mri} Comparison between sampling probabilit\replace{ies}{y distributions} chosen according to different strategies\replace{ for 1D signals (a,b,c) or 2D signals (d,e,f), that are sparse in levels}{: in (a,b,c) for 1D signals with structured sparsity described in Figure \ref{fig:boxplot_15meas}(a), in (d,e,f) for 2D signals with the corresponding tensorized structured sparsity}. 
In (a) and (d), the sampling probability \replace{}{distribution} $\pi^\infty$ is optimized to minimize the global coherence, i.e.\ $\pi_k \propto \|a_k\|_\infty^2$; 
in (b) and (e), the sampling probability \replace{}{distribution} $\pi^\Theta$ is optimized to minimize an upper bound to $\Theta(S,\pi)$, i.e.\ $\pi_k \propto \|a_k\|_\infty \sum_{j=0}^J s_j  \|a_{k,\Omega_j} \|_\infty$; in (c) and (f), the sampling probability \replace{}{distribution} $\pi^\Lambda$ is optimized to minimize an upper bound to $\Lambda(S,\pi)$, i.e.\ $\pi_k \propto \sum_{j=0}^J \|a_k\|_\infty \|a_{k,\Omega_j} \|_\infty^2$; the probability \replace{densities}{distributions} in (d), (e), and (f) are constructed accordingly in the 2D case.}
\end{figure}

\begin{figure}
\begin{tabular}{cc}
\includegraphics[height=6cm]{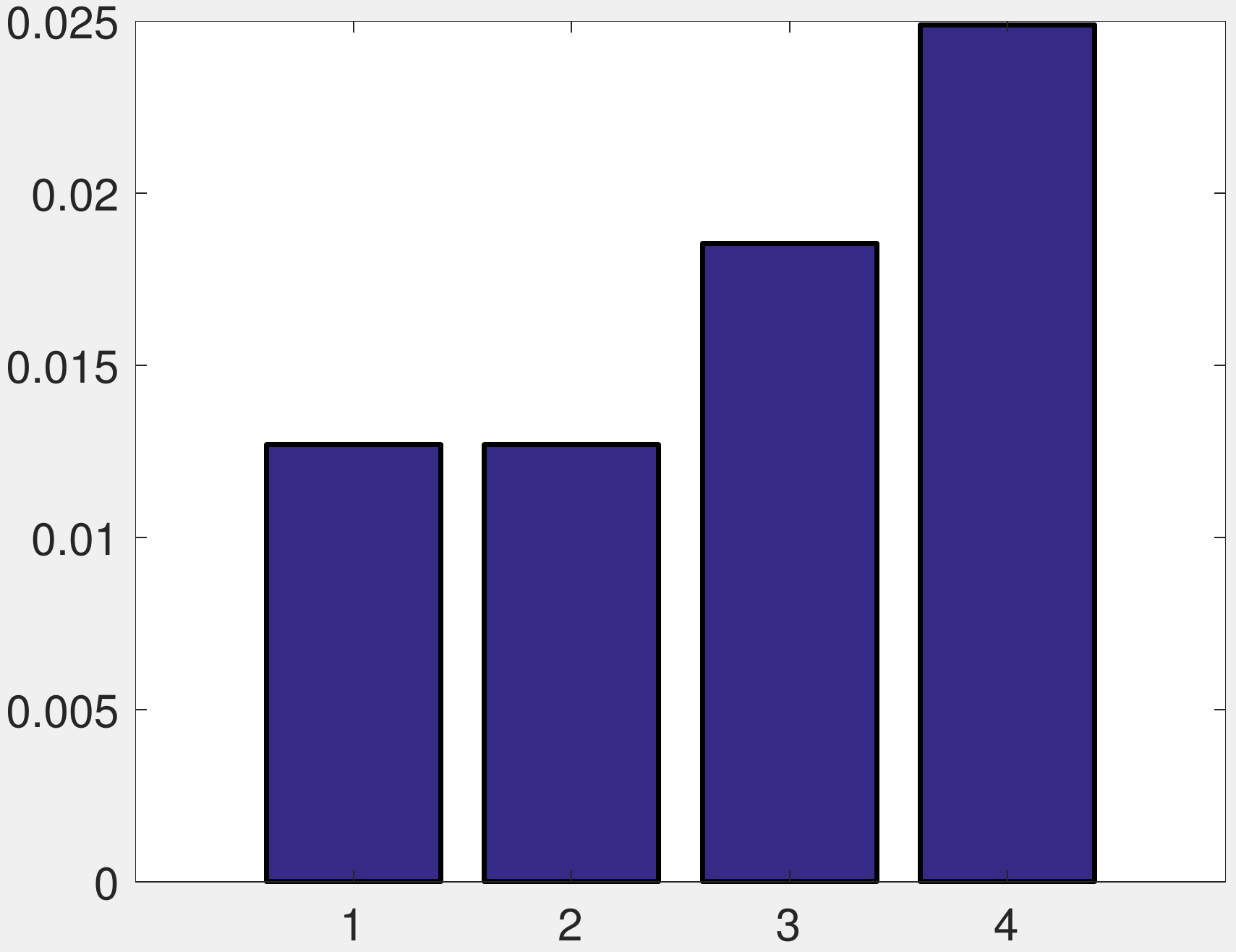} &
\includegraphics[height=6cm]{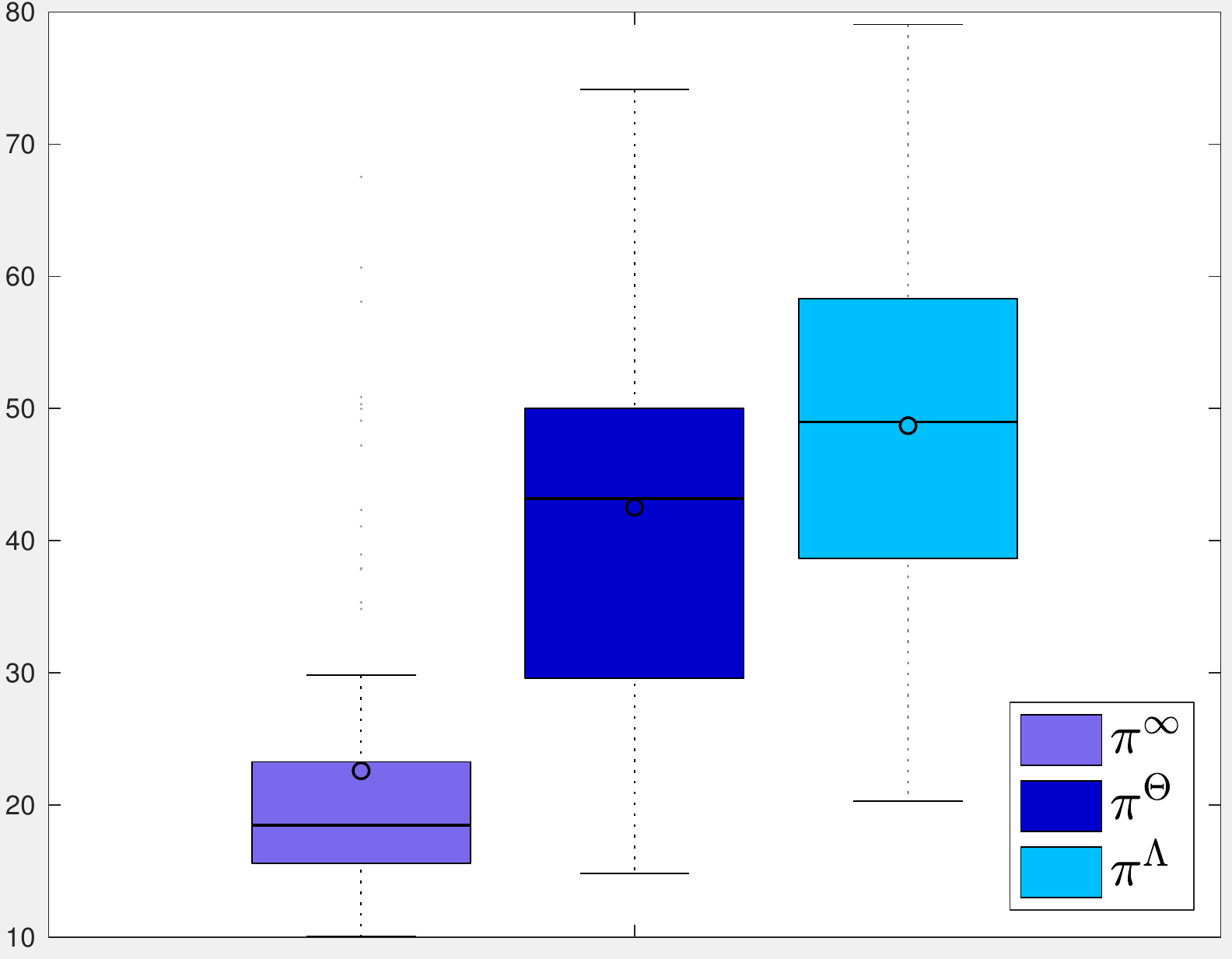} \\
{\text{(a)}} & {\text{(b)}}
\end{tabular}
\caption{\label{fig:boxplot_15meas}
Boxplot of reconstruction PSNR in (b) of 100 random signals of length $n=2048$ having a structured sparsity in the wavelet decomposition as in (a) ($s/n=6\%$) from 25\% measurements drawn according to $\pi^\infty$, $\pi^\Theta$ and $\pi^\Lambda$ as in Figure \ref{fig:proba_profile_iso_mri} (a,b,c). In (a), we represent the {sparsity in levels structure of the randomly generated signals. The plot represents the} percentage of nonzero coefficients {for each} subband. }
\end{figure}


\subsection{Sensing vertical (or horizontal) lines in MRI}
\label{sec:application_MRI}

We extend the analysis of Section~\ref{sec:appli} from the one-dimensional to the two-dimensional framework, leading a more realistic model for the MRI problem.

Let $n = 2^{2(J+1)}$ for some $J \in \Nbb$ and denote the one-dimensional Fourier-Haar transform introduced in Section~\ref{sec:applications_mri_iso} as $\phi = {\mathcal{F}}H^* \in \Cbb^{\sqrt{n} \times \sqrt{n}}$. Then, we define the two-dimensional Fourier-Haar transform as 
$$
A_0 = \phi \otimes \phi \in\Cbb^{n \times n},
$$ 
where $\otimes$ denotes the Kronecker product. The splitting of the wavelet multi-index space $\{1,\ldots,\sqrt{n}\}^2$ and of the Fourier frequency space $\{-\sqrt{n}/2 + 1, \ldots, \sqrt{n}/2\}^2$ into subbands is carried out by tensorizing the one-dimensional splitting presented in Section~\ref{sec:applications_mri_iso}. Namely, we consider subbands defined, for $j,j' = 0,\ldots,J$, as
$$
\Omega_{j,j'}:= \Omega_{j} \times \Omega_{j'}
\quad \text{and} \quad
W_{j,j'}:= W_j \times W_{j'}. 
$$
Moreover, we will represent the image wavelet coefficients as a vector $x \in \Cbb^n$ or, equivalently, as a matrix $X = \Cbb^{\sqrt{n} \times \sqrt{n}}$, such that $x = \vect(X)$, where $\vect: \Cbb^{\sqrt{n} \times \sqrt{n}} \to \Cbb^n$ is the vectorization operator that stacks the columns of a matrix atop, i.e., $\vect(x_1|\cdots|x_{\sqrt{n}}) = (x_1^*,\ldots,x_{\sqrt{n}}^*)^*$. Within this multi-level framework, we consider the following generalization of sparsity in levels.
\begin{definition}
\label{def:sparsity_in_levels_2D}
Let $x \in \Cbb^n$ and $X \in \Cbb^{\sqrt{n} \times \sqrt{n}}$ be such that $x = \vect(X)$. Then, given $S = \text{supp}(X)$, we define the following sparsities in levels for the vector $x$:
\begin{align}
\label{eq:scol}
s^{r}_{j} &:= \max_{0\leq \ell \leq J} \max_{ k \in {\Omega}_{\ell}  } \left| S \cap \Omega_{\ell,j} \cap R_k\right| , 
\end{align}
where $R_k {:=\{k\}\times \{1,\ldots,\sqrt{n}\}}$ represents the set corresponding to the $k$-th horizontal line or row of the wavelet multi-index space $\{1,\ldots,\sqrt{n}\}^2$  (see Figure \ref{fig:illusEns}). 
\end{definition}

 \begin{figure}[t]
  \begin{center}
\begin{tikzpicture}
\draw (0,0) rectangle (1,1) node[midway,color=Cerulean] {\large$\Omega_{0,0}$};
\draw (1,0) rectangle (2,1) node[midway,color=Cerulean] {\large$\Omega_{0,1}$};
\draw (2,0) rectangle (4,1) node[midway,color=Cerulean] {\large ...};
\draw (4,0) rectangle (8,1) node[midway,color=Cerulean] {\large$\Omega_{0,J}$};

\draw (0,0) rectangle (1,-1) node[midway,color=Cerulean] {\large $\Omega_{1,0}$};
\draw (0,-1) rectangle (1,-3) node[midway,color=Cerulean] {\large $\vdots$};
\draw (0,-3) rectangle (1,-7) node[midway,color=Cerulean] {\large $\Omega_{J,0}$};

\draw (1,0) rectangle (2,-1) ;
\draw (2,0) rectangle (4,-1) node[midway] {};
\draw (4,0) rectangle (8,-1) ;

\draw (1,-1) rectangle (2,-3) ;
\draw (1,-3) rectangle (2,-7) ;

\draw (2,-1) rectangle (4,-3) node[midway,color=Cerulean] {\large $\ddots$};
\draw (2,-3) rectangle (4,-7) ;

\draw (4,-1) rectangle (8,-3) ;

\draw (4,-3) rectangle (8,-7) node[midway,color=Cerulean] {\large $\Omega_{J,J}$};

%

\draw (8.2,-3.7) node[right, color=Emerald]{$R_k$} ;
\fill [Emerald,opacity=0.6] (0,-3.5) rectangle (8,-3.8);
\draw[->,>=latex, color=Violet] (8.2,-4.5) to[bend left] (6.5,-3.8);
\draw (8.2,-4.5) node[right, color=Violet]{$R_k\cap \Omega_{J,J}$} ;
\fill [Violet,opacity=0.5] (4,-3.5) rectangle (8,-3.8);

\draw (0,2) -- (8,2);
\draw (0,2) node {$\bullet$} ;
\draw (1,2) node {$\bullet$} ;
\draw (2,2) node {$\bullet$} ;
\draw (4,2) node {$\bullet$} ;
\draw (8,2) node {$\bullet$} ;

\draw (0,2.2) node[above] {$1$} ;
\draw (8,2.2) node[above] {$\sqrt{n}$} ;
\draw (-1,-3.65) node {$\blacksquare$} ;
\draw (-1.2,-3.65) node[left] {$k$} ;

\draw[decoration={brace,raise=5pt},decorate]
  (0,2.7) -- node[above=6pt] {$\Omega_0$} (1,2.7);
\draw[decoration={brace,raise=5pt},decorate]
  (1,2.7) -- node[above=6pt] {$\Omega_1$} (2,2.7);
 \draw[decoration={brace,raise=5pt},decorate]
  (4,2.7) -- node[above=6pt] {$\Omega_J$} (8,2.7);

\draw (-1,1) -- (-1,-7);
\draw (-1,1) node {$\bullet$} ;
\draw (-1,0) node {$\bullet$} ;
\draw (-1,-1) node {$\bullet$} ;
\draw (-1,-3) node {$\bullet$} ;
\draw (-1,-7) node {$\bullet$} ;

\draw (-1.2,1) node[left] {$1$} ;
\draw (-1.2,-7) node[left] {$\sqrt{n}$} ;

\draw[decoration={brace,mirror,raise=5pt},decorate]
  (-1.7,1) -- node[left=6pt] {$\Omega_0$} (-1.7,0);
\draw[decoration={brace,mirror,raise=5pt},decorate]
  (-1.7,0) -- node[left=6pt] {$\Omega_1$} (-1.7,-1);
\draw[decoration={brace,mirror,raise=5pt},decorate]
  (-1.7,-3) -- node[left=6pt] {$\Omega_J$} (-1.7,-7);

\end{tikzpicture}

 \caption{\label{fig:illusEns}2D view of the signal $x\in \Cbb^n$ to reconstruct. The vector $x$ can be reshaped into a $\sqrt{n}\times \sqrt{n}$ matrix. The \replace{indexes}{indices} can be partitioned according to the sets $\Omega_{j,j'} = \Omega_j \times \Omega_{j'} $ for $0\leq j,j'\leq J$. ${R}_k$ represents the coefficient \replace{indexes}{indices} corresponding to the $k$-th horizontal row. }
  \end{center}
 \end{figure}
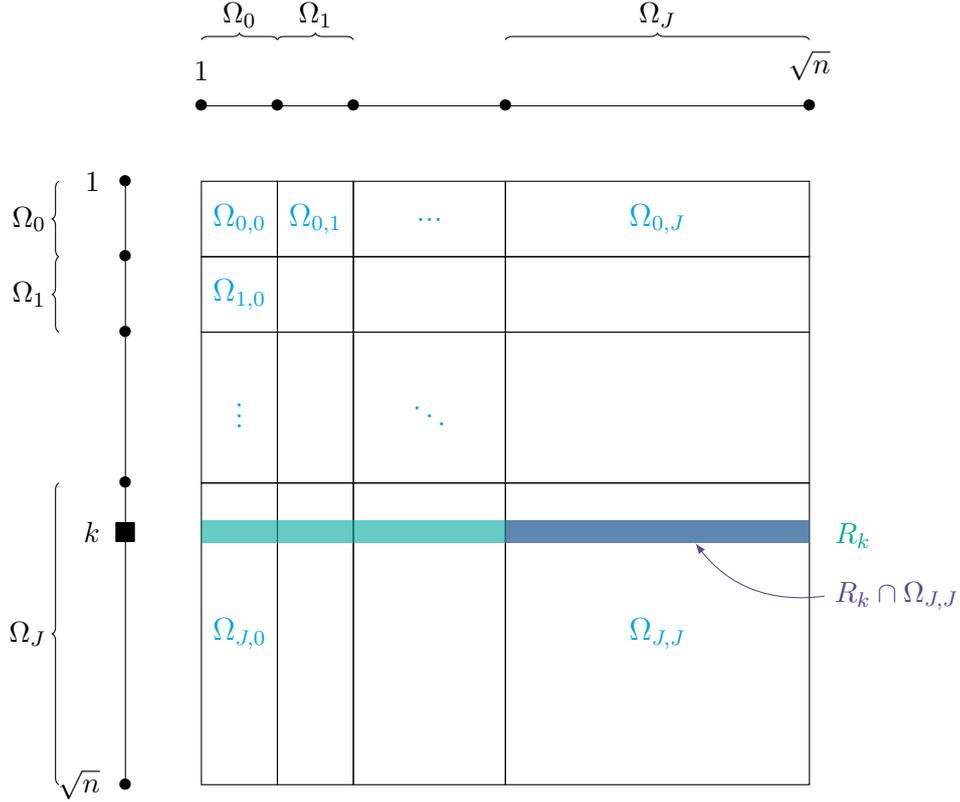

The sparsities in levels $(s_j^r)_{0\leq j\leq J}$ are intrinsically anisotropic. According to this definition, the wavelet coefficients restricted to the vertical subband $\{1,\ldots,\sqrt{n}\}\times\Omega_j$ and to a generic row can be at most $s_j^r$-sparse.


 \begin{figure}[t]
  \begin{center}
\begin{tikzpicture}
\draw (0,0) rectangle (1,-1);
\draw (1,0) rectangle (2,-1);
\draw (2,0) rectangle (4,-1);
\draw (4,0) rectangle (8,-1);

\draw (0,-1) rectangle (1,-2);
\draw (1,-1) rectangle (2,-2);
\draw (2,-1) rectangle (4,-2);
\draw (4,-1) rectangle (8,-2);

\draw (0,-2) rectangle (1,-4);
\draw (1,-2) rectangle (2,-4);
\draw (2,-2) rectangle (4,-4);
\draw (4,-2) rectangle (8,-4);

\draw (0,-4) rectangle (1,-8);
\draw (1,-4) rectangle (2,-8);
\draw (2,-4) rectangle (4,-8);
\draw (4,-4) rectangle (8,-8);

\draw (0,1) -- (8,1);
\draw (0,1) node {$\bullet$} ;
\draw (1,1) node {$\bullet$} ;
\draw (2,1) node {$\bullet$} ;
\draw (4,1) node {$\bullet$} ;
\draw (8,1) node {$\bullet$} ;

\draw (0,1.2) node[above] {$1$} ;
\draw (8,1.2) node[above] {$\sqrt{n}$} ;

\draw[decoration={brace,raise=5pt},decorate]
  (0,1.7) -- node[above=6pt] {$\Omega_0$} (1,1.7);
\draw[decoration={brace,raise=5pt},decorate]
  (1,1.7) -- node[above=6pt] {$\Omega_1$} (2,1.7);
 \draw[decoration={brace,raise=5pt},decorate]
  (4,1.7) -- node[above=6pt] {$\Omega_J$} (8,1.7);

\draw (-1,0) -- (-1,-8);
\draw (-1,0) node {$\bullet$} ;
\draw (-1,-1) node {$\bullet$} ;
\draw (-1,-2) node {$\bullet$} ;
\draw (-1,-4) node {$\bullet$} ;
\draw (-1,-8) node {$\bullet$} ;

\draw (-1.2,0) node[left] {$1$} ;
\draw (-1.2,-8) node[left] {$\sqrt{n}$} ;

\draw[decoration={brace,mirror,raise=5pt},decorate]
  (-1.7,0) -- node[left=6pt] {$\Omega_0$} (-1.7,-1);
\draw[decoration={brace,mirror,raise=5pt},decorate]
  (-1.7,-1) -- node[left=6pt] {$\Omega_1$} (-1.7,-2);
\draw[decoration={brace,mirror,raise=5pt},decorate]
  (-1.7,-4) -- node[left=6pt] {$\Omega_J$} (-1.7,-8);
  
\draw (0.25,-0.25) node[color=Gray] {\large$\blacksquare$} ;
\draw (0.75,-0.25) node[color=Gray] {\large$\blacksquare$} ;
\draw (0.25,-0.75) node[color=Gray] {\large$\blacksquare$} ;
\draw (0.75,-0.75) node[color=Gray] {\large$\blacksquare$} ;

\draw (0.25,-0.25) node[color=RubineRed,opacity=0.2] {\large$\blacksquare$} ;
\draw (0.75,-0.25) node[color=RubineRed,opacity=0.2] {\large$\blacksquare$} ;
\draw (0.25,-0.75) node[color=RubineRed,opacity=0.2] {\large$\blacksquare$} ;
\draw (0.75,-0.75) node[color=RubineRed,opacity=0.2] {\large$\blacksquare$} ;

\draw (1.25,-0.25) node[color=Gray] {\large$\blacksquare$} ;
\draw (1.75,-0.25) node[color=Gray] {\large$\blacksquare$} ;
\draw (1.25,-0.75) node[color=Gray] {\large$\blacksquare$} ;

\draw (1.25,-0.25) node[color=RubineRed,opacity=0.2] {\large$\blacksquare$} ;
\draw (1.75,-0.25) node[color=RubineRed,opacity=0.2] {\large$\blacksquare$} ;

\draw (2.25,-0.25) node[color=Gray] {\large$\blacksquare$} ;
\draw (2.75,-0.25) node[color=Gray] {\large$\blacksquare$} ;
\draw (3.75,-0.25) node[color=Gray] {\large$\blacksquare$} ;
\draw (2.25,-0.75) node[color=Gray] {\large$\blacksquare$} ;

\draw (2.25,-0.25) node[color=RubineRed,opacity=0.2] {\large$\blacksquare$} ;
\draw (2.75,-0.25) node[color=RubineRed,opacity=0.2] {\large$\blacksquare$} ;
\draw (3.75,-0.25) node[color=RubineRed,opacity=0.2] {\large$\blacksquare$} ;

\draw (0.25,-1.25) node[color=Gray] {\large$\blacksquare$} ;
\draw (0.75,-1.25) node[color=Gray] {\large$\blacksquare$} ;
\draw (0.75,-1.75) node[color=Gray] {\large$\blacksquare$} ;

\draw (0.25,-1.25) node[color=RubineRed,opacity=0.2] {\large$\blacksquare$} ;
\draw (0.75,-1.25) node[color=RubineRed,opacity=0.2] {\large$\blacksquare$} ;

\draw (1.25,-1.25) node[color=Gray] {\large$\blacksquare$} ;
\draw (1.75,-1.75) node[color=Gray] {\large$\blacksquare$} ;

\draw (2.25,-1.25) node[color=Gray] {\large$\blacksquare$} ;
\draw (3.25,-1.25) node[color=Gray] {\large$\blacksquare$} ;
\draw (2.75,-1.75) node[color=Gray] {\large$\blacksquare$} ;
\draw (3.75,-1.75) node[color=Gray] {\large$\blacksquare$} ;
\draw (3.25,-1.75) node[color=Gray] {\large$\blacksquare$} ;
\draw (3.75,-1.75) node[color=RubineRed,opacity=0.2] {\large$\blacksquare$} ;
\draw (3.25,-1.75) node[color=RubineRed,opacity=0.2] {\large$\blacksquare$} ;
\draw (2.75,-1.75) node[color=RubineRed,opacity=0.2] {\large$\blacksquare$} ;

\draw (0.25,-2.25) node[color=Gray] {\large$\blacksquare$} ;
\draw (0.75,-2.75) node[color=Gray] {\large$\blacksquare$} ;
\draw (0.25,-3.25) node[color=Gray] {\large$\blacksquare$} ;

\draw (0.25,-4.25) node[color=Gray] {\large$\blacksquare$} ;
\draw (0.75,-6.75) node[color=Gray] {\large$\blacksquare$} ;
\draw (0.25,-7.75) node[color=Gray] {\large$\blacksquare$} ;

\draw (1.75,-2.25) node[color=Gray] {\large$\blacksquare$} ;
\draw (1.25,-2.75) node[color=Gray] {\large$\blacksquare$} ;
\draw (1.75,-3.75) node[color=Gray] {\large$\blacksquare$} ;

\draw (1.25,-4.75) node[color=Gray] {\large$\blacksquare$} ;
\draw (1.75,-5.75) node[color=Gray] {\large$\blacksquare$} ;
\draw (1.25,-6.25) node[color=Gray] {\large$\blacksquare$} ;

\draw (2.25,-2.25) node[color=Gray] {\large$\blacksquare$} ;
\draw (3.75,-3.25) node[color=Gray] {\large$\blacksquare$} ;
\draw (2.75,-3.25) node[color=Gray] {\large$\blacksquare$} ;

\draw (2.75,-4.75) node[color=Gray] {\large$\blacksquare$} ;
\draw (3.75,-5.75) node[color=Gray] {\large$\blacksquare$} ;
\draw (3.25,-7.25) node[color=Gray] {\large$\blacksquare$} ;


\draw (6,-8.7) node[color=RubineRed] {$\dots$};
\draw (0.5,-8.9) node[rotate=45,color=RubineRed]{$s_0^r=2$};
\draw (1.5,-8.9) node[rotate=45,color=RubineRed]{$s_1^r=2$};
\draw (3,-8.9) node[rotate=45,color=RubineRed]{$s_2^r=3$};

\end{tikzpicture}
 \caption{\label{fig:struct_sparsity_ex}Example of a signal $x\in \Cbb^n$ with support denoted by squares in grey for the first vertical subbands. The quantities $(s_j^r)_j$  are the maximal degrees of sparsity restricted to rows and respectively to vertical subbands $\cup_{j'} \Omega_{j',j}$. In the Figure, they are reached for support locations highlighted in red.
 }
  \end{center}
 \end{figure}
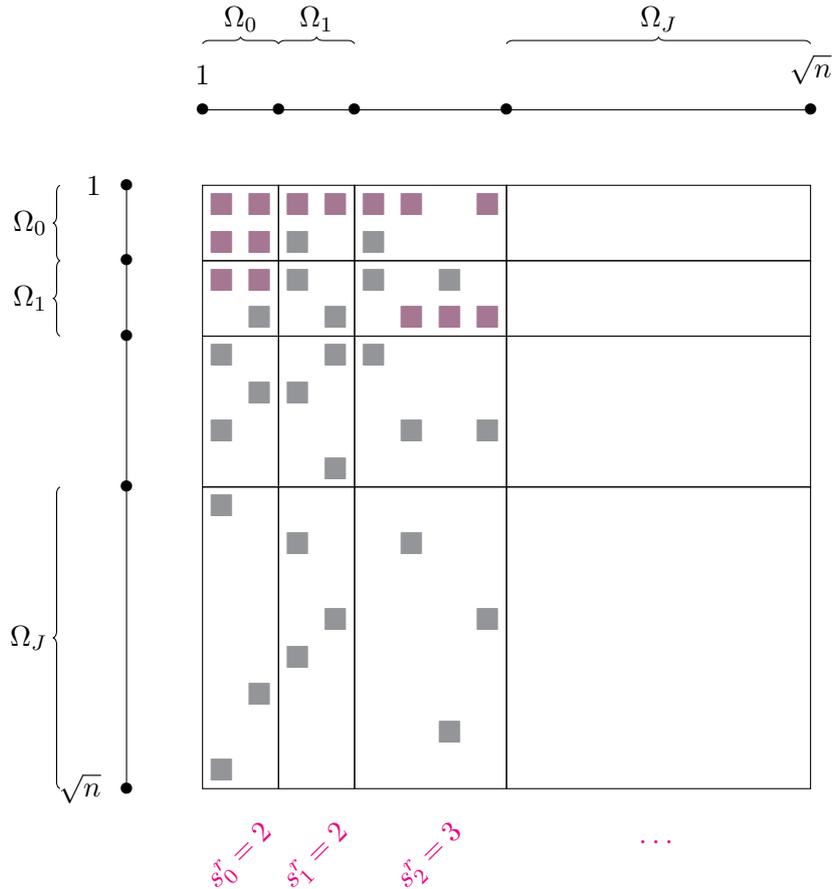

{Equipped with this notion of structured sparsity, we consider measurements corresponding to vertical lines in the frequency space. This sensing strategy models MRI acquisition in a more realistic way than isolated Fourier measurements and corresponds to the setting \ref{block_finite_setting} in Section~\ref{subsec:sampling}. Specifically, we partition the set of $n$ rows of the matrix $A_0$ into $\sqrt{n}$ blocks of size $\sqrt{n} \times n$ defined for $k = 1,\ldots,\sqrt{n}$ as
$$
D_{k} := \phi_{k,:} \otimes \phi = \begin{pmatrix} \phi_{k,1} \phi | \cdots | \phi_{k,\sqrt{n}} \phi\end{pmatrix} \in \mathbb{C}^{\sqrt{n} \times n},
$$
where $\phi_{k,:}$ denotes the $k$-th row of $\phi$. Drawing the block $D_k$ corresponds to sampling along the vertical line $\{1,\ldots,\sqrt{n}\}\times\{k\}$ in the frequency space. 

In this setting, the analysis based on the quantity $\Theta(S,\pi)$ leads to the following result.}

\begin{corollary}
\label{corol:linesHaarMRI}
Let $A_0\in \Cbb^{n \times n}$ be the two-dimensional Fourier-Haar transform and consider the splitting of the Haar and of the Fourier spaces into subbands as defined above.
Fix $\varepsilon \in (0,1)$. 
Suppose that $x \in \Cbb^{n}$ is an $S$-sparse vector with random signs and associated structured sparsities {in levels} $(s_j^{r})_{0\leq j \leq J}$.
Then, {the same} recovery guarantees of Theorem \ref{thm:noisy} hold with probability at least $1-\varepsilon$ by
drawing vertical lines \replace{with probability}{according to the probability distribution} $\left( \pi_k \right)_{1\leq k \leq \sqrt{n}}$ defined as 
$$ {\pi}_k = \frac{2^{-j(k)} \sum_{j=0}^J 2^{-|j(k)-j|/2} s^{r}_{j}}{ 
\sum_{\ell=1}^{\sqrt{n}}  
2^{-j(\ell)} \sum_{j=0}^J 
2^{-|j(\ell) -j|/2} s^{r}_{j}  },
$$
with a number $m$ of drawn vertical lines of the order
\begin{equation}
\label{eq:horLinesMRI}
m \gtrsim  \sum_{j=0}^J \left(  s^{r}_{j}  + \sum_{j'=0 \atop j'\neq j}^J 
2^{-|j -j'|/2} s^{r}_{j'}  \right) \cdot\ln^2 \left( \frac{6 n }{\varepsilon}\right).
\end{equation}
{In particular, the probability \replace{}{distribution} $\left( \pi_k \right)_{1\leq k \leq \sqrt{n}}$ is constant on each frequency subband $W_j$.}
\end{corollary}

The proof of this results is given in Appendix~\ref{proof:corollinesHaarMRI_lambda}.
Corollary \ref{corol:linesHaarMRI} meets the results of \cite[Corollary 4.10]{boyer2017compressed}: the bounds on the number of blocks of measurements are of the same order. Nevertheless, note that Corollary \ref{corol:linesHaarMRI} requires an additional assumption on the sign randomness. Moreover, recall that Theorem~\ref{thm:noisy} also ensures stable and robust recovery.

We can improve this result by analyzing the quantity $\Lambda(S,\pi)$.

\begin{corollary}
\label{corol:linesHaarMRI_lambda}
Let $A_0 \in \Cbb^{n \times n}$ be the two-dimensional Fourier-Haar transform and consider the splitting of the Haar and of the Fourier spaces into subbands as defined above. Fix $\varepsilon \in (0,1)$. 
Suppose that $x \in \Cbb^{n}$ is an $S$-sparse vector with random signs and associated structured sparsities $(s_j^{r})_{0\leq j \leq J}$.
 If
$$
\min_{0\leq j \leq J}    \left(  s_{j}^{r} +  \sum_{j' = 0 \atop j'\neq j}^J   2^{-|j-j'|} s_{j'}^{r} \right) \cdot \log_2(\sqrt{n})  \gtrsim \ln(3 n/\varepsilon),
$$
then, the same recovery guarantees of Theorem \ref{thm:killthetalog} hold with probability at least $1-\varepsilon$ by
drawing vertical lines \replace{with probability}{according to the probability distribution} $\left( \pi_k \right)_{1\leq k \leq \sqrt{n}}$ defined as 
$$ {\pi}_k 
= 
{\frac{2^{-j(k)} \sum_{j = 0}^{J} 2^{-|j-j(k)|} s_{j}^r}{\sum_{\ell = 1}^{\sqrt{n}} 2^{-j(\ell)} \sum_{j = 0}^{J} 2^{-|j(\ell)-j|} s_{j}^r}},
$$
with a number $m$ of drawn vertical lines of the order
\begin{equation}
\label{eq:horLinesMRI_lambda}
m \gtrsim  \left( \sum_{j=0}^J  s_{j}^{r}  +\sum_{j'=0 \atop j'\neq j}^J   2^{-|j-j'|} s_{j'}^{r} \right)\cdot \log_2(\sqrt{n})  \cdot \ln \left( \frac{6 n }{\varepsilon}\right). 
\end{equation}
In particular, the probability \replace{}{distribution} $\left( \pi_k \right)_{1\leq k \leq \sqrt{n}}$ is constant on each frequency subband $W_j$.
\end{corollary}

The proof of Corollary~\ref{corol:linesHaarMRI_lambda} is given in Appendix~\ref{proof:corollinesHaarMRI_lambda}.

Note that the bound \eqref{eq:horLinesMRI_lambda} improves Corollary~\ref{corol:linesHaarMRI} by  attenuating the interference between different subbands sparsities.

\begin{remark}
Results analogous to those presented in this Section hold when sampling horizontal lines in the Fourier space, up to replacing the structured sparsities $s_j^r$ in Definition~\ref{def:sparsity_in_levels_2D} with 
$$
s_j^c := \max_{0 \leq \ell \leq J} \max_{k \in \Omega_\ell} |S \cap \Omega_{j,\ell} \cap C_k|, 
$$
where $C_k = \{1,\ldots,\sqrt{n}\}\times \{k\}$ is the $k$-th vertical line of the wavelet multi-index space. In this case, sampling along the $k$-th horizontal line corresponds to draw a block $D_k = \phi \otimes \phi_{k,:}$. The proofs in Appendix~\ref{proof:corollinesHaarMRI_lambda} can be easily adapted to this setting by interchanging the role of rows and columns in the Haar and in the Fourier multi-index spaces, respectively (or, equivalently, by switching the order of tensorization).
\end{remark}
\begin{remark} 
There is a mismatch between the notation adopted here and  in \cite{boyer2017compressed}, where the role of rows and columns in the Haar and Fourier spaces is switched. As a result, horizontal Fourier sampling and structured sparsities $s_j^c$ in \cite{boyer2017compressed} correspond to vertical Fourier sampling and structured sparsities $s_j^r$, respectively, in our framework. Although choosing one of these two conventions might be considered just a matter of taste, we have changed notation in order to adhere to the way the two-dimensional Fourier transform \replace{in}{is} performed by the \textsc{Matlab}$^{\tiny\textregistered}$ command \texttt{fft2}. In our setting, fully-sampled measurements $y = A_0 x = (\mathcal{F}H^* \otimes \mathcal{F}H^*)x$ correspond to $Y = \mathcal{F}H^* X (\mathcal{F}H^*)^T$, where $x = \vect(X)$ and $y = \vect(Y)$. In particular, the block measurement $D_k x = ((\mathcal{F}H^*)_{k,:}\otimes \mathcal{F}H^*)x$ corresponds to the $k$-th column (or vertical line) of the matrix $Y$. 
\end{remark}


\subsection{Adaptive sampling for function approximation}
\label{sec:adaptive}

In this section, we examine the problem of approximating a function from random pointwise samples using standard sparsity with respect to Legendre polynomials \cite{adcock2017compressed,chkifa2017polynomial,rauhut2012sparse}. In particular, we will see how the theoretical results shown in this paper can be employed to construct effective adaptive sampling strategies. For the sake of simplicity, we will focus on the one-dimensional case, although the strategies presented here can be generalized to multiple dimensions.

\paragraph{Adaptive sampling strategies}
Let $f : [-1,1] \to \mathbb{C}$ and let $\{L_j\}_{j = 1}^{n}$ be the family of Legendre orthogonal polynomials normalized such that $\int_{-1}^1L_j(x)L_k(x) dx = 2 \delta_{jk}$. We aim at approximating $f$ as a sparse expansion of Legendre polynomials from a fixed budget of $m$ adaptively chosen pointwise samples.

In order to put ourselves in the framework of subsampled isometries, we consider the family of Gauss-Legendre quadrature points $\{g_j\}_{j = 1}^n$ on $[-1,1]$ and their respective quadrature weights $\{w_j\}_{j = 1}^n$ (we recall that $\{g_j\}_{j = 1}^n$ are the roots of the polynomial $L_{n+1}$ and the weights satisfy $\sum_{j = 1}^n w_j = 2$). The resulting quadrature formula is exact on polynomials of degree less than or equal to $2n-1$. In particular, the matrix $A_0 \in \mathbb{R}^{n \times n}$, defined as
\begin{equation}
(A_0)_{ij} = (\replacemath{a}{d}_1|\cdots|\replacemath{a}{d}_n)^*, \quad \replacemath{a}{d}_i = \sqrt{\frac{w_i}{2}}( L_j(g_i))_{j = 1}^n, 
\end{equation} 
is orthogonal.  Based on the theory presented in this paper, we consider two adaptive sampling strategies based on successive approximations of the support via $\ell^1$ minimization. We refer to these strategies as (Adapt~I) and (Adapt~II). They are outlined in Algorithm~\ref{adapt_sampl_I} and described below.

Let us fix a target sparsity level $s \leq n$ and two numbers $K,m_1 \in \mathbb{N}$ such that $s \leq m_1 \leq N/K$. Both procedures draw $m_1$ samples $K$ times, resulting in a total of $m = K m_1$ samples. At each iteration, we compute an approximation to the function $f$ based on partial support information and we update the sampling measure accordingly. In particular, in the $k$-th iteration, (Adapt~I) updates the sampling measure based on the support corresponding to the $s$ entries of the $(k-1)$-th approximation with largest magnitude. On the other hand, (Adapt~II) updates the measure by taking advantage of the $(k-1) s$ entries of the $(k-1)$-th approximation having largest magnitude. This difference corresponds to line \ref{support_update} of  Algorithm~\ref{adapt_sampl_I}.

\begin{algorithm}[h]
\caption{\label{adapt_sampl_I}Adaptive sampling strategies (Adapt~I) and (Adapt~II)}
\textbf{Inputs:} $s,n,K,m_1 \in \mathbb{N}$, $\eta >0$\\
\textbf{Output:} $\hat{x} \in \mathbb{R}^n$

\begin{algorithmic}[1]
\STATE Draw $m_1$ samples $\{J_\ell^{(1)}\}_{\ell =1}^{m_1}$ i.i.d.\ uniformly from $\{1,\ldots,N\}$ (i.e., $\pi_j^{(1)} = 1/N$)
\STATE Define $A^{(1)} = \sqrt{\frac{n}{m_1}} (\replacemath{a}{d}_{J_\replacemath{i}{\ell}^{(1)}})_{\replacemath{i}{\ell}=1}^{m_1}\in \mathbb{R}^{m_1 \times n}$ and $y^{(1)} = \sqrt{\frac{n}{m_1}} (\sqrt{\frac{w_\replacemath{i}{\ell}}{2}} f(g_{J_\replacemath{i}{\ell}^{(1)}}))_{\replacemath{i}{\ell}=1}^{m_1}\in \mathbb{R}^{m_1}$
\STATE Find $\hat{x}^{(1)} := \displaystyle\arg \min_{z \in \mathbb{R}^n}\|z\|_1$ s.t. $\|A^{(1)}z -y^{(1)}\|_2 \leq \eta$
\FOR{$k = 2,\ldots, K$} 
\STATE \label{support_update}Let $\displaystyle S^{(k-1)} = 
\begin{cases}
\arg\min_{S : |S| = s}\|\hat{x}^{(k-1)}_S - \hat{x}^{(k-1)}\|_2, 
& \textnormal{in the case of (Adapt~I)}\\ 
\arg\min_{S : |S| = (k-1)s}\|\hat{x}^{(k-1)}_S - \hat{x}^{(k-1)}\|_2, 
& \textnormal{in the case of (Adapt~II)}
\end{cases}$
\STATE Define the measure $\pi^{(k)}_j \propto \|\replacemath{a}{d}_{j,S^{(k-1)}}\|_2^2$ for every $j =1,\ldots,N$
\STATE Draw $\replace{m_1}{m_k}$ indices $\{J_\ell^{(k)}\}_{\ell=1}^{m_1}$  i.i.d.\ according to $\pi^{(k)}$ from $\{1,\ldots,N\}$
\STATE Define $A^{(k)} = ( \replacemath{a}{d}_{J_\replacemath{i}{\ell}^{(k)}}  c_\replacemath{i}{\ell}^{(k)})_{\replacemath{i}{\ell}=1}^{\replacemath{m_1}{m_k}}$ and $y^{(k)} = (\sqrt{\frac{w_\replacemath{i}{\ell}}{2}} f(g_{J_\replacemath{i}{\ell}^{(k)}}) c_\replacemath{i}{\ell}^{(k)})_{\replacemath{i}{\ell}=1}^{\replacemath{m_1}{m_k}}$, with $c_\replacemath{i}{\ell}^{(k)}=1/\sqrt{m_1 \pi^{(k)}_{J_\replacemath{i}{\ell}^{(k)}}}$
\STATE Define $\tilde{A}^{(k)} = \frac{1}{\sqrt{k}}\begin{bmatrix}A^{(1)}\\ \vdots\\A^{(k)}\end{bmatrix}$ and $\tilde{y}^{(k)} = \frac{1}{\sqrt{k}}\begin{bmatrix}y^{(1)}\\ \vdots\\y^{(k)}\end{bmatrix}$
\STATE Find $\hat{x}^{(k)} := \displaystyle\arg \min_{z \in \mathbb{R}^n}\|z\|_1$ s.t. $\|\tilde{A}^{(k)}z -\tilde{y}^{(k)}\|_2 \leq \eta$
\ENDFOR
\STATE Let $\hat{x} = \hat{x}^{(K)}$
\end{algorithmic}
\end{algorithm}

\paragraph{Adaptive \emph{vs.}\ nonadaptive sampling.} In order to evaluate the performance of adaptive sampling, we compare the following five sampling strategies:
\begin{description}
\item [(Adapt~I)] Adaptive sampling with fixed support size;
\item [(Adapt~II)] Adaptive sampling with increasing support size;
\item [(Unif I)] Random sampling from the continuous uniform measure over $(-1,1)$;
\item [(Unif II)] Uniform random sampling from $A_0$;
\item [(Cheby)] Random sampling from the continuous Chebyshev measure  $\frac{1}{\pi\sqrt{1-t^2}}$ over $(-1,1)$;

\end{description} 
We fix $n  = 150$ and $s  = 5$. For the strategies (Adapt~I) and (Adapt~II), we chose $m_1 = 15$, $K  = 5$, leading to a total number of $m = K m_1 = 75$ adaptive measurements. We also fix the number of measurements as $m = 75$ for the nonadaptive strategies (Unif I), (Unif II), and (Cheby). Moreover, we fix $\eta = 0$ (using the \texttt{spg\_bp} command of the \textsc{SPGL1} \textsc{Matlab}$^{\tiny\textregistered}$ package to solve basis pursuit \cite{spgl1:2007,BergFriedlander:2008}). We run the following random experiment 100 times. For each run, we randomly generate a 5-sparse Legendre polynomial by selecting 5 indices uniformly at random form $\{1,\ldots,n\}$ and by generating the respective coefficients independently at random according to the distribution $\mathcal{N}(0,1)$. Then, we apply the five strategies to recover the resulting function from pointwise samples.  

The results of this comparison are shown in Figure~\ref{fig:f_approx_comparison}, where we plot the $\ell^2$ error of the approximate solution.
\begin{figure}[t]
\centering
\includegraphics[width = 8cm]{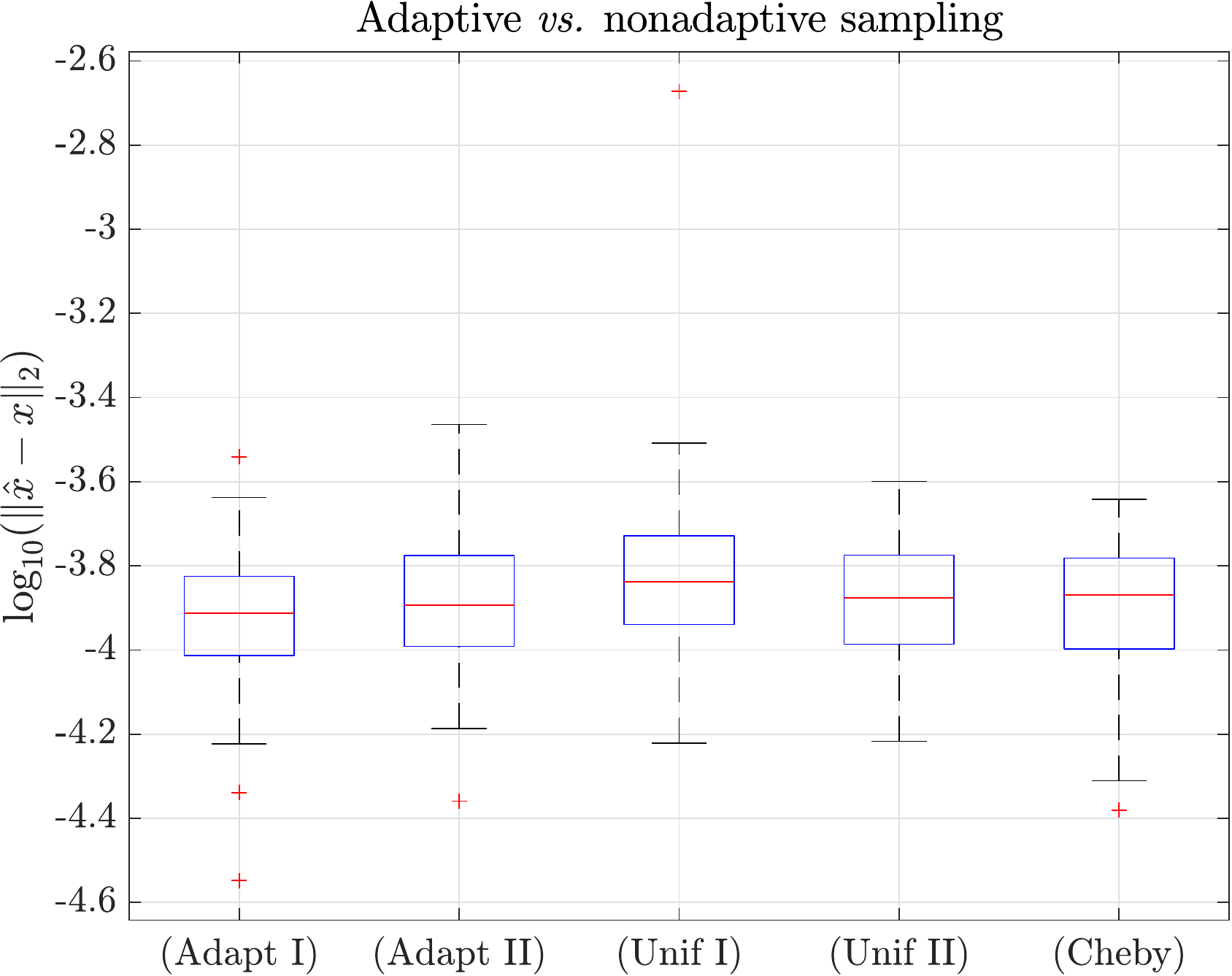}
\caption{\label{fig:f_approx_comparison}Approximation of $f$ using adaptive strategies (Adapt~I), (Adapt~II) and nonadaptive strategies (Unif I), (Unif II), and (Cheby).}
\end{figure}
Adaptive sampling slightly outperforms nonadaptive sampling, but the gain is not substantial.

We also visualize the evolution of the sampling measure $\pi^{(k)}$ as a function of the iteration $k$ in Figure~\ref{fig:evol_measure} for (Adapt~I) and (Adapt~II). 
\begin{figure}
\centering
\includegraphics[width = 7.5cm]{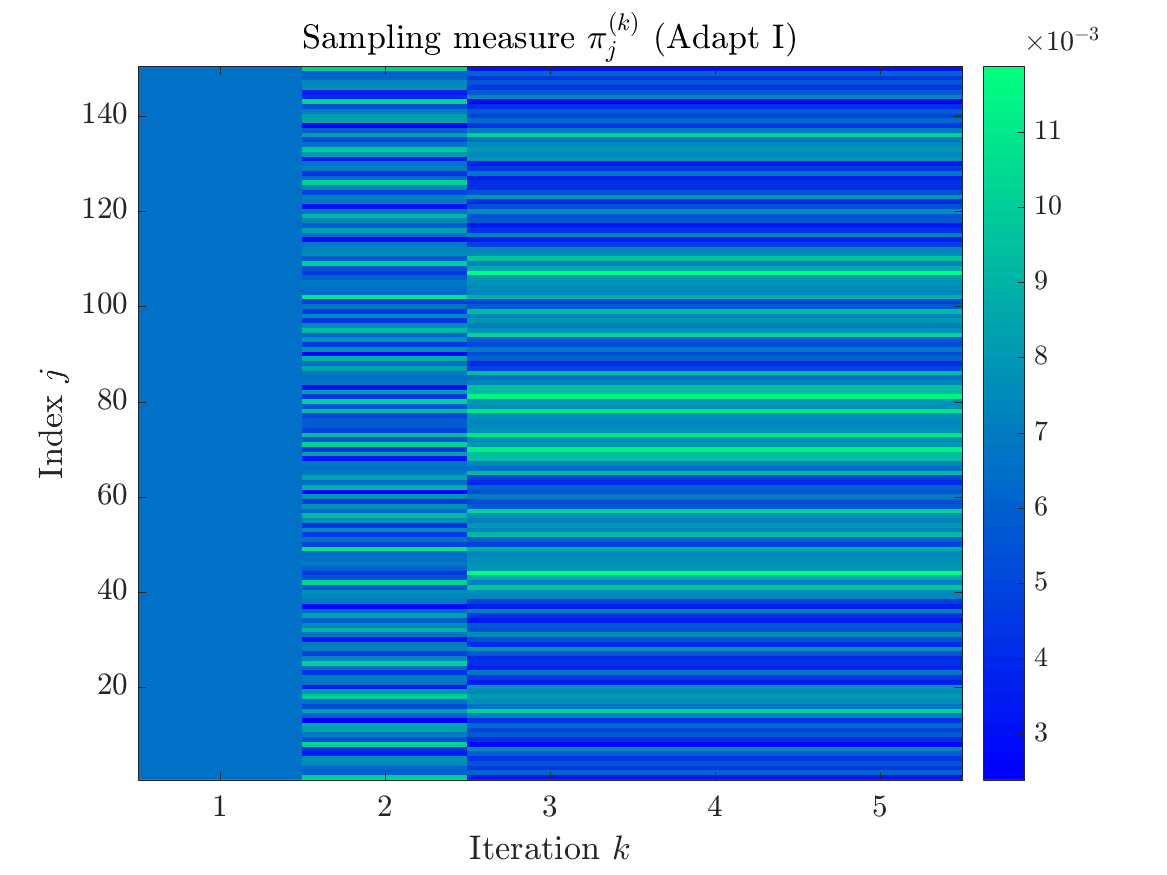}
\includegraphics[width = 7.5cm]{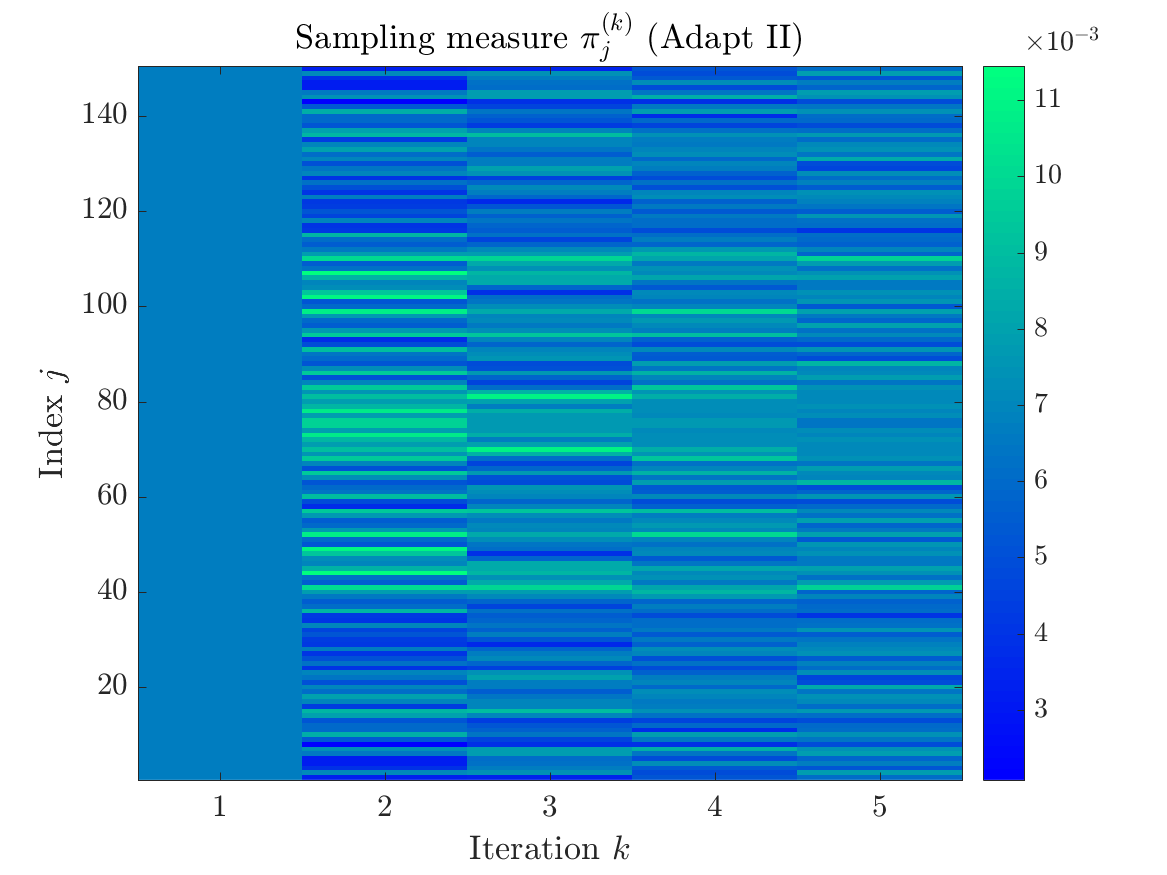}
\caption{\label{fig:evol_measure}Evolution of the sampling measure $\pi^{(k)}$ as a function of the iteration $k$ for adaptive strategies (Adapt~I) and (Adapt~II).}
\end{figure}
We notice that the strategy (Adapt~I) does not refine the measure from the third iteration. On the contrary, the measure is updated at each iteration in the case of (Adapt~II), which takes advantage of more support information.

\paragraph{Convergence of the adapted measure.}

An interesting feature of the proposed adaptive sampling schemes is its capacity to converge to the optimal sampling measure. In the next experiment, we consider a random $5$-sparse combination of the first $n = 100$ Legendre polynomials. We compare the strategies (Adapt~I) and (Adapt~II) with $m_1 = 10$,  $K = 5$, and $s = 5$. We start sampling from a uniform grid on $(-1,1)$ of $n^2 = 10000$ points ($-1$ and $1$ are excluded). Moreover, we fix $\eta = 0$.

In Figure~\ref{fig:convergence} we compare the measure $\pi^{(k)}$  for $k = 1,\ldots, K$ with the Chebyshev measure $\pi_{\text{Cheb}}$ on the uniform grid (i.e., such that $\pi_{\text{Cheb}}(t_j) \propto 1/\sqrt{1-t_j^2}$ for every point $t_j$ of the uniform grid). We can see that for both (Adapt~I) and (Adapt~II) the adapted measure tends to approximate the Chebyshev measure. The convergence is more evident for the approach (Adapt~II), where the approximate support gets larger and larger at each iteration. 
\begin{figure}
\centering
\begin{tabular}{c|c}
(Adapt~I) & (Adapt~II) \\
\hline
\includegraphics[width = 6cm]{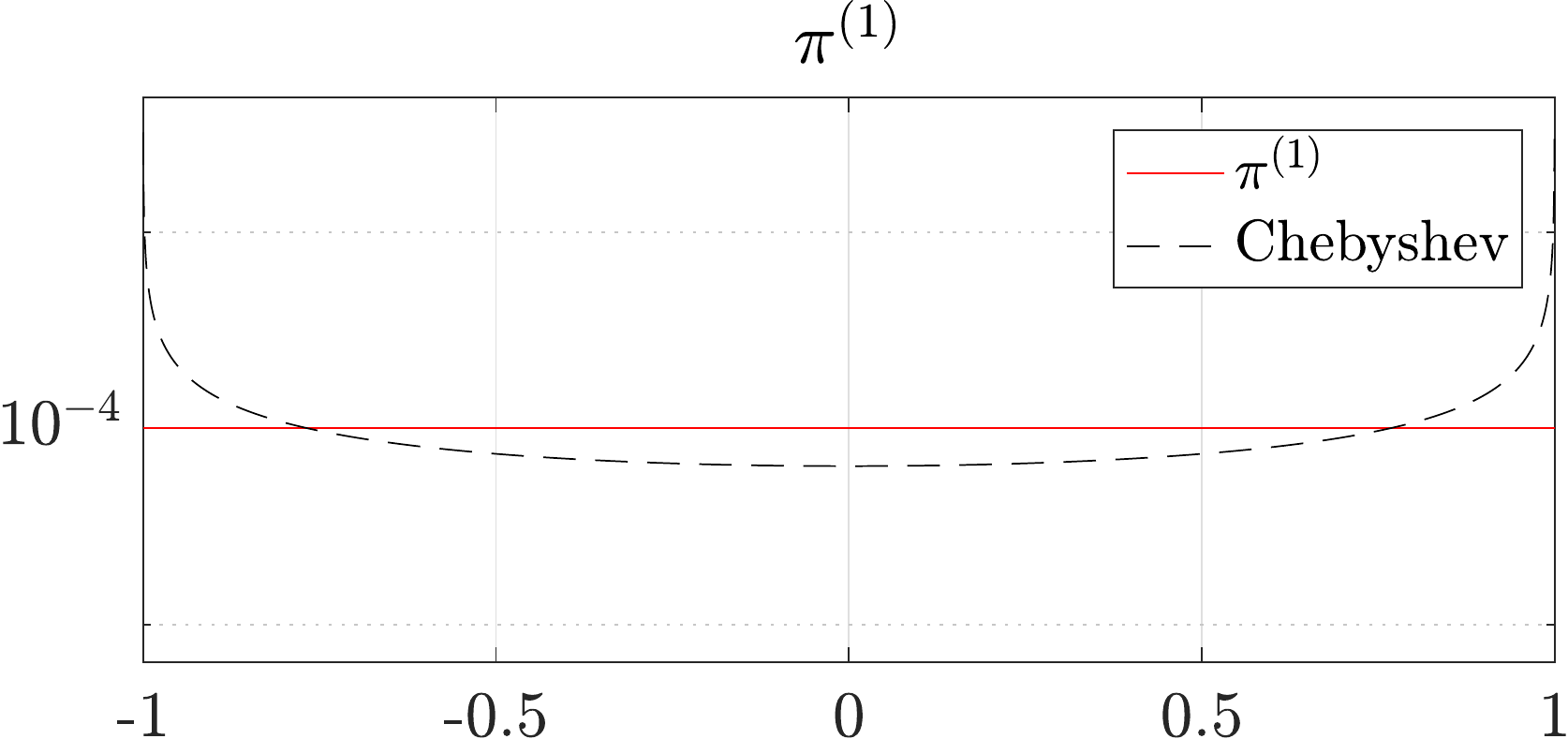} &
\includegraphics[width = 6cm]{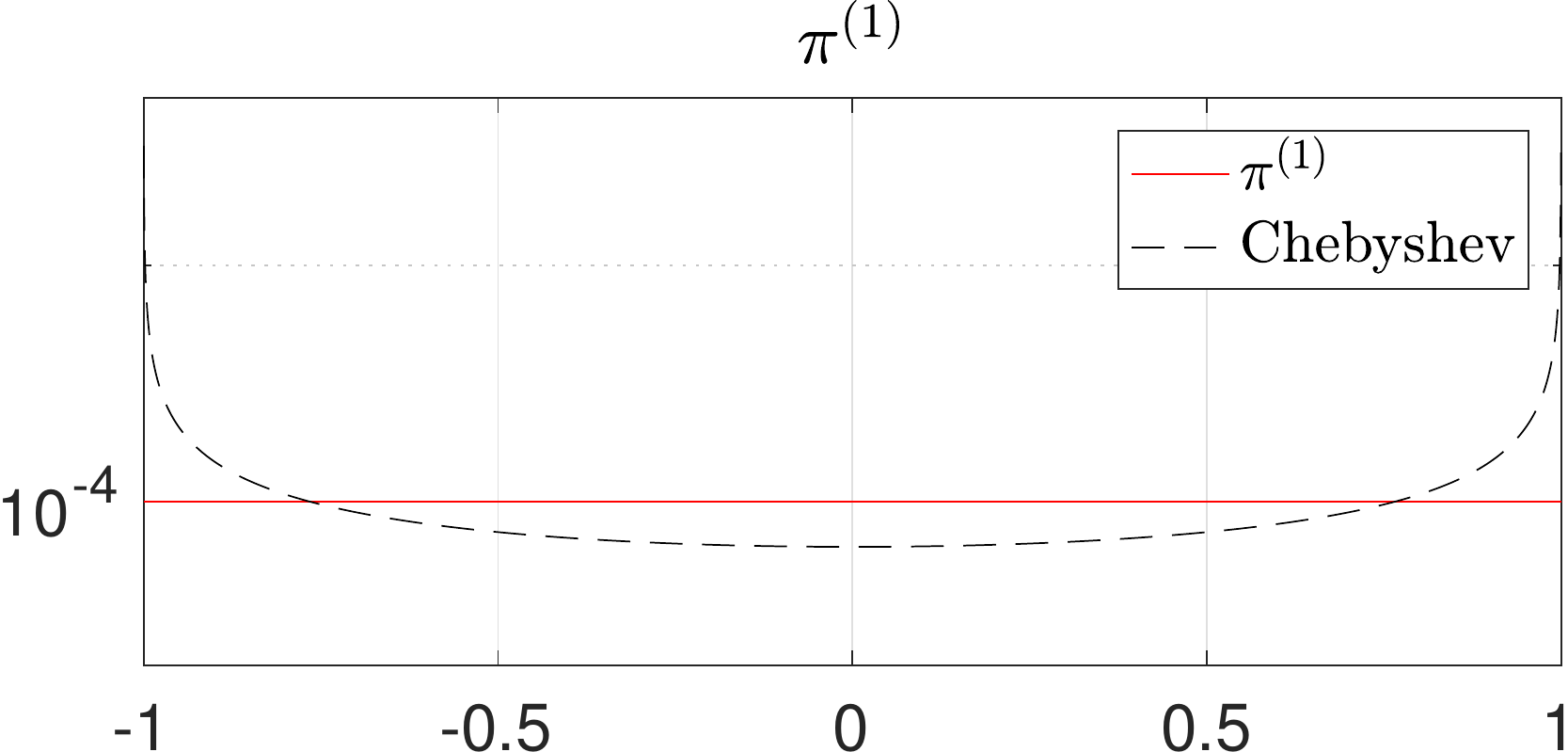} \\
\includegraphics[width = 6cm]{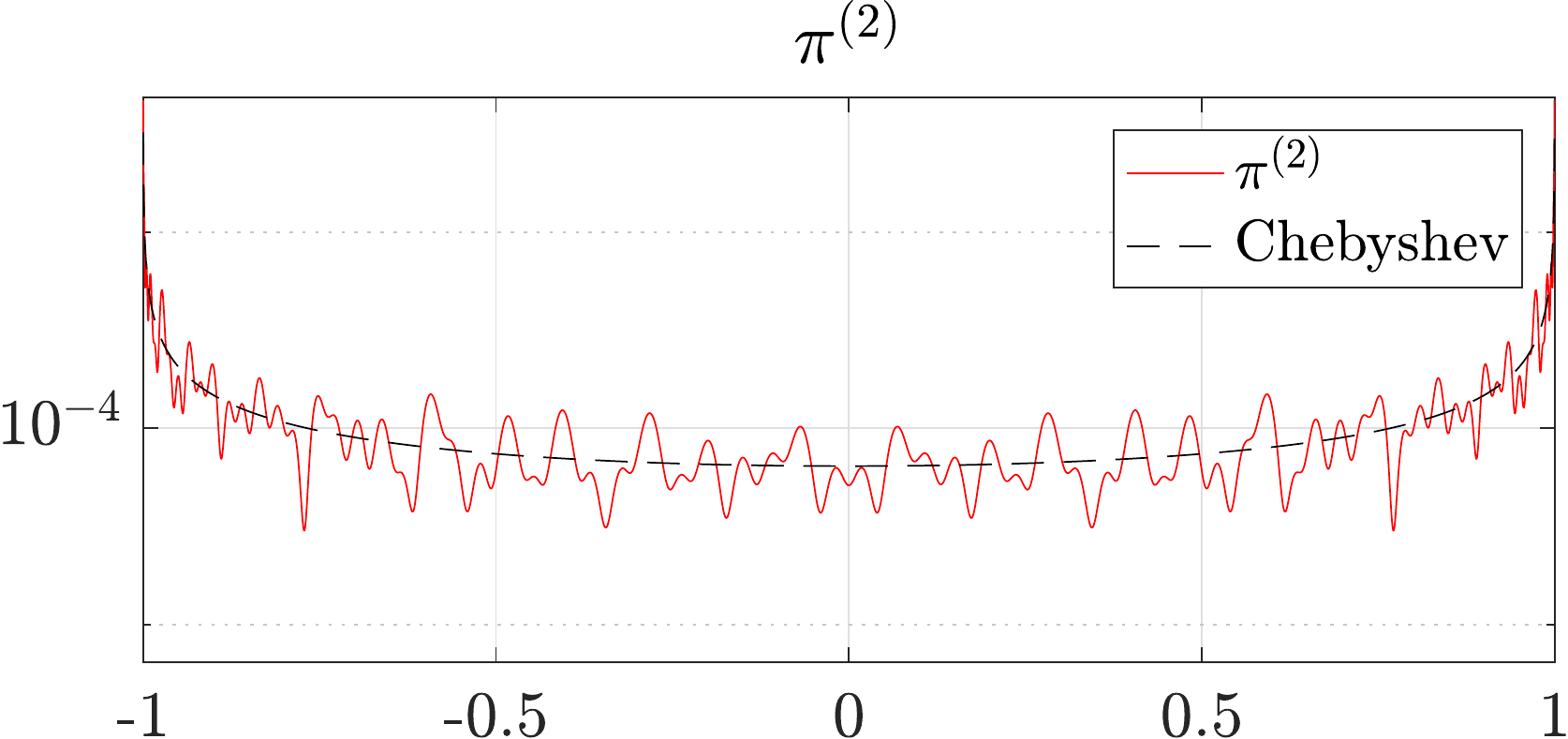} &
\includegraphics[width = 6cm]{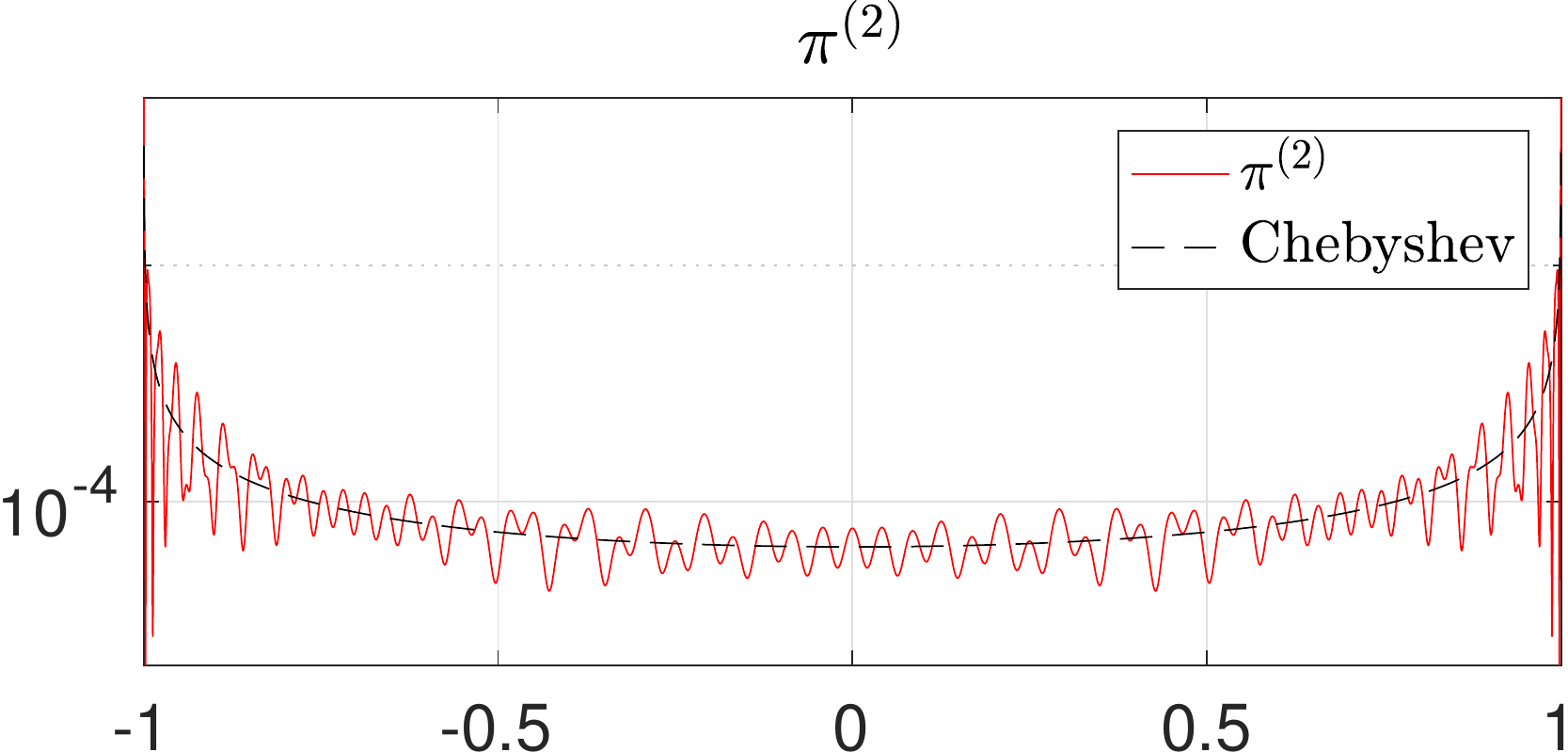} \\
\includegraphics[width = 6cm]{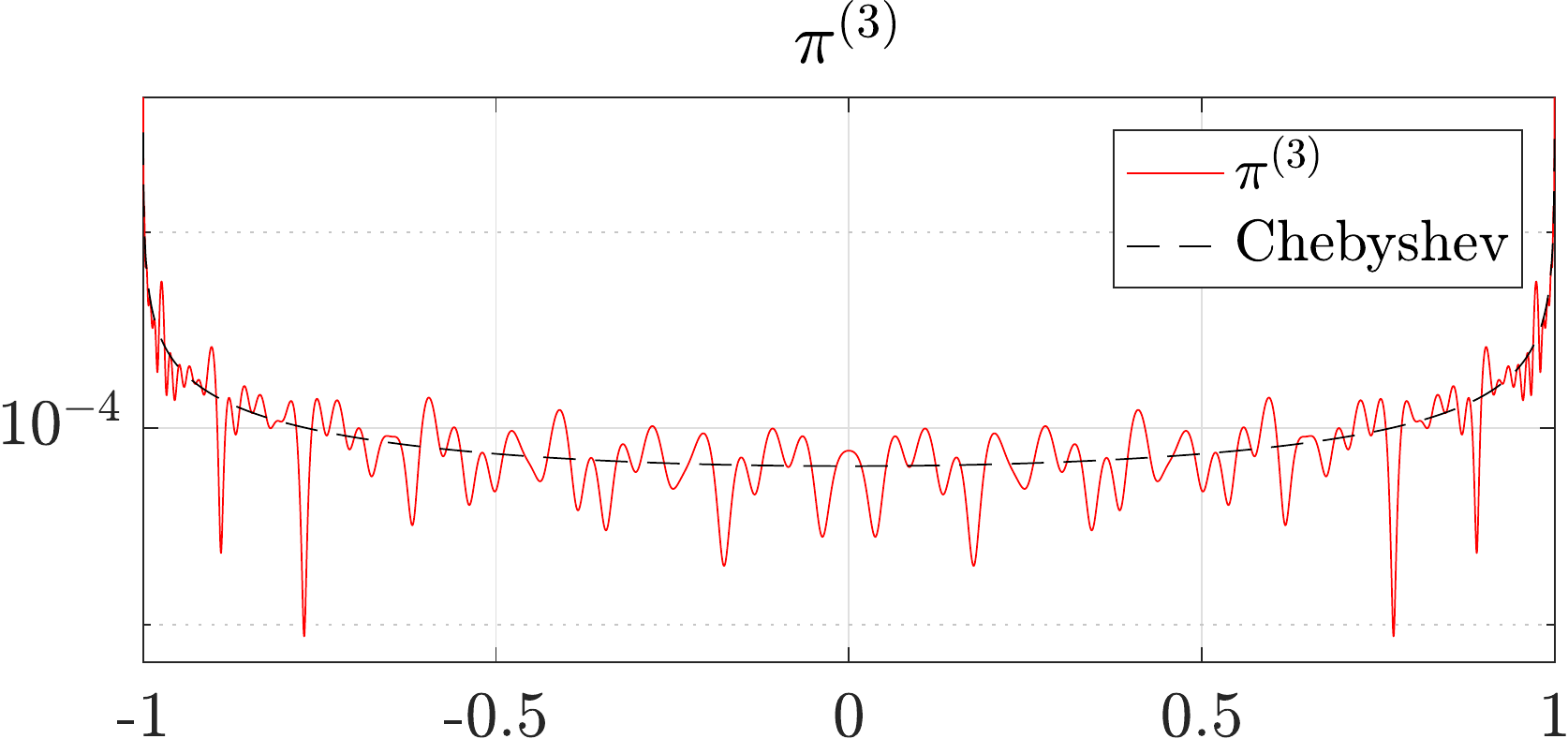} &
\includegraphics[width = 6cm]{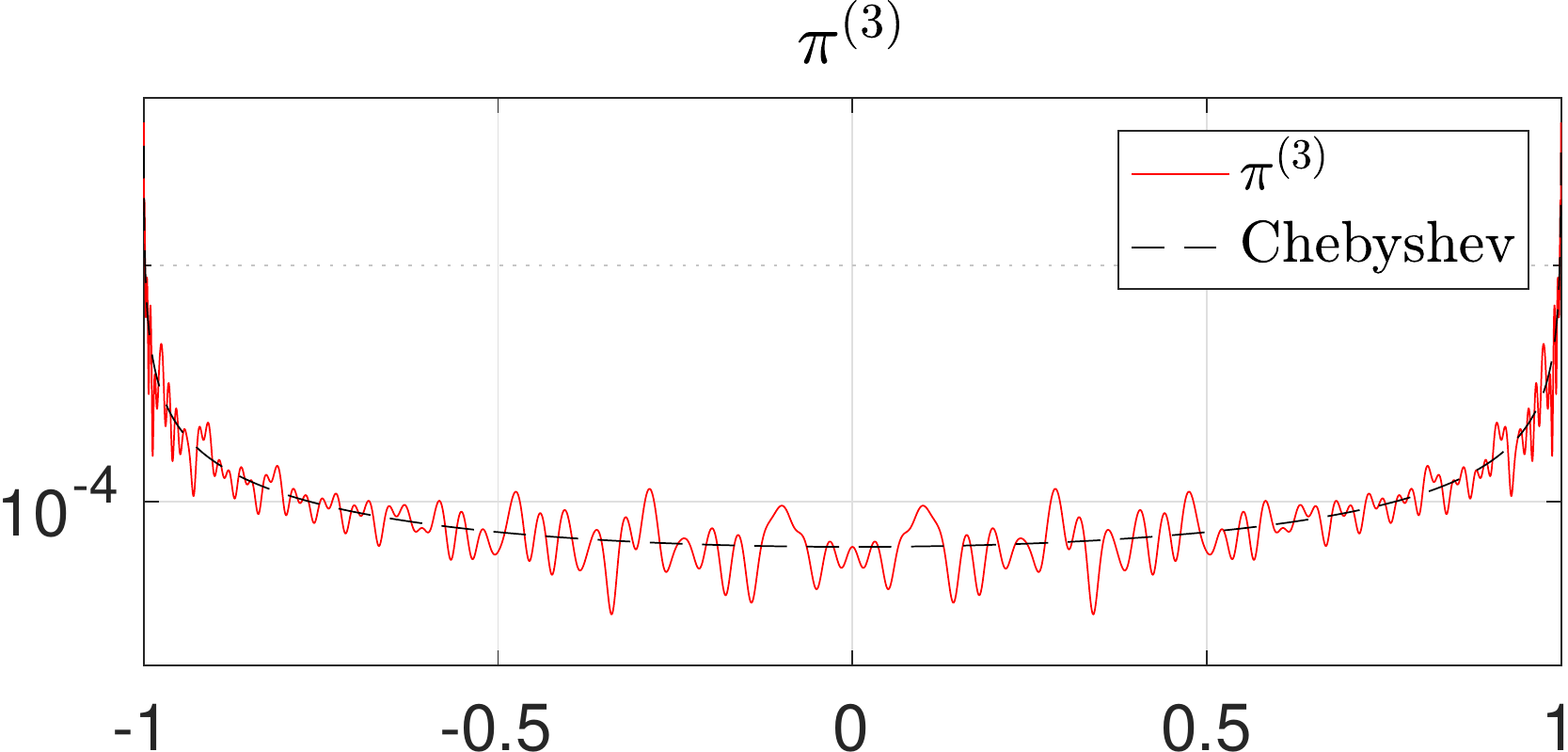} \\
\includegraphics[width = 6cm]{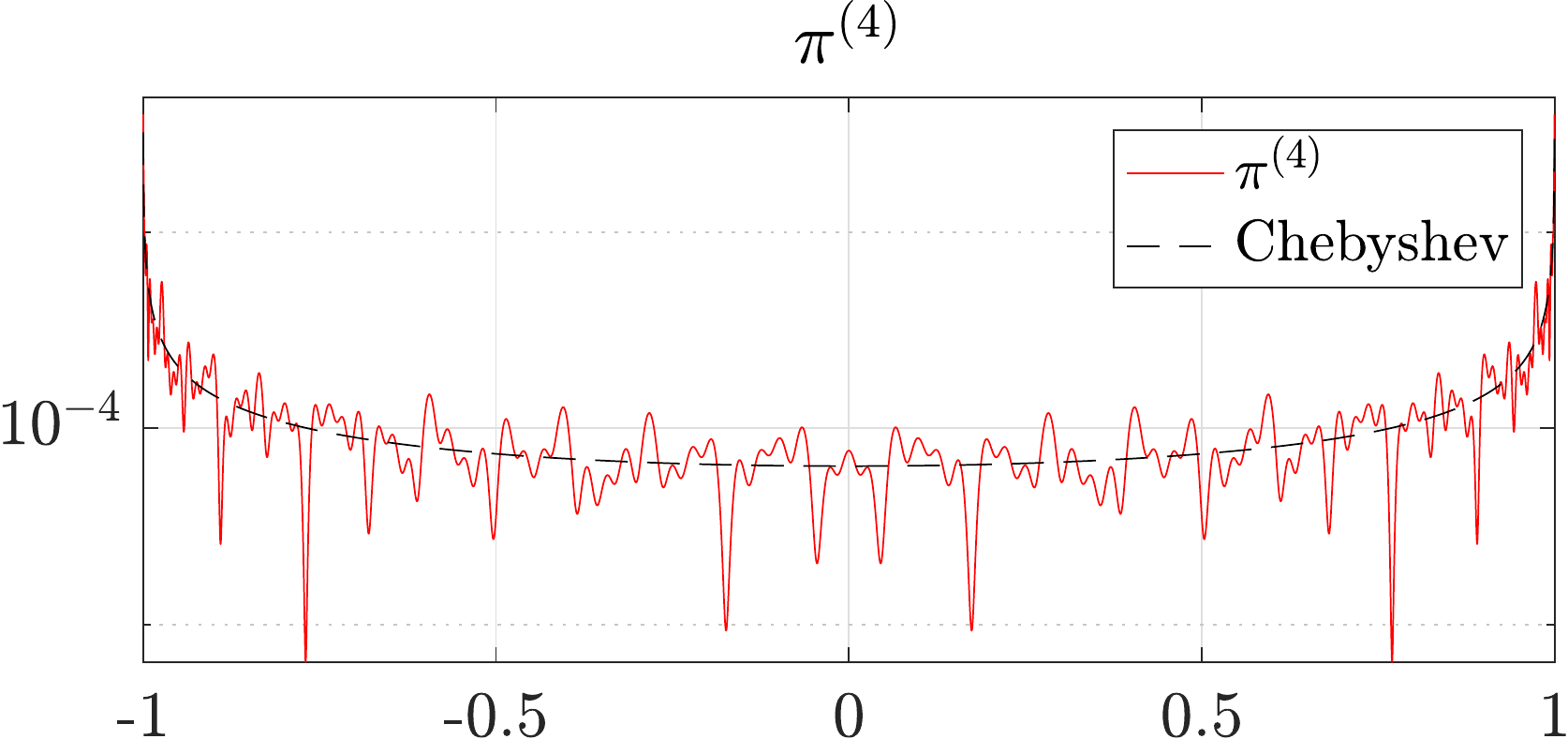} &
\includegraphics[width = 6cm]{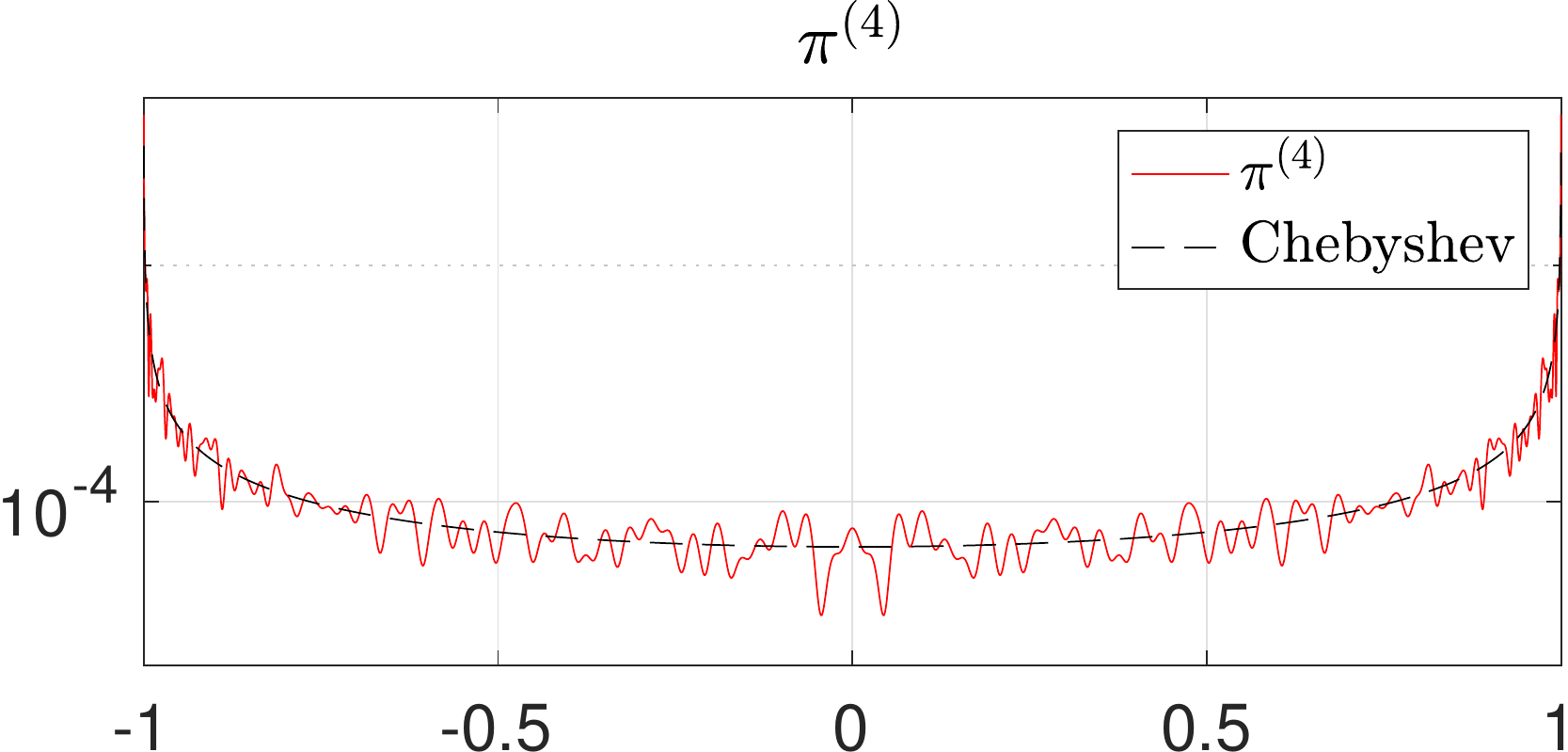} \\
\includegraphics[width = 6cm]{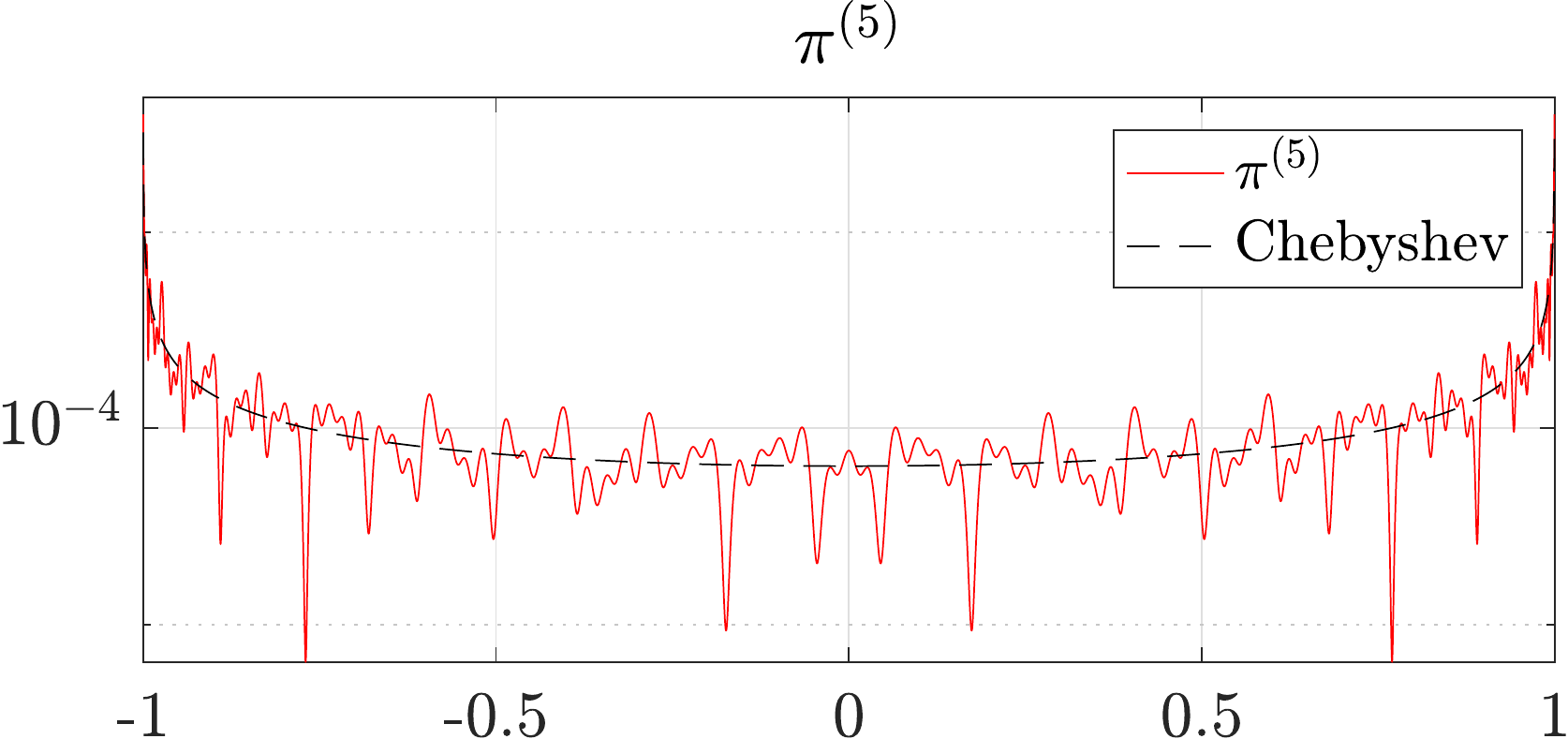} &
\includegraphics[width = 6cm]{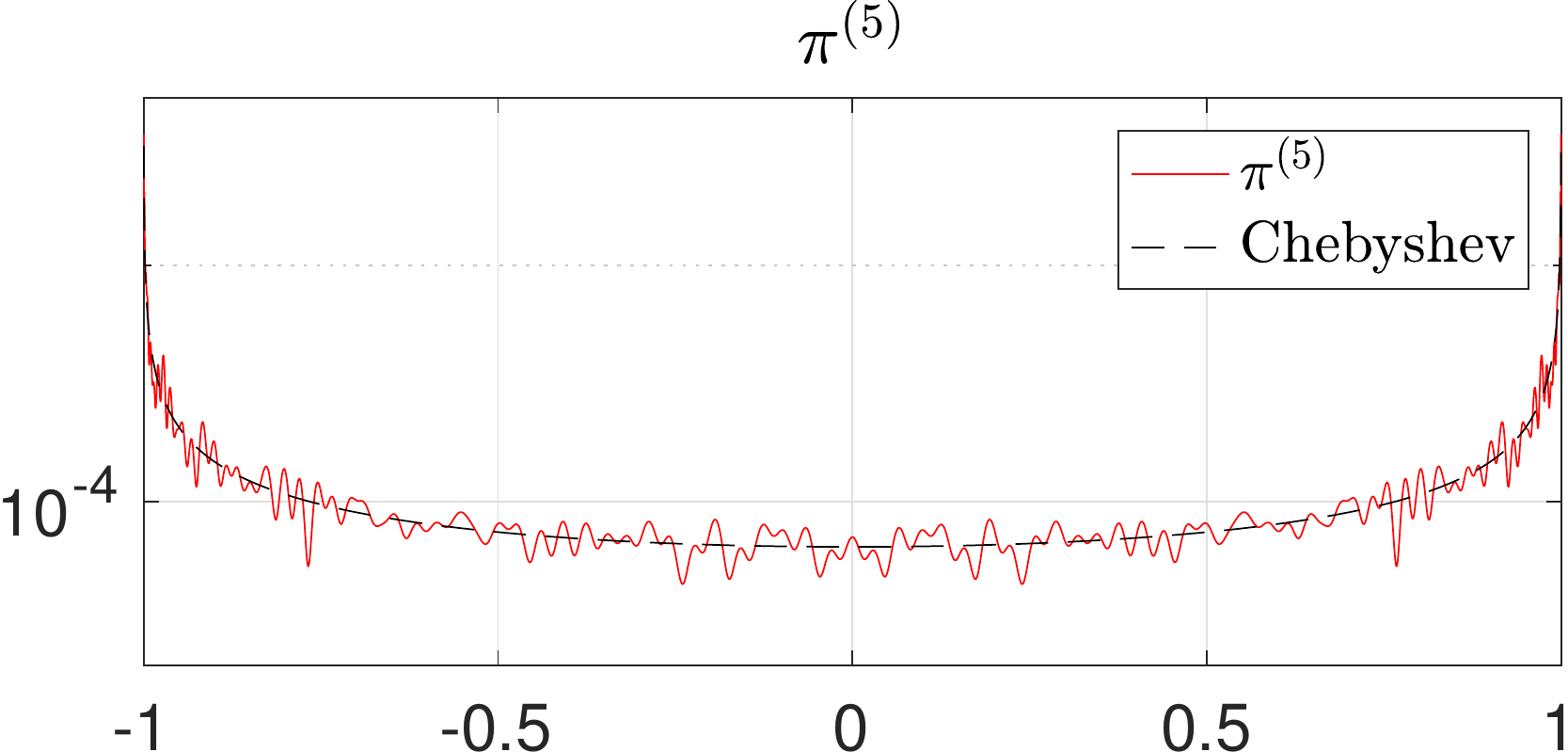} \\
\end{tabular}
\caption{\label{fig:convergence}Convergence of the adapted measure $\pi^{(k)}$ to the Chebyshev measure for (Adapt~I), on the left column, and for (Adapt~II), on the right column.}
\end{figure}

A more quantitative convergence analysis is provided in Table~\ref{tab:convergence}, where we show the absolute error $\|\pi^{(k)}-\pi_{\text{Cheb}}\|_2$ as a function of the iteration $k = 1,\ldots, K$ for both approaches. For (Adapt~I), the error essentially stabilizes from the fourth iteration, whereas in the case of (Adapt~II) it is monotonically decreasing. Comparing these data with Figure~\ref{fig:convergence}, we can see how (Adapt~II) is able to correct the  behavior of $\pi^{(k)}$ near the extrema $\pm1$ to a substantial extent at each iteration. On the contrary, in the case of (Adapt~I) the measure $\pi^{(k)}$ exhibits severe oscillations near $\pm 1$.
\begin{table}
\centering
\begin{tabular}{c|cc}
& (Adapt~I) & (Adapt~II)\\
\hline
$\|\pi^{(1)}- \pi_{\text{Cheb}}\|_2$  & 1.0102e-02 & 1.0102e-02\\
$\|\pi^{(2)}- \pi_{\text{Cheb}}\|_2$  & 4.4442e-03 & 6.3474e-03\\
$\|\pi^{(3)}- \pi_{\text{Cheb}}\|_2$  & 4.2014e-03 & 3.3659e-03\\
$\|\pi^{(4)}- \pi_{\text{Cheb}}\|_2$  & 3.8730e-03 & 2.7679e-03\\
$\|\pi^{(5)}- \pi_{\text{Cheb}}\|_2$  & 3.8730e-03 & 2.3916e-03\\
\end{tabular}
\caption{\label{tab:convergence}Distance between the adapted measure $\pi^{(k)}$ and the Chebyshev measure $\pi_{\text{Cheb}}$ with respect to the $\ell^2$ norm for (Adapt~I) and (Adapt~II).}
\end{table}

The convergence property of the proposed adaptive sampling strategies is particularly promising for high-dimensional approximation and, specifically, for the case of nontensorial domains, where the optimal sampling measure is not known \emph{a priori}. However, this is beyond the scope of this paper and is left for future investigation.


\section{Conclusions}

We have derived novel oracle-type inequalities for the number of measurements needed to guarantee stable and robust recovery for compressed sensing in the noisy setting. Our analysis relies on a random sign assumption for the signal to be recovered and encompasses the frameworks of block-structured and isolated random measurements, in particular subsampled from a finite-dimensional isometry. 

The proposed analysis reveals a direct link between the number of measurements and the support of the signal to be recovered. This allows one to derive optimal sampling strategies in order to minimize the number of measurements and that are tailored to particular sparsity structure in the signal support. 

We have derived optimal sampling strategies in the case of (i) subsampling one-dimensional Fourier-Haar transform via isolated measurements combined with the sparsity in levels structure and (ii)  subsampling the two-dimensional Fourier-Haar transform via block-structured measurements combined with anisotropic (horizontal or vertical) sparsity in levels. Finally, we have shown how to perform adaptive sampling for one-dimensional function approximation from pointwise data.

All these results are based on a random sign assumption; for future work, it would be worth devising alternative proof techniques in order to remove it.
The analysis of the extra assumptions \eqref{cond:killthetheta} and \eqref{cond:killthethetalog} involving $\Lambda$ and $\Gamma$ could be also refined: are they necessary conditions to obtain oracle-type results? As for applications, the polynomial approximation example reveals an adaptive sampling strategy that is worth deepening and extending to the multivariate case as well. Besides, one could adapt the oracle-type  study to the case of weighted $\ell^1$ minimization which is of particular interest for polynomial approximation.
\replace{Finally, one could mention that the oracle-type inequality considered in this paper ensures robust but not necessary stable recovery for the oracle least-squares estimator. Extending the analysis to stable oracle recovery is still an open issue.}{}

\section*{Acknowledgement}

This work was supported by the Natural Sciences and Engineering Research Council of Canada [grant number 611675 to B.A. and S.B.]; and the Pacific Institute for the Mathematical Sciences (PIMS) [``PIMS Distinguished Visitor" program to C.B. and ``PIMS Postdoctoral Training Centre in Stochastics'' program to S.B].

\noindent The authors would like to thank Pierre Weiss for raising the optimal sampling concern \replace{}{and the anonymous referees for their insightful and constructive comments}.


\appendix


\section{Proof of the main results}
\label{sec:proof_main}

\subsection{Proof of Theorem~\ref{thm:noiseless}}
\label{proof:noiseless}

In order to prove that $x$ is the unique solution of \eqref{pb:BP}, one can use Fermat's rule for \eqref{pb:BP}. Under injectivity of $A_S$, 
{this corresponds to the quest of} a vector $v \in \Rc (A^*)$, called dual certificate, such that
\[
\left\{
\begin{array}{ll}
v_S &= \sgn (x_S ) \\
\| v_{S^c} \|_\infty &<1.
\end{array}
\right.
\]
A natural candidate for such a vector $v$ is the dual certificate with minimal $\ell^2$-norm, i.e.
$$
v = A^* (A_S^\dagger)^* \sgn ( x_{S}) = A^* A_S (A_S^* A_S)^{-1}  \sgn (x_{S}),
$$
which trivially lies in the range of $A^*$ and satisfies $v_S=\sgn (x_S )$. The second condition $\| v_{S^c} \|_\infty <1$ remains to be satisfied, as required by assumption (ii) of the following proposition.
\begin{prop}[{\cite[Corollary 4.28]{foucart2013mathematical}}]
\label{cor:FR}
For $x\in \Cbb^n$ with support $S$, if
\begin{enumerate}
\item $A_S$ is injective,
\item $\left|\left\langle A_S^\dagger A e_\ell, \sgn (x_S) \right\rangle \right| < 1$ for all $\ell\in S^c$,
\end{enumerate} 
then the vector $x$ is the unique solution of \eqref{pb:BP} with $y=Ax$.
\end{prop}

By Lemma \ref{lem:localIsometry_ext}, one gets the injectivity of $A_S$ with high probability.
The rest of the proof is then dedicated to ensure (ii) of Proposition \ref{cor:FR}. By a union bound, one can control the probability that \eqref{pb:BP} fails to recover $x$ as follows:
\begin{align*}
\Pbb &( \text{failure of BP}  ) \\
&\replacemath{=}{\leq}\Pbb \left( \max_{\ell \in S^c} \left|\left\langle A_S^\dagger A e_\ell, \sgn (x_S) \right\rangle \right| \geq 1\right) \\
&\leq \Pbb\left(  \max_{\ell \in S^c} \left|\left\langle A_S^\dagger A e_\ell, \sgn (x_S) \right\rangle \right| \geq 1 \left| \max_{\ell \in S^c} \|A_S^\dagger Ae_\ell \|_2 \leq \alpha \right. \right) + \Pbb \left(  \max_{\ell \in S^c} \|A_S^\dagger Ae_\ell \|_2 \geq \alpha \right) 
\\
&\leq \sum_{\ell\in S^c} \Pbb\left(  \left|\left\langle A_S^\dagger A e_\ell, \sgn (x_S) \right\rangle \right| \geq \|A_S^\dagger Ae_\ell \|_2\alpha^{-1}   \left| \max_{\ell \in S^c} \|A_S^\dagger Ae_\ell \|_2 \leq \alpha \right. \right) + \Pbb \left(  \max_{\ell \in S^c} \|A_S^\dagger Ae_\ell \|_2 \geq \alpha \right) 
\\
& \leq 2 n \exp \left( -\frac{1}{2\alpha^2} \right) + \Pbb \left(  \max_{\ell \in S^c} \|A_S^\dagger Ae_\ell \|_2 \geq \alpha \right), 
\end{align*}
where the last inequality is obtained using a Hoeffding-type bound (see  \cite[Corollary 7.21 {and Corollary 8.10}]{foucart2013mathematical} {for  Rademacher and Steinhaus sequences, respectively)}.
To control the second term, one can remark that for all $\ell\in S^c$
$$
\|A_S^\dagger Ae_\ell \|_2 = \| (A_S^* A_S)^{-1} A_S^* A e_\ell \|_2 \leq \| (A_S^* A_S)^{-1} \|_{2\to 2} \|A_S^* A e_\ell \|_2.
$$
Using Lemma~\ref{lem:localIsometry_ext}, $\| (A_S^* A_S)^{-1} \|_{2\to 2} $ is bounded by $1/(1-\delta)$ for some $\delta>0$ with high probability.
Using Lemma~\ref{lem:C.3Ext}, $\max_{\ell\in S^c} \|A_S^* A e_\ell \|_2 \leq t$ with high probability. Then, we set $$\alpha:=t/(1-\delta).$$
The probability that \eqref{pb:BP} fails is then bounded by
\begin{align*}
\Pbb &\left( \max_{\ell \in S^c} \left|\left\langle A_S^\dagger A e_\ell, \sgn (x_S) \right\rangle \right| \geq 1\right) \\
&\leq \underbrace{2n \exp \left( -\frac{1}{2\alpha^2} \right) }_{P_1}
+\underbrace{\Pbb \left( \| A_S^*A_S - {\Id} \|_{2\to 2}  \geq \delta \right)}_{P_2} + \underbrace{\Pbb\left(  \max_{\ell\in S^c} \|A_S^* A e_\ell \|_2 \geq t  \right) }_{P_3}.
\end{align*}

Note that $P_2 \leq \varepsilon/3$  by Lemma \ref{lem:localIsometry_ext} if
\begin{align}
\label{eq:m1}
m \geq \frac{1+2\delta/3}{\delta^2/2} \cdot \Lambda(S,F) \cdot \ln \left( \frac{6s}{\varepsilon} \right).
\end{align}

As for the last term $P_3$, 
using Lemma \ref{lem:C.3Ext}, choosing $t = \sqrt{\Theta/m} + t'$ for some $t'>0$ to be fixed later, one has
 $P_3\leq \varepsilon/3$ if 
\begin{align}
\label{eq:m2}
m \geq  \frac{2}{(t')^2} \cdot  \Theta (S,F) \cdot \left( 3+2 t'/3 \right)  \ln \left( \frac{3n}{\varepsilon}\right),
\end{align}
{where we have used} that $m \geq 4 \Theta(S,F)$.
Finally \replace{vy}{by} setting $t'=\delta$, one has
\begin{align*}
P_1 \leq \varepsilon/3 
& \Longleftrightarrow 2n \exp \left( -\frac{1}{2\alpha^2} \right) \leq \varepsilon /3 
 \Longleftrightarrow 2n \exp \left( -\frac{(1-\delta)^2}{2t^2} \right) \leq \varepsilon /3 \\
& \Longleftrightarrow 2n \exp \left( -\frac{(1-\delta)^2}{2\left( \sqrt{\Theta/m} +\delta\right)^2} \right) \leq \varepsilon /3
  \Longleftrightarrow\frac{(1-\delta)^2}{2\left( \sqrt{\Theta/m} +\delta\right)^2} \geq \ln\left(\frac{6n}{\varepsilon} \right).
\end{align*}
Assuming that $m\geq c\cdot \Theta(S, F) \ln (6n/\varepsilon)$ and choosing $\delta = \sqrt{\frac{1}{\replacemath{c'}{c} \ln (6n/\varepsilon)}}$ with \replace{$\min (c ,c') \geq 8$}{$c \geq 13$} leads to the previous inequality (notice that the constant \replace{8}{13} could be optimized in principle). Consequently, if 
$$
\delta = \replacemath{\frac{1}{2\sqrt{2 \ln \left( \frac{6n}{\varepsilon} \right)} }}{\frac{1}{\sqrt{13 \ln \left( \frac{6n}{\varepsilon} \right)} }}
\qquad \text{and}\qquad
m \geq \replacemath{8}{13} \cdot \Theta {(S,F)} \cdot \ln \left( \frac{6n}{\varepsilon} \right),
$$
then $P_1 \leq \varepsilon/3$.

Plugging this value of $\delta$ into \eqref{eq:m1} and \eqref{eq:m2},  one gets the following conditions on the number of measurements:
\[
\left\{
\begin{array}{ll}
m & \geq \replacemath{19}{30}   \cdot \Lambda  (S,F)   \cdot \ln \left( \frac{6n}{\varepsilon} \right) \ln \left( \frac{6s}{\varepsilon} \right),   \\
m &\geq  \replacemath{51}{82} \cdot \Theta  (S,F)  \cdot \ln^2 \left( \frac{6n}{\varepsilon} \right).
\end{array}
\right.
\]
Finally, the probability of recovery failure from \eqref{pb:BP} is bounded by $\varepsilon$, if the total number of blocks of measurements satisfy
\begin{align*}
m & \geq  \replacemath{51}{82} \cdot \Theta(S,F)  \cdot \ln^2 \left( \frac{6n}{\varepsilon} \right).
\end{align*}
This concludes the proof of Theorem \ref{thm:noiseless}.

\subsection{Proof of Theorem \ref{thm:noisy}}
\label{proof:noisy}

In order to ensure stable and robust recovery via \eqref{pb:qBP}, we take advantage of a result analogous to  Proposition~\ref{cor:FR}, which hinges again on the concept of dual certificate. (Notice that in \cite[Theorem 4.33]{foucart2013mathematical} $S$ is assumed to be the set of $s$ largest absolute entries of $x$, but an inspection of the proof reveals that $S$ can be an arbitrary subset of $\{1,\ldots,n\}$).

\begin{prop}[{\cite[Theorem 4.33]{foucart2013mathematical}}]
\label{prop:FRstable}
Let $x \in \Cbb^n$, $S \subseteq\{1,\ldots,n\}$, and  $y = Ax + \epsilon$ with $\|\epsilon\|_2 \leq \eta$. For $\delta,t,\gamma,\theta,\tau \geq 0$, with $\delta < 1$, assume that
\begin{equation}
\label{eq:FRstable_cond1}
\|A^*_S A_S - \Id\|_{2\to2} \leq \delta,
\quad 
\max_{\replacemath{\ell \in S}{\ell \in S^c}}\|A^*_S A e_\ell\|_2 \leq t,
\end{equation}
and that there exists a vector $u = A^*h \in \Cbb^n$ with $h$ such that
\begin{equation*}
\|u_S - \sgn(x_S)\|_2 \leq \gamma, 
\quad \|u_S\|_\infty \leq \theta, 
\quad \|h\|_2 \leq \tau \sqrt{s}. 
\end{equation*}
If $\rho := \theta + t \gamma /(1-\delta) < 1$, then a minimizer $x^\sharp$ of \eqref{pb:qBP} satistfies
$$
\|x - x^\sharp\|_2 \leq C_1 \|x-x_S\|_1 + (C_2 + C_3\sqrt{s})\eta
$$
for some constants $C_1, C_2, C_3 > 0$ depending only on $\delta, t, \gamma, \theta, \tau$. 
\end{prop}

Using arguments analogous to the proof of Theorem \ref{thm:noiseless}, one  
can show that \eqref{eq:FRstable_cond1} holds with high probability.
Therefore, using Proposition~\ref{prop:FRstable}, we choose the dual certificate to be $u = A^* h$ with $h=  A_S (A_S^* A_S)^{-1}\sgn (x_S)$. 

Since $u_S = \sgn(x_S)$, for any $\gamma \geq 0$, the condition 
$ \|u_S - \sgn(x_S) \|_2\leq \gamma$
is trivially satisfied  {(in particular, we can set $\gamma = 0$)}. 

As for ensuring that $\|u_{S^c}\|_\infty \leq \theta$ with high probability, 
we have to slightly modify a part of the proof in the noiseless setting. First, let us observe that
\begin{align*}
\|u_{S^c}\|_\infty &= \max_{\ell \in S^c} \left|\left\langle A^*  A_S (A_S^* A_S)^{-1} \sgn (x_S) , e_\ell \right\rangle \right| 
= \max_{\ell \in S^c} \left|\left\langle  \sgn (x_S) , A_{S}^\dagger A e_\ell \right\rangle \right|. 
\end{align*}
For $0<\theta=1/4<1$, using again a union bound and a Hoeffding-type inequality (see \cite[Corollary 7.21 and Corollary 8.10]{foucart2013mathematical} for  Rademacher and Steinhaus sequences, respectively), one has

\begin{align*}
\Pbb &\left( \max_{\ell \in S^c} \left|\left\langle A_S^\dagger A e_\ell, \sgn (x_S) \right\rangle \right| \geq  \theta\right) \\
&\leq \Pbb\left(  \max_{\ell \in S^c} \left|\left\langle A_S^\dagger A e_\ell, \sgn (x_S) \right\rangle \right| \geq \theta \left| \max_{\ell \in S^c} \|A_S^\dagger Ae_\ell \|_2 \leq \alpha \right. \right) + \Pbb \left(  \max_{\ell \in S^c} \|A_S^\dagger Ae_\ell \|_2 \geq \alpha \right) 
\\
&\leq \sum_{\ell\in S^c} \Pbb\left(  \left|\left\langle A_S^\dagger A e_\ell, \sgn (x_S) \right\rangle \right| \geq \theta \|A_S^\dagger Ae_\ell \|_2 \alpha^{-1}   \left| \max_{\ell \in S^c} \|A_S^\dagger Ae_\ell \|_2 \leq \alpha \right. \right) + \Pbb \left(  \max_{\ell \in S^c} \|A_S^\dagger Ae_\ell \|_2 \geq \alpha \right) 
\\
& \leq  2n \exp \left( -\frac{\theta^2}{2\alpha^2} \right) + \Pbb \left(  \max_{\ell \in S^c} \|A_S^\dagger Ae_\ell \|_2 \geq \alpha \right) \\
& \leq 
\underbrace{2n \exp \left( -\frac{\theta^2}{2\alpha^2} \right)}_{P_1} 
+\underbrace{\Pbb \left( \| A_S^*A_S - {\Id} \|_{2\to 2}  \geq \delta \right)}_{P_2} + \underbrace{\Pbb\left(  \max_{\ell\in S^c} \|A_S^* A e_\ell \|_2 \geq t  \right) }_{P_3}.
\end{align*}
Now, analogously to the proof of Theorem~\ref{thm:noiseless}, we fix
$$
\alpha = \frac{t}{1-\delta}, \quad t = \sqrt{\frac{\Theta}{m}}+t', \quad t' = \delta.
$$ 
With these choices of parameters, conditions \eqref{eq:m1} and \eqref{eq:m2} suffice to guarantee that $P_2 \leq \varepsilon/3$ and $P_3\leq \varepsilon/3$, respectively. {Moreover}, since $\theta=1/4$,
if
$$
\delta = \replacemath{\frac{1}{2\sqrt{32 \ln \left( \frac{6n}{\varepsilon} \right)}}}{\frac{1}{\sqrt{146 \ln \left( \frac{6n}{\varepsilon} \right)}}}
\quad \text{and} \quad m \geq {\replacemath{128}{146}} \cdot \Theta(S,F) \cdot\ln\left(\frac{6n}{\varepsilon}\right), 
$$
then ${P_1} \leq \varepsilon/3$.
Plugging this value of $\delta$ in \eqref{eq:m1} and \eqref{eq:m2},  one gets the following conditions on the number of measurements:
\[
\left\{
\begin{array}{ll}
m & \geq \replacemath{268}{305} \cdot  \Lambda  (S,F) \cdot \ln \left( \frac{6n}{\varepsilon} \right) \ln \left( \frac{6s}{\varepsilon} \right)   \\
m &\geq \replacemath{780}{889} \cdot \Theta  (S,F)  \cdot \ln^2 \left( \frac{6n}{\varepsilon} \right).
\end{array}
\right.
\]

Now, we have to find a suitable constant $\tau > 0$ such that $\|h\|_2 \leq \tau \sqrt{s}$. Note that
\begin{align*}
\|h\|_2 &= \|  A_S (A_S^* A_S)^{-1}\sgn (x_S) \|_2 \leq  \|  A_S (A_S^* A_S)^{-1} \|_{2\to 2} \| \sgn (x_S) \|_2  =\frac{1}{\sigma_{\min}(A_S)} \sqrt{s},
\end{align*}
since $A_S$ has full column rank. Moreover, observe that
$$
\|A_S^* A_S - \Id \|_{2\to 2} \leq \delta \Rightarrow \sigma_{\min} (A_S) \geq \sqrt{1-\delta}.
$$
Hence, one can choose $\tau = \frac{1}{\sqrt{1-\delta}}$, so that
\begin{align*}
\Pbb \left(  \| h\|_2 \geq  \sqrt{\frac{s}{1-\delta}} \right) &\leq  \Pbb \left(  \frac{1}{\sigma_{\min}(A_S)} \geq  \sqrt{\frac{1}{1-\delta}} \right)  \leq \Pbb \left(  \|A_S^*A_S -\Id \|_{2\to 2} \geq \delta\right),
\end{align*}
which has been already controlled as $P_2 \leq \varepsilon / 3$.

Finally, in order to apply Proposition~\ref{prop:FRstable}, let us notice that $\theta + \gamma t/(1-\delta) <1$ {since $\gamma = 0$ and $\theta = 1/4$. 

This} leads to the desired result.

\subsection{Proof of Theorem~\ref{thm:oracle_theta}}
\label{proof:oracle_theta}

We will show a more general result, which will imply Theorem~\ref{thm:oracle_theta} as a corollary. In particular, in order to obtain Theorem~\ref{thm:oracle_theta} it will suffice to consider the partition $S^{c}$ formed by its singletons in Theorem~\ref{thm:partition_trick}, i.e.\ $S^c = \sqcup_{j \in S^c}\{j\}$.

\begin{thm}
\label{thm:partition_trick}
Let \replace{$x\in \Cbb^n$}{$x\in \Rbb^n$ or $\Cbb^n$} be a vector supported on $S$, such that $\sgn(x_S)$ forms a Rademacher or Steinhaus sequence. Let $A$ be the random sensing matrix defined in \eqref{eq:sensing_matrix} with parameter $\Lambda(S,F)$ {and let $y = Ax + \epsilon$, with $\|\epsilon\|_2 \leq \eta$}. {Then, there exist constants $c_1,C_1,C_2  > 0$ such that the following holds. For every $0<\varepsilon < 1$, if} 
$$
m \geq c_1  \min_{\substack{\{S_r\}\\ \textnormal{partition of }S^c}}\max_r \Lambda(S\cup S_r,F) \ln \left(\frac{6  |\{S_r\}| \replacemath{}{\max_r}|S \cup S_r|}{\varepsilon}\right)\ln\left(\frac{6n}{\varepsilon}\right),
$$ 
then, with probability at least $1-\varepsilon$, a minimizer $x^\sharp$ of \eqref{pb:qBP} satisfies
\begin{align*}
\|x-x^\sharp\|_2 \leq  (C_1 + C_2\sqrt{s})\eta.
\end{align*}
In particular, in the noiseless case (i.e., $\eta = 0$), $x$ is exactly recovered via \eqref{pb:BP} with probability at least $1-\varepsilon$ and with constant $c_1 = \replacemath{36}{19}$.
\end{thm}

For the sake of simplicity, let us address the case $\eta = 0$ (the argument can be generalized to the case $\eta >0$ in the same spirit of Appendix~\ref{proof:noisy}). Analogously to the proof of Theorem~\ref{thm:noiseless} in Appendix~\ref{proof:noiseless}, our aim is to find sufficient conditions such that 
$$
\underbrace{2n \exp\left(-\frac{1}{2\replacemath{\alpha}{\alpha^2}}\right)}_{\replacemath{P1}{P_1}} 
+ 
\underbrace{\Pbb\left(\|A_S^*A_S - \Id\|_{2 \to 2} \geq \delta \right)}_{P_2}
+ 
\underbrace{\Pbb\left(\max_{\ell\in S^c} \|A_S^* A e_\ell \|_2 \geq t\right)}_{P_3} \leq \varepsilon,
$$
where $\alpha = t/(1-\delta)$. In the following, we show how to ensure that $P_i \leq \varepsilon /3$, for $i = 1,2,3$.

By using Lemma~\ref{lem:localIsometry_ext}, we have $P_2 \leq \varepsilon /3$ under condition \eqref{eq:m1}.



In order to control $P_3$, let us consider a partition $\{ S_r \}_r$ of $S^c$ and note that
\begin{align*}
\max_{\ell \in S^c} \|A_S^* A e_\ell \|_2 
&\leq \max_{r} \|A_S^* A_{S_r} \|_{2 \to 2 } 
= \max_{r} \|P_S A^* A P_{S_r}^{*} \|_{2 \to 2 } 
= \max_{r} \|P_S (A^* A - \Id) P_{S_r}^{*} \|_{2 \to 2 } \\
& \leq \max_{r} \|P_{S\cup S_r} (A^* A - \Id) P_{S\cup S_r}^{*} \|_{2 \to 2 } 
 = \max_{r} \|A^*_{T_r} A_{T_r} - \Id \|_{2\to 2},
\end{align*}
where we have used that $P_S P_{S_r}^* = 0$ and where we have defined $T_r= S \cup S_r$ for each index $r$ of the partition. Now, combining the above inequality with Lemma~\ref{lem:localIsometry_ext} and with a union bound over the elements $S_r$ of the partition $\{S_r\}$, we obtain
$$
P_3 \leq \Pbb\left(\max_r \|A_{T_r}^* A_{T_r} - \Id\|_{2 \to 2} \geq t \right)
\leq 
2 |\{S_r\}|\left( \max_r |T_r| \right) \exp \left(-\frac{mt^2/2}{\left(\displaystyle\max_{r}\Lambda(T_r,F)\right) (1+2t/3)}\right).
$$
Notice that introducing the partition $\{S_r\}$ allowed us to control $P_3$ by the quantity $\displaystyle\max_{r}\Lambda(T_r,F)$ rather than $\Theta(S,F)$. Simple algebraic manipulations show that the condition
\begin{equation}
\label{eq:condLambdaTr}
m \geq \frac{2}{t^2} \left(1 + \frac{2t}{3}\right) \max_r \Lambda(T_r,F) \ln \left(\frac{6 |\{S_r\}| \replacemath{}{\max_r}|T_r| }{\varepsilon}\right),
\end{equation}
is sufficient to have \replacemath{$P_2 \leq \varepsilon$}{$P_3 \leq \varepsilon /3$}.

Now, by choosing 
$$
t = \delta \quad \text{and} \quad \delta = \frac{1}{\sqrt{8 \ln(6n /\varepsilon)}},
$$
it is not difficult to show that $P_1 \leq \varepsilon /3$ is satisfied. Moreover, with this choice, condition \eqref{eq:condLambdaTr} implies \eqref{eq:m1}.

Finally, observing that
$$
\frac{2}{\delta^2} \left(1 + \frac{2\delta}{3}\right) \leq \replacemath{36}{19} \ln\left(\frac{6n}{\varepsilon}\right) 
$$
and taking the minimum over all partitions $\{S_r\}$ of $S^c$ in \eqref{eq:condLambdaTr} concludes the proof.

\subsection{Proof of Theorem \ref{thm:killthetalog}}
\label{proof:noiseless_log}

The proof   is analogous to that of Theorem \ref{thm:noiseless}. Therefore, we will employ the same notation as in Appendix~\ref{proof:noiseless}. By similar arguments, one has
\begin{align*}
\Pbb ( \text{failure of BP}  ) 
 &\leq \underbrace{2n \exp \left( -\frac{1}{2\alpha^2} \right) }_{P_1}
+\underbrace{\Pbb \left( \| A_S^*A_S - {\Id} \|_{2\to 2}  \geq \delta \right)}_{P_2} + \underbrace{\Pbb\left(  \max_{\ell\in S^c} \|A_S^* A e_\ell \|_2 \geq t  \right) }_{P_3},
\end{align*}
for some $t, \delta >0$ and $\alpha = t /(1-\delta)$. 
Note that $P_2 $ is controlled as before using Lemma \ref{lem:localIsometry_ext}. In particular, $P_2 \leq \varepsilon/3 $ if
\begin{align}
\label{eq:m1_log}
m \geq \frac{1+\delta/3}{\delta^2/2} \cdot \Lambda(S,F) \cdot \ln \left( \frac{6s}{\varepsilon} \right).
\end{align}
Thus, let us suppose for the rest of the proof to choose 
\begin{equation}
\label{eq:m_choice}
{m \geq c_2 \cdot \Lambda(S,F) \cdot \ln\left(\frac{3n}{\varepsilon}\right)},
\end{equation}
where $c_{2}\geq (1+\delta/3)/(\delta^2/2)$ is a  constant that will be fixed later and where we are using $s \leq  n/2$.
The slight modification in the proof compared to the one in Appendix \ref{proof:noiseless} appears in the control of $P_3$. Indeed, using Lemma \ref{lem:C.5Ext}, one has for some $t'>0$
$$
\Pbb \left( \max_{i\in S^c} \| A_S^* A e_i \|_2 \geq \sqrt{\frac{\Lambda}{m}} + t' \right) \leq n \exp\left(-\frac{m {(t')^2}/2}{\Gamma + 4 \Lambda \sqrt{\Gamma /m} + 2 t' \sqrt{\Lambda\Gamma}/3}\right).
$$
Therefore, combining the above inequality with \eqref{eq:m_choice}, fixing $t'= \frac{1}{2}\sqrt{\frac{\Lambda}{m}}$, and assuming that $ \Lambda \geq c_1 \Gamma \ln(3n/\varepsilon)$,  one has
\begin{align*}
P_3 &\leq n \exp \left(- \frac{\Lambda /8}{\Gamma + 13/3 \Lambda \sqrt{\Gamma /{m}}} \right) 
\leq n \exp\left( -\frac{1/8}{1/{c_1} + 13/(3 \sqrt{c_1 c_2})} \ln \left( 3n / \varepsilon \right)  \right) 
\leq \varepsilon/3, 
\end{align*}
provided that
\begin{equation}
\label{eq:cond_c1_c2}
 \frac{1}{8} \geq \frac{1}{{c_1}} + \frac{13}{3\sqrt{c_1 c_2}}.
\end{equation}

Condition on $P_1 \leq \varepsilon / 3$ remains to be checked. In accordance with the previous computation, we set $t= \frac{3}{2}\sqrt{\frac{\Lambda}{m}} $ and if we choose $\delta =1/2$, with $\alpha = t/(1-\delta) $, we obtain
\begin{align*}
P_1 &\leq \varepsilon/3 \qquad  \Longleftarrow \qquad  m \geq 18 \cdot \Lambda(S,F)  \cdot \ln\left( \frac{3n}{\varepsilon} \right).
\end{align*}
Finally, {we note that \eqref{eq:cond_c1_c2} holds true for, e.g., $c_1 = 50$ and $c_2 = 100$. Therefore,} conditions
\[
\left\{
\begin{array}{ll}
\Lambda (S,F) \geq 50 \cdot  \Gamma (F) \cdot  \ln(3n/\varepsilon)  \\
m \geq 100 \cdot  \Lambda(S,F) \cdot \ln\left( \frac{3n}{\varepsilon} \right),
\end{array}
\right.
\]
ensure that \eqref{pb:BP} exactly  recovers $x$ with probability larger than $1-\varepsilon$. 

Note that no effort was made in order to optimize the constants $c_1$ and $c_2$, and the result in the noisy setting can be easily deduced from the noiseless one as in Appendix \ref{proof:noisy}.




\section{A discussion on stability}
\label{app:stability}
\replace{}{
In Theorems~\ref{thm:oracle_theta}, \ref{thm:killthetheta}, and \ref{thm:killthetalog}, we assume that the vector $x$ is exactly $s$-sparse. However, it is possible to extend these results to the case of an arbitrary vector $x \in \mathbb{R}^n$ or $\mathbb{C}^n$ such that $\sgn(x_S)$ is a random Rademacher or Steinhaus sequence. Indeed, in this case, if $|S|=s$ and the measurements are corrupted by noise $\epsilon$  such that $\|\epsilon\|_2 \leq \eta$, Proposition~\ref{prop:FRstable} ensures a recovery guarantee of the form
$$
\|x-x^{\sharp}\|_2 \leq (C_1 + C_2 \sqrt{s}) \eta + C_3 \|x_{S^c}\|_1.
$$
This implies the stability of the recovery guarantees with respect to the standard sparsity model, in addition to its robustness to bounded noise. We now clarify in what sense this generalization of Theorems~\ref{thm:oracle_theta}, \ref{thm:killthetheta}, and \ref{thm:killthetalog} would still be of ``oracle type''. 

Let us go back to the proof of Proposition~\ref{prop:oracle}. We showed that the definition \eqref{eq:oracle_ls_estimator} of  the oracle least-squares estimator $x^\star$ is sufficient to have
\begin{equation}
\label{eq:oracle_LS_bound}
\| x^\star - x \|_2   
\leq \frac{1}{\sigma_{\min}(A_S)}  \| \epsilon \|_2 +  \| (A_S^* A_S)^{-1} \|_{2\to 2} \| A_S^* A_{S^c} \|_{1\to 2} \| x_{S^c} \|_1 .
\end{equation}
Moreover,  in view of Lemma~\ref{lem:localIsometry_ext} condition \eqref{eq:bound_meas_oracle} implies $\sigma_{\min}(A_S) \geq \sqrt{1-\delta}$ and  $\| (A_S^* A_S)^{-1} \|_{2\to 2} \leq 1/(1-\delta)$. 
In order to make \eqref{eq:oracle_LS_bound} a stable and robust recovery guarantee, we need to control the quantity $\| A_S^* A_{S^c} \|_{1\to 2}$. To achieve this, we observe that $\| A_S^* A_{S^c} \|_{1\to 2} = \max_{j \in S^c} \|A_S^* A e_j\|_2$. Therefore, the condition 
\begin{equation}
\label{eq:control_off_support_oracle}
\max_{j \in S^c} \|A_S^* A e_j\|_2 \leq t,
\end{equation}
combined with \eqref{eq:oracle_LS_bound}, implies the following stable and robust recovery guarantee for the oracle least-squares estimator:
$$
\| x^\star - x \|_2   
\leq \frac{1}{\sqrt{1-\delta}}  \| \epsilon \|_2 +  \frac{t}{1-\delta} \| x_{S^c} \|_1 .
$$
We conclude by observing that \eqref{eq:control_off_support_oracle} is one of the hypotheses of Proposition~\ref{prop:FRstable} (see \eqref{eq:FRstable_cond1}) and that it holds with probability at least $1- \varepsilon/3$ under the assumptions on $m$ stated in any of the Theorems~\ref{thm:oracle_theta}, \ref{thm:killthetheta}, or \ref{thm:killthetalog} (recall the definition of $P_3$ in the corresponding proofs). 

}


\section{Bernstein's inequalities}
\label{app:bernstein}

In this section, we present two Bernstein-type inequalities employed in the proofs of Appendix~\ref{app:aux_lemmas}.

First, we present an extension of the vector Bernstein inequality \cite[Theorem 8.45]{foucart2013mathematical} in order to handle random vectors $Y_i$ that are independent but not necessarily identically distributed. Then, we recall a Bernstein inequality for self-adjoint matrices, corresponding to \cite[Corollary 8.15]{foucart2013mathematical}. Note that both results are independent of the dimension of the random objects (vectors or matrices) involved.

\begin{thm}[Vector Bernstein Inequality]
\label{thm:vecBernsteinExt}
Consider a set of independent random vectors $Y_1, \ldots, Y_m $ such that  
$$
\Ebb Y_i = 0, \quad \|Y_i\|_2 \leq K, \text{ a.s.}, \quad \forall i= 1,\ldots,m,
$$
and let $\sigma, \mu >0 $ such that 
$$
\sup_{\|x\|_2 \leq 1} \sum_{i=1}^m \Ebb |\langle x,Y_i\rangle|^2 \leq \sigma^2, 
\quad \Ebb Z \leq \mu, \quad \text{where }
Z:=\left\|\sum_{i = 1}^m Y_i\right\|_{2}.
$$
Then, for every $t > 0$, the following holds:
$$
\Pbb\left(Z \geq \mu + t\right)
\leq \exp\left(-\frac{t^2/2}{\sigma^2 + 2K \mu + t K/3}\right).
$$
\end{thm}
\begin{proof}
Assume $Y_i \in \mathbb{C}^n$, consider $B := \{x \in \mathbb{C}^n : \|x\|_2 \leq 1\}$, and let $\tilde B$ be a dense countable subset of $B$. Define $F_x(Y):=\Re\langle x, Y \rangle$. Then, $Z$ can be expressed as a supremum of an empirical process. Indeed,
\begin{align*}
Z & = \left\|\sum_{i = 1}^m Y_i\right\|_2
 = \sup_{x \in B} \Re \langle x,\sum_{i=1}^m Y_i\rangle
 = \sup_{x \in \tilde B}\sum_{i = 1}^mF_x(Y_i)
 = \sup_{F \in \mathcal{F}} \sum_{i = 1}^m F(Y_i),
\end{align*}
where we have  defined $\mathcal{F}:=\{F_x : x \in \tilde B\}$ in the last step. Now, we verify the hypotheses needed to apply Talagrand's inequality (see \cite[Theorem 8.42]{foucart2013mathematical}).

First, $\Ebb F_x(Y_i) = \Ebb \Re \langle x, Y_i \rangle =  0$ because the $Y_i$'s are centered. Moreover, for every $x \in \tilde B$, $F_x(Y_i) = \Ebb \Re \langle x, Y_i \rangle \leq \|Y_i\|_2 \leq K$ almost surely and
$$
\sum_{i = 1}^m \Ebb F_x^2(Y_i) 
= \sum_{i = 1}^m \Ebb (\Re\langle x, Y_i \rangle) ^2
\leq \sum_{i = 1}^m \Ebb |\langle x, Y_i \rangle| ^2
\leq \sigma^2.
$$
Finally, Talagrand's inequality yields
$$
\Pbb \left( Z \geq \mu + t\right)
\leq \Pbb \left( Z \geq \Ebb Z + t\right)
\leq \exp\left(-\frac{t^2/2}{\sigma^2 + 2K \Ebb Z + t K/3}\right)
\leq \exp\left(-\frac{t^2/2}{\sigma^2 + 2K \mu + t K/3}\right),
$$
which is the desired result.
\end{proof}

\begin{thm}[Bernstein Inequality for self-adjoint matrices, {\cite[Corollary 8.15]{foucart2013mathematical}}]
\label{cor8.15FR} 
 Let $(Z_k)_{1 \leq k \leq n}$ be a finite sequence of  independent, random, self-adjoint matrices such   that $\Ebb Z_k = 0$ and that $\| Z_k \|_{2\rightarrow 2} \leq K$ a.s.\ for some constant $K > 0$ independent of $k$. Define 
$$ \sigma^2 =  \left\|\sum_{k = 1}^{n} \Ebb  Z_k^2 \right\|_{2\rightarrow 2}  
$$
Then, for any $t > 0$, we have that
$$ \Pbb \left( \left\|\sum_{k = 1}^{n} Z_k  \right\|_{2\rightarrow 2} \geq t \right) \leq 2 d \exp\left( -\frac{t^2/2}{\sigma^2 + Kt/3} \right).
$$
\end{thm}

\section{Auxiliary lemmas}
\label{app:aux_lemmas}

In this section, we show some deviation inequalities involving submatrices of the sensing matrix $A$ corresponding to the setting described in Section~\ref{subsec:sampling}. In particular, the tail probabilities will be controlled by using the quantities $\Lambda = \Lambda(S,F)$, $\Theta = \Theta(S,F)$, and $\Gamma = \Gamma(F)$ introduced in Definition~\ref{def:quantities}, where $S \subseteq \{1,\ldots,n\}$.

The first auxiliary lemma controls the deviation of $A_S^*A_S$ from the identity by means of the quantity $\Lambda(S,F)$.

\begin{lemme}
\label{lem:localIsometry_ext}
For every $S \subseteq\{1,\ldots,n\}$ with $|S| = s$ and for every $\delta > 0$, the following holds 
$$
\Pbb(\|A_S^*A_S - \Id \|_{2 \to 2} \replacemath{{\leq}}{\geq} \delta) 
\leq 2s \exp\left(-\frac{m\delta^2/2}{\Lambda (1+2\delta/3) }\right).
$$
\end{lemme}

\begin{proof}
Our objective is to apply Theorem~\ref{cor8.15FR}.
Consider the splitting
$$
A^*_SA_S - \Id 
=\sum_{i = 1}^m \underbrace{\frac{1}{m} (B_{i,S}^*B_{i,S} - \Ebb[B_{i,S}^*B_{i,S}])}_{=:X_i}
=  \sum_{i = 1}^m X_i \in \mathbb{C}^{s \times s}.
$$
Of course, $\Ebb X_i = 0$. Moreover,
\begin{align*}
\|X_i\|_{2 \to 2} 
 = \frac{1}{m} \sup_{\|x\|_2 \leq 1}|\langle x, X_i x \rangle|
 = \frac{1}{m} \sup_{\|x\|_2 \leq 1} \left| \|B_{i,S}x\|_2^2 - \Ebb\|B_{i,S}x\|_2^2\right|.
\end{align*}
Now, observing that $\|B_{i,S}x\|_2^2 \leq \Lambda \|x\|_2^2$, we have
$$
\|X_i\|_{2 \to 2} \leq 2 \Lambda / m =: K.
$$
Being $\sum_{i= 1}^m \Ebb X_i^2$ self-adjoint, we have 
$$
\sigma^2:= \left\|\sum_{i = 1}^m \Ebb X_i^2\right\|_{2 \to 2} 
= \sup_{\|x\|_2 \leq 1} \sum_{i = 1}^m \langle x, \Ebb X_i^2 x\rangle.
$$
In order to estimate this term, we notice that
$$
\Ebb X_i^2 
= \frac{1}{m^2}\Ebb (B_{i,S}^* B_{i,S} - \Ebb[B_{i,S}^* B_{i,S}])^2
= \frac{1}{m^2} \left[ \Ebb (B_{i,S}^* B_{i,S})^2 - (\Ebb[B_{i,S}^* B_{i,S}])^2 \right].
$$
Now, observing that $\langle x, (\Ebb[B_{i,S}^* B_{i,S}])^2 x\rangle = \|\Ebb[B_{i,S}^* B_{i,S}]x\|_2^2 \geq 0$ for every $x$, we have
\begin{align*}
\sum_{i = 1}^m\langle x, \Ebb X_i^2 x\rangle
& \leq \frac{1}{m^2} \sum_{i = 1}^m \langle x, \Ebb (B_{i,S}^* B_{i,S})^2 x\rangle
= \frac{1}{m^2} \sum_{i = 1}^m \Ebb \|B_{i,S}^* B_{i,S} x \|_2^2\\
& \leq \frac{1}{m^2} \Lambda \sum_{i = 1}^m \Ebb \|B_{i,S} x \|_2^2
= \frac{\Lambda}{m^2} \|x\|_2^2
\end{align*}
Therefore, $\sigma^2 \leq \Lambda/m^2$. Now, we apply Theorem~\ref{cor8.15FR} and we obtain
\begin{align*}
\Pbb(\|A_S^*A_S - I \|_{2 \to 2} \leq \delta) 
&\leq 2s \exp\left(-\frac{\delta^2/2}{\sigma^2 + K \delta/3}\right)
\leq 2s \exp\left(-\frac{m\delta^2/2}{\Lambda (1/m+2\delta/3) }\right), \\
&\leq 2s \exp\left(-\frac{m\delta^2/2}{\Lambda (1+2\delta/3) }\right),
\end{align*}
since $m\geq 1$, which concludes the proof.
\end{proof}


The next two lemmas aim at controlling the growth of the absolute entries of $A^*_S A_{S^c}$ in terms of the quantity $\Theta(S,F)$.

\begin{lemme}
\label{lem:C.3Ext}
Let $S \subseteq \{1,\ldots,n\}$. Then, for every $t>0$
$$
\Pbb \left(\max_{i \in S^c}\|A_S^* A e_i\|_2 \geq \sqrt{\Theta/m} + t\right)
\leq n \exp\left(-\frac{m t^2/2}{\Theta ( 1 + 4  \sqrt{\Theta / m } + 2 t /3 )}\right).
$$
\end{lemme}
\begin{proof}

This proof is relies on Theorem~\ref{thm:vecBernsteinExt}. Fix $i \in S^c$ and consider $\|A^*_S A e_i\|_2$.  Note that $\Ebb A_S^* A e_i = 0$ since $i \in S^c$. Moreover, 
$$
\|A_S^*A e_i\|_2 
= \left\|\sum_{j = 1}^m \underbrace{\frac{1}{m}(B_{j,S}^*B_je_i - \Ebb(B_{j,S}^*B_je_i))}_{=: Y_j}\right\|_2
=  \left\|\sum_{j = 1}^m Y_j\right\|_2=:Z.
$$
Now, we verify the hypotheses of Theorem~\ref{thm:vecBernsteinExt} for the random vectors $Y_1,\ldots,Y_m$. Of course, $\Ebb Y_j = 0$. Moreover, recalling the definition of $\Theta$ and of $\|\cdot\|_{\infty \to \infty}$, we have
$$
\|B_{j,S}^*B_j e_i\|_2 \leq \|B_{j,S}^*B_j e_i\|_1 \leq \Theta.
$$
Therefore, we obtain
$$
\|Y_j\|_2 
= \frac{1}{m} \|B_{j,S}^*B_je_i - \Ebb(B_{j,S}^*B_je_i)\|_2
\leq 2 \Theta / m=: K.
$$
Then, using the Cauchy-Schwarz inequality, we estimate
\begin{align*}
\sup_{\|x\|_2 \leq 1}\sum_{j = 1}^m \Ebb |\langle x, Y_j \rangle|^2
& = \frac{1}{m^2}  \sup_{\|x\|_2 \leq 1}\sum_{j = 1}^m \Ebb |\langle x, B_{j,S}^* B_j e_i \rangle|^2
\leq \frac{1}{m^2} \sup_{\|x\|_2 \leq 1}\sum_{j = 1}^m \Ebb [\|B_{j,S}\|_{2\to 2}^2 \|x\|_2^2 \|B_j e_i\|_2^2]\\
& \leq \frac{\Lambda}{m^2} \sum_{j = 1}^m \Ebb \|B_j e_i\|_2^2
= \frac{\Lambda}{m} \| e_i \|_2^2 = \frac{\Lambda}{m} \leq \frac{\Theta}{m}=: \sigma^2,
\end{align*}
where in the last step we used $\frac{1}{m} \sum_{j=1}^m \Ebb (B_j^* B_j)= \Id$.
Furthermore, using the independence of the $Y_i$'s and their zero-mean property, we see that
$$
(\Ebb Z)^2 
\leq \Ebb Z^2 
= \Ebb \left\|\sum_{j = 1}^m Y_j\right\|_2^2
= \Ebb \sum_{j = 1}^m\sum_{k = 1}^m \langle Y_j, Y_k \rangle
= \sum_{j = 1}^m \Ebb\|Y_j\|^2
+
\underbrace{\sum_{j = 1}^m\sum_{k \neq j} \langle \Ebb Y_j, \Ebb Y_k \rangle}_{=0}.
$$
Now, combining the above inequality with the following estimate
\begin{align*}
\sum_{j = 1}^m \Ebb \|Y_j\|_2^2
& = \frac{1}{m^2} \sum_{j = 1}^m \Ebb\|B_{j,S}^*B_je_i - \Ebb(B_{j,S}^*B_je_i)\|_2^2
= \frac{1}{m^2} \sum_{j = 1}^m (\Ebb \|B_{j,S}^*B_je_i\|_2^2 - \|\Ebb (B_{j,S}^*B_je_i)\|_2^2)\\
& \leq \frac{1}{m^2}\sum_{j = 1}^m \Ebb \|B_{j,S}^*B_je_i\|_2^2
\leq \frac{\Lambda}{m^2} \sum_{j = 1}^m \Ebb \|B_je_i\|_2^2 = \frac{\Lambda}{m} \leq \frac{\Theta}{m}.
\end{align*}
we obtain
$$
\Ebb Z \leq \sqrt{\Theta/m} =: \mu.
$$
We are now in a position to apply Theorem~\ref{thm:vecBernsteinExt}. We have
$$
\Pbb \left(\|A_S^* A e_i\|_2 \geq \sqrt{\Theta/m } + t\right)
\leq \exp\left(-\frac{m t^2/2}{\Theta ( 1 + 4  \sqrt{\Theta / m } + 2 t /3 )}\right).
$$
A union bound over $i \in S^c$ concludes the proof.
\end{proof}

The next lemma is a refined version of Lemma \ref{lem:C.3Ext}, involving the quantities $\Lambda(S,F)$ and $\Gamma(F)$ instead of $\Theta(S,F)$. Indeed, it is worth recalling that $\Theta(S,F) \geq \Lambda(S,F)$.
\begin{lemme}
\label{lem:C.5Ext}
Let $S \subseteq \{1,\ldots,n\}$. Then, for every $t>0$
$$
\Pbb \left(\max_{i \in S^c}\|A_S^* A e_i\|_2 \geq \sqrt{\Lambda/m} + t\right)
\leq n \exp\left(-\frac{m t^2/2}{\Gamma + 4 \Lambda \sqrt{\Gamma/m}+ 2 t \sqrt{\Lambda\Gamma}/3}\right).
$$
\end{lemme}
\begin{proof}
The proof relies on the Bernstein-type inequality in Theorem~\ref{thm:vecBernsteinExt}. Fix $i \in S^c$ and consider $\|A^*_S A e_i\|_2$.  Note that $\Ebb A_S^* A e_i = 0$ since $i \in S^c$. Moreover, 
$$
\|A_S^*A e_i\|_2 
= \left\|\sum_{j = 1}^m \underbrace{\frac{1}{m}(B_{j,S}^*B_je_i - \Ebb(B_{j,S}^*B_je_i))}_{=: Y_j}\right\|_2
=  \left\|\sum_{j = 1}^m Y_j\right\|_2=:Z.
$$
Now, we verify the hypotheses of Theorem~\ref{thm:vecBernsteinExt} for the random vectors $Y_1,\ldots,Y_m$. Of course, $\Ebb Y_j = 0$. Moreover, since
$$
\|B_{j,S}^*B_j e_i\|_2 \leq \sqrt{\Lambda} \|B_j e_i\|_2 \leq \sqrt{\Lambda\Gamma},
$$
we obtain
$$
\|Y_j\|_2 
= \frac{1}{m} \|B_{j,S}^*B_je_i - \Ebb(B_{j,S}^*B_je_i)\|_2
\leq 2 \sqrt{\Lambda\Gamma} / m=: K.
$$
Then, using the Cauchy-Schwarz inequality, we estimate
\begin{align*}
\sup_{\|x\|_2 \leq 1}\sum_{j = 1}^m \Ebb |\langle x, Y_j \rangle|^2
& = \frac{1}{m^2}  \sup_{\|x\|_2 \leq 1}\sum_{j = 1}^m \Ebb |\langle x, B_{j,S}^* B_j e_i \rangle|^2
\leq \frac{1}{m^2} \sup_{\|x\|_2 \leq 1}\sum_{j = 1}^m \Ebb [\|B_{j,S}x\|_2^2 \|B_j e_i\|_2^2]\\
& \leq \frac{\Gamma}{m^2} \sup_{\|x\|_2 \leq 1}\sum_{j = 1}^m \Ebb \|B_{j,S}x\|_2^2 
= \frac{\Gamma}{m} \sup_{\|x\|_2 \leq 1} \|x_S\|_2^2 = \frac{\Gamma}{m} =: \sigma^2,
\end{align*}
where in the last step we used $\frac{1}{m} \sum_{j=1}^m \Ebb (B_j^* B_j)= \Id$.
Furthermore, using the independence of the $Y_i$'s, we see that
$$
(\Ebb Z)^2 
\leq \Ebb Z^2 
= \Ebb \left\|\sum_{j = 1}^m Y_j\right\|_2^2
= \Ebb \sum_{j = 1}^m\sum_{k = 1}^m \langle Y_j, Y_k \rangle
= \sum_{j = 1}^m \Ebb\|Y_j\|^2
+
\underbrace{\sum_{j = 1}^m\sum_{k \neq j} \langle \Ebb Y_j, \Ebb Y_k \rangle}_{=0}.
$$
Now, combining the above inequality with the following estimate
\begin{align*}
\sum_{j = 1}^m \Ebb \|Y_j\|_2^2
& = \frac{1}{m^2} \sum_{j = 1}^m \Ebb\|B_{j,S}^*B_je_i - \Ebb(B_{j,S}^*B_je_i)\|_2^2
= \frac{1}{m^2} \sum_{j = 1}^m (\Ebb \|B_{j,S}^*B_je_i\|_2^2 - \|\Ebb (B_{j,S}^*B_je_i)\|_2^2)\\
& \leq \frac{1}{m^2}\sum_{j = 1}^m \Ebb \|B_{j,S}^*B_je_i\|_2^2
\leq \frac{\Lambda}{m^2} \sum_{j = 1}^m \Ebb \|B_je_i\|_2^2 = \frac{\Lambda}{m},
\end{align*}
we obtain
$$
\Ebb Z \leq \sqrt{\Lambda/m} =: \mu.
$$
We are now in a position to apply Theorem~\ref{thm:vecBernsteinExt}. We have
$$
\Pbb \left(\|A_S^* A e_i\|_2 \geq \sqrt{\Lambda/m } + t\right)
\leq \exp\left(-\frac{m t^2/2}{\Gamma + 4 \Lambda \sqrt{\Gamma/m} + 2 t \sqrt{\Lambda\Gamma}/3}\right).
$$
A union bound over $i \in S^c$ concludes the proof.
\end{proof}


\section{Proofs of Section \ref{sec:optimal_drawing}}
\label{app:proof_proba}

Propositions in Section \ref{subsec:opt_proba} are direct consequences of the minimization of the bounds obtained on the number of measurements in Theorems of Section \ref{sec:main_results}, given the following lemma.

\begin{lemme}
\label{lem:max_over_simplex}
Given $(\gamma_k)_{1\leq k \leq p}$ such that for all $1\leq k \leq p, \gamma_k\geq 0$. Define for $\pi$ a probability distribution on $\{1, \hdots , p \}$, the following function:
$$
K( \pi) := \max_{1\leq k \leq p} \frac{\gamma_k}{\pi_k}. 
$$
Then, it holds that
$$
\min_{\pi_k \geq 0 \atop \sum_k \pi_k =1} K(\pi) = \sum_{j=1}^p \gamma_j,
$$
and the unique minimizer $\pi^\star$ of $K$ over discrete probability distributions is for all $1 \leq k \leq p$
$$
\pi_k^\star = \frac{\gamma_k}{\sum_{j=1}^p \gamma_j}.
$$
\end{lemme}

\begin{proof}
The proof is inspired by an intermediary result in \cite{chauffert2013variable}. We give the proof here for completeness. Note that for $\pi^\star$, $K(\pi^\star) = \sum_{j=1}^p \gamma_j$. For a probability distribution $q \neq \pi^\star$, there exists an index $j_0$ such that $q_{j_0} < \pi^\star_{j_0}$ since both sum to 1. Then, $K(q) \geq \gamma_{j_0} / q_{j_0} > \gamma_{j_0} / \pi^\star_{j_0} = \sum_{j=1}^p \gamma_j = K(\pi^\star)$.
\end{proof}

Therefore, the proofs of Propositions \ref{prop_opt_drawing_proba}, \ref{prop_opt_proba_iso}, \ref{prop_opt_drawing_proba_lambda}, \ref{prop_opt_drawing_proba_iso_lambda} easily follow Lemma \ref{lem:max_over_simplex} combined respectively with Theorems \ref{thm:noisy}, \ref{thm:killthetalog}.


\section{Proofs of Section~\ref{sec:appli}}
\label{proof:appli}

The appendix contains the proofs of the results stated in Section~\ref{sec:appli}, regarding the application of the proposed analysis to the case of one-dimensional and two-dimensional Fourier-Haar setting.

\subsection{Proof of Corollary \ref{corol:mri_iso_lambda}}
\label{proof:mri_setup_iso}

Let $A_0 = \mathcal{F}H^* = \replacemath{(a_1 | a_2 | \hdots |a_n )^*}{(d_1 | d_2 | \hdots |d_n )^*} \in \Cbb^{n \times n}$ with $n = 2^{J+1}$ be the Fourier-Haar transform and consider wavelet and frequency subbands $(\Omega_j)_{0 \leq j \leq J}$ and $(W_j)_{0 \leq j \leq J}$ as illustrated in Section~\ref{sec:applications_mri_iso}. 

Recalling the local coherence estimate \eqref{eq:localcoherenceFH} from \cite[Lemma 1]{adcock2016note}, one has for $k \in \{1,\hdots , n \}$
\begin{align*}
\| \replacemath{a}{d}_{k,{S}} \|_2^2 = \sum_{j=0}^J \| P_S P_{\Omega_{j}} \replacemath{a}{d}_{k} \|_2^2 \leq \sum_{j=0}^J {s}_j \|  P_{\Omega_{j}} \replacemath{a}{d}_{k} \|_\infty^2 \leq C \cdot 2^{-j(k)} \sum_{j=0}^J {s}_j  2^{-|j(k)-j|},
\end{align*}
where $j(k)$ is the frequency level corresponding to index $k$ and where $C>0$ is a universal constant. Consequently, an estimate for $\Lambda({S},\pi)$ is
$$
\Lambda({S},\pi) := \max_{1\leq k \leq n} \frac{C \cdot 2^{-j(k)} \sum_{j=0}^J  {s}_j  2^{-|j(k)-j|}}{ \pi_k }.
$$
Recalling Lemma~\ref{lem:max_over_simplex}, the probability \replace{}{distribution} $\pi^\Lambda$ minimizing the previous bound is such that for all $k \in \{1 ,\hdots , n \}$
\begin{align*}
\pi_k^\Lambda = \frac{2^{-j(k)} \sum_{j=0}^J  {s}_j  2^{-|j(k)-j|}}{ \sum_{p=1}^n 2^{-j(p)}   \sum_{j=0}^J  {s}_j  2^{-|j(p)-j|} }.
\end{align*}
Since the size of the frequency subbands is $|W_j| = 2^{\max(j,1)}$, one can rewrite, 
for all $k \in \{1 ,\hdots , n \}$
\begin{align*}
\pi_k^\Lambda &= \frac{2^{-j(k)} \sum_{j=0}^J  {s}_j  2^{-|j(k)-j|}}{ \sum_{j'=0}^J  2^{\max(j',1)-j'}  \sum_{j=0}^J  {s}_j   2^{-|j'-j|} } .
\end{align*}
Therefore, since $2^{j} \leq 2^{\max(j,1)} \leq 2^{j+1}$, we have
\begin{align*}
\Lambda({S},\pi^\Lambda) 
= \sum_{j'= 0}^J 2^{\max(j',1)-j'}\sum_{j = 0}^J s_j 2^{-|j'-j|}
= D \cdot
\sum_{j=0}^J  \left( s_j + \sum_{\substack{j'=0\\ {j' \neq j}}}^J  {s}_{j'}  2^{-|j-j'|} \right),
\end{align*}
with $1 \leq D \leq 2$.
With such a choice of probability \replace{}{distribution} $\pi^\Lambda$, an estimate for $\Gamma$ can be derived as follows:
\begin{align*}
\Gamma(\pi^\Lambda) 
&:= \max_{1\leq k \leq n} \frac{{C\, \cdot \,}2^{-j(k)}}{\pi^{\replacemath{}{\Lambda}}_k} 
= \max_{1\leq k \leq n}  \frac{{C \,\cdot \,}2^{-j(k)} \Lambda(S,\pi^\Lambda) }{ 2^{-j(k)} \sum_{j'=0}^J  {s}_{j'}  2^{-|j(k)-j'|}   } 
= \frac{{C \,\cdot \,}\Lambda(S,\pi^\Lambda)}{\displaystyle\min_{1\leq k \leq n} \sum_{j'=0}^J  {s}_{j'}  2^{-|j(k)-j'|} } ,
\end{align*}
where we used that $\|\replacemath{a}{d}_k\|_\infty^2 \leq C \cdot 2^{-j(k)}$ thanks to \eqref{eq:localcoherenceFH}. Therefore Condition \eqref{cond:killthethetalog} is satisfied if
$$
\min_{1\leq k \leq n} s_{j(k)}  + \sum_{j'\neq j(k)}  {s}_{j'}  2^{-|j(k)-j'|} \gtrsim \ln(3n/\varepsilon),
$$ 
which, in turn, is equivalent to
$$
\min_{0\leq j \leq J} s_{j}  + \sum_{j'\neq j}  {s}_{j'}  2^{-|j-j'|} \gtrsim \ln(3n/\varepsilon).
$$
Theorem \ref{thm:killthetalog} can be now applied, ensuring stable and robust recovery with probability $1-\varepsilon$ with a required number of measurements 

\begin{align*}
m& \gtrsim   \sum_{j=0}^J  \left( s_j + \sum_{j'=0\atop j' \neq j}^J  {s}_j'  2^{-|j-j'|} \right) \cdot  \ln({3}n/\varepsilon).
\end{align*}
This concludes the proof of the corollary.

\subsection{Proof of Corollary \ref{corol:linesHaarMRI}}
\label{proof:corollinesHaarMRI}

The proof of Corollary~\ref{corol:linesHaarMRI} relies on the estimate of $\Theta(S,\pi)$ and on the choice of the corresponding optimal measure $\pi^\Theta$ on the space of vertical lines. With a slight abuse of notation, here $S$ can be interpreted as a subset of $\{1,\ldots,\sqrt{n}\}^2$ or as a subset of $\{1,\ldots,n\}$, depending on whether vectorization is used or not. In order to estimate $\Theta(S,\pi)$, we consider the term
$$
\|D_k^* D_{k,S}\|_{\infty \to \infty}
= 
\max_{1 \leq i \leq n} \|e_i^* D_k^* D_{k,S}\|_1.
$$
Recalling that $D_{k} = \phi_{k,:} \otimes \phi$ and that $\phi$ is an isometry, we obtain
\begin{align*}
D_k^*D_{k,S} 
= (\phi_{k,:}\otimes \phi)^*(\phi_{k,:}\otimes \phi)P_S^*
= ((\phi_{k,:})^* \phi_{k,:} \otimes \phi^* \phi)P_S^*
= ((\phi_{k,:})^* \phi_{k,:} \otimes I)P_S^*.
\end{align*}
Let us explicitly write
$$
(\phi_{k,:})^* \phi_{k,:} \otimes I
=
\begin{pmatrix}
\overline{\phi_{k,1}} \phi_{k,1} I & \cdots & \overline{\phi_{k,1}} \phi_{k,\sqrt{n}} I\\
\vdots & \ddots & \vdots \\
\overline{\phi_{k,\sqrt{n}}} \phi_{k,1} I & \cdots & \overline{\phi_{k,\sqrt{n}}} \phi_{k,\sqrt{n}} I
\end{pmatrix}.
$$
Then, we have
\begin{align*}
\max_{1 \leq i \leq n} \|e_i^* D_k^* D_{k,S}\|_1
& = \max_{1 \leq q \leq \sqrt{n}} \max_{1 \leq t \leq \sqrt{n}} \|e_t^* \begin{pmatrix}\overline{\phi_{k,q}} \phi_{k,1} I | \cdots | \overline{\phi_{k,q}} \phi_{k,\sqrt{n}} I\end{pmatrix}P_S^*\|_1\\
& = \max_{1 \leq q \leq \sqrt{n}} |\phi_{k,q}|\max_{1 \leq t \leq \sqrt{n}} \|e_t^* \begin{pmatrix} \phi_{k,1} I | \cdots |  \phi_{k,\sqrt{n}} I\end{pmatrix}P_S^*\|_1.
\end{align*}
The local coherence bound \eqref{eq:localcoherenceFH} yields $\displaystyle\max_{1 \leq q \leq \sqrt{n}} |\phi_{k,q}| \lesssim 2^{-j(k)/2}$. We estimate
\begin{align*}
\max_{1 \leq t \leq \sqrt{n}} \|e_t^* \begin{pmatrix} \phi_{k,1} I | \cdots |  \phi_{k,\sqrt{n}} I\end{pmatrix}P_S^*\|_1
& = \max_{1 \leq t \leq \sqrt{n}}  \left\|\begin{pmatrix} \phi_{k,1} (P_{S \cap C_1})_{t,:} | \cdots |  \phi_{k,\sqrt{n}} (P_{S \cap C_{\sqrt{n}}})_{t,:}\end{pmatrix}\right\|_1\\
& = \max_{1 \leq t \leq \sqrt{n}}  \sum_{q = 1}^{\sqrt{n}}|\phi_{k,q}|\left\|(P_{S \cap C_q})_{t,:} \right\|_1,
\end{align*}
where $C_q = \{1,\ldots,\sqrt{n}\}\times \{q\}$ is the $q$-th vertical line of $\{1,\ldots,n\}^2$. Now, we observe that 
$$
\left\|(P_{S \cap C_q})_{t,:} \right\|_1
= \mathbbm{1}_{S \cap C_q}(t) := 
\begin{cases}
1 & \text{if } t \in S \cap C_q\\
0 & \text{otherwise}.
\end{cases}
$$
Finally, using the above relation, the decomposition of $\{1,\ldots,n\}^2$ into subbands $(\Omega_{\ell,j})_{0\leq \ell,j \leq J}$, and resorting again to \eqref{eq:localcoherenceFH}, we see that
\begin{align*}
\max_{1 \leq t \leq \sqrt{n}}  \sum_{q = 1}^{\sqrt{n}}|\phi_{k,q}|\left\|(P_{S \cap C_q})_{t,:} \right\|_1
& = \max_{1 \leq t \leq \sqrt{n}}  \sum_{q = 1}^{\sqrt{n}}|\phi_{k,q}|\mathbbm{1}_{S \cap C_q}(t)
 = \max_{0 \leq \ell \leq J} \max_{t \in \Omega_\ell}
\sum_{j = 0}^J \sum_{q \in \Omega_j} |\phi_{k,q}|\mathbbm{1}_{S \cap C_q}(t)\\
& \leq \max_{0 \leq \ell \leq J} \max_{t \in \Omega_\ell}
\sum_{j = 0}^J \max_{q \in \Omega_j} |\phi_{k,q}| \sum_{q \in \Omega_j} \mathbbm{1}_{S \cap C_q}(t)\\
& = \max_{0 \leq \ell \leq J} \max_{t \in \Omega_\ell}
\sum_{j = 0}^J \max_{q \in \Omega_j} |\phi_{k,q}| |S \cap \Omega_{\ell,j} \cap R_t|\\
& \lesssim \sum_{j = 0}^J 2^{-j(k)/2} 2^{-|j(k)-j|/2} s_j^r.
\end{align*}
This leads to the desired estimate of $\Theta(S,\pi)$. Using Lemma~\ref{lem:max_over_simplex} we can then derive the corresponding optimal measure $\pi^\Theta$ and conclude the proof.

\subsection{Proof of Corollary \ref{corol:linesHaarMRI_lambda}}
\label{proof:corollinesHaarMRI_lambda}

{We divide the proof of Corollary~\ref{corol:linesHaarMRI_lambda} into two steps. First, we find the optimal sampling measure $\pi^\Lambda$ that minimizes $\Lambda(S,\pi)$ and compute the corresponding $\Lambda(S,\pi^\Lambda)$. In the second step, we estimate $\Gamma(\pi^\Lambda)$ and derive the extra condition on the structured sparsity of the signal.

\paragraph{Step 1: Estimate of $\Lambda$ and derivation of $\pi^\Lambda$}
 
As in the previous section, here $S$ can be interpreted as a subset of $\{1,\ldots,\sqrt{n}\}^2$ or as a subset of $\{1,\ldots,n\}$, depending on whether vectorization is used or not.  Let us fix a frequency $k \in W_j$. Our first goal is to estimate the quantity $\|D_{k,S}^* D_{k,S}\|_{2 \to 2}$. This will lead us to the optimal choice of $\pi = \pi^{\Lambda}$. 

First, we observe that $\|D_{k,S}^* D_{k,S}\|_{2 \to 2} = \|D_{k,S}\|_{2 \to 2}^2$. Now, let $v \in \mathbb{C}^n$ and $V \in \Cbb^{\sqrt{n} \times \sqrt{n}}$ be such that $v = \vect (V)$ (recall that the vectorization operator $\vect(\cdot)$ stacks the columns of a matrix on top of each other). Taking into account the structure of $D_{k,S}$ and $v$, and recalling that $\phi$ is an isometry,  we have
\begin{align*}
\|D_{k,S} v \|_{2}^2 
& = \Bigg\|\sum_{i = 1}^{\sqrt{n}}\phi_{k,i} \phi (V_S)_{:,i} \Bigg\|_2^2
= \Bigg\|\phi \sum_{i = 1}^{\sqrt{n}}\phi_{k,i} (V_S)_{:,i} \Bigg\|_2^2
= \Bigg\|\sum_{i = 1}^{\sqrt{n}}\phi_{k,i} (V_S)_{:,i} \Bigg\|_2^2
= \sum_{q = 1}^{\sqrt{n}} \Bigg| \sum_{i = 1}^{\sqrt{n}}\phi_{k,i} (V_S)_{q,i}\Bigg|^2.
\end{align*}
Taking advantage of the decomposition of the wavelet multi-index space into tensor product subbands $(\Omega_{\ell,j})_{0 \leq \ell, j \leq J}$ and recalling \eqref{eq:localcoherenceFH}, we extend the previous chain of equalities as follows:
\begin{align*}
\|D_{k,S}v\|_2^2 
& = \sum_{\ell = 0}^J \sum_{q \in \Omega_\ell} \Bigg| \sum_{j = 0}^{J} \sum_{i \in \Omega_{j}} \phi_{k,i} (V_S)_{q,i} \Bigg|^2
 \leq \sum_{\ell= 0}^J \sum_{q \in \Omega_\ell} \Bigg( \sum_{j = 0}^{J} \max_{i \in \Omega_{j}}|\phi_{k,i}|\sum_{i \in \Omega_{j}}  |(V_S)_{q,i}| \Bigg)^2\\
& = \sum_{\ell= 0}^J \sum_{q \in \Omega_\ell} \Bigg( \sum_{j = 0}^{J} \sqrt{\mu_{j(k),j}} \|V_{S\cap\Omega_{\ell,j}\cap R_q}\|_1 \Bigg)^2,
\end{align*}
where $j(k)$ is such that $k \in W_{j(k)}$ and $\mu_{\ell,j}$ is the local coherence of the Fourier-Haar transform defined in \eqref{eq:localcoherenceFH}. Using the Cauchy-Schwarz inequality, rearranging the summation, and recalling Definition~\ref{def:sparsity_in_levels_2D}, we see that
\begin{align*}
\|D_{k,S}v\|_2^2 
& \leq
\sum_{\ell= 0}^J \sum_{q \in \Omega_\ell} \Bigg( \sum_{j = 0}^{J} \sqrt{\mu_{j(k),j}} \sqrt{s_{j}^r} \|V_{S\cap\Omega_{\ell,j}\cap R_q}\|_2 \Bigg)^2
 \leq 
J  \sum_{\ell= 0}^J \sum_{q \in \Omega_\ell}  \sum_{j = 0}^{J} \mu_{j(k),j} s_{j}^r \|V_{S\cap\Omega_{\ell,j}\cap R_q}\|_2^2 \\
& = 
J   \Bigg(\sum_{j = 0}^{J} \mu_{j(k),j} s_{j}^r \Bigg)\sum_{j = 0}^{J}\sum_{\ell= 0}^J \sum_{q \in \Omega_\ell}\|V_{S\cap\Omega_{\ell,j}\cap R_q}\|_2^2 
 = 
J \Bigg(\sum_{\ell = 0}^{J} \mu_{j(k),j} s_{j}^r \Bigg) \|v\|_2^2.
\end{align*}
Combining the above relations an using again \eqref{eq:localcoherenceFH} yields
$$
\|D_{k,S}^* D_{k,S}\|_{2 \to 2} 
\leq J \sum_{j = 0}^{J} \mu_{j(k),j} s_{j}^r
\leq C \cdot J 2^{-j(k)} \sum_{j = 0}^{J} 2^{-|j(k)-j|} s_{j}^r,
$$
where $C>0$ is the universal constant in \eqref{eq:localcoherenceFH}. Using Lemma~\ref{lem:max_over_simplex} and recalling that $|W_j| = 2^{\max(1,j)}$, the optimal probability $\pi^\Lambda$ is, for all $k = 1,\ldots,\sqrt{n}$,
$$
\pi^\Lambda_{k}
= \frac{ 2^{-j(k)} \sum_{j = 0}^{J} 2^{-|j(k)-j|} s_{j}^r}{\sum_{t = 0}^{\sqrt{n}} 2^{-j(t)} \sum_{j = 0}^{J} 2^{-|j(t)-j|} s_{j}^r} 
= \frac{ 2^{-j(k)} \sum_{j = 0}^{J} 2^{-|j(k)-j|} s_{j}^r}{\sum_{\ell = 0}^J 2^{\max(1,\ell)-\ell} \sum_{j = 0}^{J} 2^{-|\ell-j|} s_{j}^r}.
$$
This leads to 
\begin{align*}
\max_{0 \leq k \leq \sqrt{n}} \frac{\|D_{k,S}^* D_{k,S}\|_{2 \to 2}}{\pi_k^{\Lambda}}
\leq C J \cdot \sum_{\ell = 0}^J 2^{\max(1,\ell)-\ell} \sum_{j = 0}^{J} 2^{-|j-\ell|} s_{j}^r
=: \Lambda(S,\pi^{\Lambda}).
\end{align*}

\paragraph{Step 2: Estimate of $\Gamma$ and conclusion}

Employing \eqref{eq:localcoherenceFH} and recalling that $\phi$ is an isometry, we have
$$
\|D_k\|_{1 \to 2}^2 
= 
\max_{1 \leq i \leq \sqrt{n}} \max_{1 \leq q \leq \sqrt{n}} \|\phi_{k,i} \phi_{:,q}\|_2^2
= 
\max_{1 \leq i \leq \sqrt{n}} \max_{1 \leq q \leq \sqrt{n}} |\phi_{k,i}|^2 \underbrace{\|\phi_{:,q}\|_2^2}_{=1}
\leq C \cdot 2^{-j(k)}.
$$
Therefore, we estimate
\begin{align*}
\max_{1 \leq k \leq \sqrt{n}} \frac{\|D_k\|_{1 \to 2}}{\pi_k^\Lambda}
& \leq 
\max_{1 \leq k \leq \sqrt{n}}\frac{C \cdot 2^{-j(k)}}{
2^{-j(k)} \sum_{j = 0}^{J} 2^{-|j-j(k)|} s_{j}^r} \Big(\sum_{\ell = 0}^J 2^{\max(1,\ell)-\ell} \sum_{j = 0}^{J} 2^{-|j-\ell|} s_{j}^r\Big)\\
& = \frac{\Lambda(S,\pi^\Lambda)}{
\displaystyle J \cdot \min_{0 \leq \ell \leq J}\sum_{j = 0}^{J} 2^{-|j-\ell|} s_{j}^r}=: \Gamma(\pi^\Lambda).
\end{align*}
As a consequence, the condition \eqref{cond:killthethetalog}  of Theorem~\ref{thm:killthetalog} reads
$$
J \cdot \min_{0 \leq \ell \leq J}\sum_{j = 0}^{J} 2^{-|j-\ell|} s_{j}^r \gtrsim \ln(3n/\varepsilon)
$$
and the minimum sampling complexity is
$$
m \gtrsim \Lambda(S,\pi^{\Lambda})\ln(3n/\varepsilon), 
$$
which, in turn,  is implied by
$$
m\gtrsim \Bigg(\sum_{j = 0}^J s_j^r + \sum_{\substack{\ell = 0\\ \ell \neq j}}^{J} 2^{-|\ell-j|} s_{\ell}^r \Bigg)\log_2(\sqrt{n})\ln(3n / \varepsilon).
$$
This concludes the proof.

\small
\bibliographystyle{alpha}
\bibliography{mybib}


\end{document}